\documentclass{CSML}

\def\dOi{13(1:6)2017}
\lmcsheading%
{\dOi}
{1--56}
{}
{}
{Oct.~16, 2015}
{Feb.~\phantom06, 2016}
{}

\ACMCCS{[{\bf Theory of computation}]:  Logic---Equational logic and rewriting}

\pdfoutput=1

\usepackage{hyperref}
\usepackage{tikz}
\usepackage{amssymb}

\usepackage{xcolor,latexsym,amsmath,extarrows,alltt}
\usepackage{xspace}
\usepackage{booktabs}
\usepackage{mathtools}
\usepackage{stmaryrd}
\usepackage[verbose]{wrapfig}

\usetikzlibrary{arrows}
\usetikzlibrary{decorations.markings}

\newcommand{\m}[1]{\mathsf{#1}}
\newcommand{\CC}{\mathcal{C}}
\newcommand{\DD}{\mathcal{D}}
\newcommand{\FF}{\mathcal{F}}
\newcommand{\GG}{\mathcal{G}}
\newcommand{\HH}{\mathcal{H}}
\newcommand{\OO}{\mathcal{O}}
\newcommand{\PP}{\mathcal{P}}
\newcommand{\RR}{\mathcal{R}}
\newcommand{\TT}{\mathcal{T}}
\newcommand{\TTp}{\TT_\mathrm{p}}
\newcommand{\TTi}{\TT_\mathrm{i}}
\newcommand{\VV}{\mathcal{V}}
\newcommand{\A}{\mathcal{A}}
\newcommand{\N}{\mathbb{N}}

\newcommand{\Cost}{\mathcal{C}\kern-0.1em\mathit{ost}}
\renewcommand{\Cost}{\mathit{cost}}
\newcommand{\Size}{\mathcal{S}\kern-0.1em\mathit{ize}}
\renewcommand{\Size}{\mathit{size}}
\newcommand{\Extra}{\mathcal{E}\kern-0.1em\mathit{xtra}}
\newcommand{\IC}{\mathcal{C}}
\newcommand{\IS}{\mathcal{S}}

\newcommand{\Var}{\VV\m{ar}}
\newcommand{\Pos}{\PP\m{os}}
\newcommand{\seq}[2][n]{#2_1,\dots,#2_{#1}}
\newcommand{\nulseq}[2][k]{#2_0,\dots,#2_{#1}}
\newcommand{\instance}{\mathrel{\makebox[0pt]{\makebox[8pt][r]%
{\raise 1pt \hbox{$\cdot$}}}{\geq}}}
\newcommand{\dom}{\mathsf{dom}}
\newcommand{\I}{\mathcal{I}}
\newcommand{\J}{\mathcal{J}}
\newcommand{\nurc}{\NU_{\m{rc}}}
\newcommand{\interpret}[2]{\llbracket #1 \rrbracket^{\I}_{#2}}
\renewcommand{\interpret}[2]{[#1]_{\I}^{#2}}
\newcommand{\rhrhd}{\xrightharpoonup{\raisebox{-1mm}{$\scriptstyle
\smash{\rhd}$\,}}}
\newcommand{\quasistep}{\xrightarrow{\raisebox{-1mm}{$\scriptstyle
\smash{\rhd}$}}}

\def\systemname#1{\mbox{\textsf{#1}}\xspace}
\newcommand{\TCT}{%
 \systemname{T\kern-0.2em\raisebox{-0.2em}C\kern-0.2emT\kern-0.2em}%
}

\theoremstyle{definition}\newtheorem*{notation}{Notation}

\newcounter{recipe}
\renewcommand{\therecipe}{\Alph{recipe}}

\newcommand{\bc}{\begin{quote}\begin{alltt}}
\newcommand{\ec}{\end{alltt}\end{quote}}
\newcommand{\NU}{\upsilon}

\begin{document}

\title[Complexity of Conditional Term Rewriting]{Complexity of Conditional Term Rewriting\rsuper*}
\thanks{This research was supported by the Austrian Science Fund (FWF)
project I963 and partially supported by the Marie Sk{\l}odowska-Curie
action ``HORIP'', program H2020-MSCA-IF-2014, 658162.}

\thanks{{\lsuper a}Cynthia Kop's former affiliation: Department of Computer Science,
University of Insbruck, Austria.}
\author[C.~Kop]{Cynthia Kop\rsuper a}
\address{{\lsuper a}Department of Computer Science,
University of Copenhagen, Denmark}
\email{kop@di.ku.dk}
\author[A.~Middeldorp]{Aart Middeldorp\rsuper b}
\address{{\lsuper{b,c}}Department of Computer Science \newline
University of Innsbruck, Austria}
\email{\{aart.middeldorp, thomas.sternagel\}@uibk.ac.at}
\author[T.~Sternagel]{Thomas Sternagel\rsuper c}
\address{\vspace{-18 pt}}

\titlecomment{{\lsuper*}This article is an extended version of~\cite{KMS15},
with a more elegant transformation including a completeness proof in
Section~\ref{sec:transformation} and
a drastic extension of the interpretation-based methods in
Sections~\ref{sec:polynomial}--\ref{sec:splitsize}.}

\begin{abstract}
We propose a notion of complexity for oriented conditional
rewrite systems satisfying certain restrictions.
This notion is realistic in the sense that it measures not only successful
computations, but also partial computations that result in a
failed rule application. A transformation to unconditional
context-sensitive rewrite systems is presented which reflects this
complexity notion, as well as a technique
to derive runtime and derivational complexity bounds for the result
of this transformation.
\end{abstract}

\maketitle

\section{Introduction}

Conditional term rewriting~\cite[Chapter~7]{O02} is a well-known
computational paradigm. First studied in the eighties and early nineties
of the previous century, in more recent years transformation techniques
have received a lot of attention. Various automatic tools for
(operational) termination \cite{GBEFFOPSST14,LM14,SG09}
as well as confluence \cite{SM14} have been developed.

In this paper we consider the following question: What is the
greatest number of steps that can be done when evaluating terms,
for starting terms of a given size? For unconditional
rewrite systems this question has been investigated extensively
and numerous techniques have been developed that give
an upper bound on the resulting notions of derivational and
runtime complexity (e.g.\ \cite{BCMT01,HM14,HL89,MMNWZ11,MS11}).
Tools that support complexity methods (\cite{MM13,NEG13,ZK14})
are under active development and compete annually in the complexity
competition.\footnote{\url{http://cbr.uibk.ac.at/competition/}}

We are not aware of any techniques or tools for conditional (derivational
and runtime) complexity---or indeed, even of a \emph{definition} for
conditional complexity. This may be for a good reason, as it is not
obvious what such a definition should be.
Of course, simply counting steps without taking the conditions
into account will not do.
Counting successful rewrite steps both
in the reduction and in the evaluation of conditions
is a natural idea. This two-dimensional view is seen for instance
in studies of (operational) termination~\cite{LMM05,LM14} and certain
transformations from conditional rewrite systems to unconditional ones
(e.g., unravelings~\cite{M96,O02}).
However, we
will argue that this approach---considering only the \emph{successful}
evaluation steps---still gives rise to an
unrealistic notion of complexity. Modern rewrite engines like
Maude~\cite{Maude} that support conditional rewriting can spend
significant resources on evaluating conditions that in the end prove to
be useless for rewriting the term at hand.
This should be taken into account when defining complexity.

\subparagraph*{\textbf{Contribution.}}
We propose a new notion of conditional complexity for a relatively large
class of reasonably well-behaved conditional rewrite systems. This notion
aims to capture the maximal number of rewrite steps that can be performed
when reducing a term to normal form, including the steps that were
computed but
turned out to be
ultimately not useful.
In order to reuse existing methodology for deriving complexity bounds,
we present a transformation into unconditional rewrite systems
that can be used to estimate the conditional complexity, building
on the ideas of \emph{structure-preserving}
transformations~\cite{ABH03,SR06} but including several new ideas. The
transformed system is context-sensitive
(Lucas~\cite{L98,L02}), which is not yet supported by current complexity
tools; however, ignoring the corresponding restrictions,
we still obtain an upper bound on the conditional complexity.

\subparagraph*{\textbf{Organization.}}
The remainder of the paper is organized as follows. In the next section
we recall some preliminaries.
Based on the analysis of conditional complexity in
Section~\ref{sec:analysis}, we introduce our new notion
formally in Section~\ref{sec:complexity}.
Section~\ref{sec:transformation} presents a transformation to
context-sensitive rewrite systems, and in
Section~\ref{sec:polynomial}--\ref{sec:splitsize} we present an
interpretation-based method targeting the resulting systems, as well
as two optimizations of the technique to demonstrate that we can
obtain tight bounds on realistic systems.
Section~\ref{sec:conclusion} concludes with initial experiments,
related work, and suggestions for future work.

\section{Preliminaries}
\label{sec:preliminaries}

We assume familiarity with (conditional) term rewriting and all that
(e.g., \cite{BN98,O02,TeReSe})
and only shortly recall important notions that are used in the following.

In this paper we consider \emph{oriented} conditional (term) rewrite
systems (CTRSs for short).
\emph{Conditional rewrite rules} have the form
$\ell \to r \Leftarrow c$, where $c$ is a sequence
$a_1 \approx b_1, \dots, a_k \approx b_k$ of equations.
An oriented CTRS is a set $\RR$ of conditional rules.
The rewrite relation $\to_\RR$ associated with $\RR$
is formally defined as the union of a series of
approximations $\to_{\RR_i}$, where
\begin{itemize}
\item
$\RR_0 = \varnothing$,
\smallskip
\item
$\RR_{i+1} = \{ \ell\sigma \to r\sigma \mid
\text{$\ell \to r \Leftarrow c \in \RR$ and
$a\sigma \to^*_{\RR_i} b\sigma$ for all $a \approx b \in c$} \}$.
\end{itemize}
In the sequel we will primarily use the observation that
$s \to_\RR t$ if and only if there exist a position $p$ in $s$, a
rule $\ell \to r \Leftarrow c$ in $\RR$, and a substitution $\sigma$
such that $s|_p = \ell\sigma$, $t = s[r\sigma]_p$, and
$\RR \vdash c\sigma$, where the latter
denotes that $a\sigma \to^*_\RR b\sigma$ for all $a \approx b \in c$.
We may write $s \xrightarrow{\epsilon} t$ for a rewrite
step at the root position and $s \xrightarrow{\smash{> \epsilon}} t$
for a non-root step.

Given a (C)TRS $\RR$ over a signature $\FF$, the root symbols of
left-hand sides of rules in $\RR$ are called \emph{defined symbols}
and every other symbol in $\FF$ is a \emph{constructor symbol}.
These sets are denoted by $\FF_\mathcal{D}$ and $\FF_\mathcal{C}$,
respectively.
For a given symbol $f$, we write $\RR{\restriction}f$ for the
set of rules in $\RR$ whose left-hand sides have root symbol $f$.
A \emph{constructor term} consists of constructor symbols and
variables. A \emph{basic term} is a term $f(\seq{t})$ where
$f \in \FF_\mathcal{D}$ and $\seq{t}$ are constructor terms.
We call $\RR$ \emph{semi-finite} if $\RR{\restriction}f$ is finite
for every $f \in \FF_\mathcal{D}$.
Let $\Sigma(\FF,\VV)$ be the set of substitutions mapping to
$\TT(\FF,\VV)$. For substitutions $\sigma$ and $\tau$ we write
$\sigma \to_\RR^* \tau$ to denote $\sigma(x) \to_\RR^* \tau(x)$
for all variables $x \in \VV$.
A term $s$ is \emph{terminating} if there is no infinite reduction
$s \to_\RR s_1 \to_\RR s_2 \to_\RR \cdots$. A \emph{normal form} 
is a term $s$ such that there is no term $t$ with $s \to_\RR t$.
We say that
$t$ is a normal form of $s$ if $s \to_\RR^* t$ and $t$ is a normal
form. Note that it is possible for a normal form to instantiate the
left-hand side of a rule, which is not true for TRSs.

A (C)TRS is \emph{finitely branching} if there are only
finitely many distinct terms reachable in one rewrite step from any
given term. All semi-finite (C)TRSs are finitely branching, but
they may have an infinite signature.
Given a terminating and finitely branching TRS $\RR$ over a signature
$\FF$, the \emph{derivation height} of a term $t$ is defined as
$\m{dh}(t) = \max\,\{ n \mid \text{$t \to^n u$ for some term $u$} \}$.
This leads to the notion of \emph{derivational complexity}
$\m{dc}_\RR(n) = \max\,\{ \m{dh}(t) \mid |t| \leqslant n \}$,
where $|t|$ is the number of symbols occurring in $t$.
If we restrict the definition to basic terms $t$ we get the notion of
\emph{runtime complexity} $\m{rc}_\RR(n)$~\cite{HM08}.

Rewrite rules $\ell \to r \Leftarrow c$ of CTRSs are classified according
to the distribution of variables among $\ell$, $r$, and $c$.
In this paper we consider 3-CTRSs, where the rules satisfy
$\Var(r) \subseteq \Var(\ell,c)$. A CTRS $\RR$ is
\emph{deterministic} if for every rule
$\ell \to r \Leftarrow a_1 \approx b_1, \dots, a_k \approx b_k$ in $\RR$
we have $\Var(a_i) \subseteq \Var(\ell,\seq[i-1]{b})$ for
$1 \leqslant i \leqslant k$.

We write $s \quasistep t$ if there exist a
position $p$ in $s$, a rule
$\ell \to r \Leftarrow a_1 \approx b_1, \dots, a_k \approx b_k$, a
substitution $\sigma$, and an index $1 \leqslant i \leqslant k$ such
that $s|_p = \ell\sigma$, $a_j\sigma \to^* b_j\sigma$ for all
$1 \leqslant j < i$, and $t = a_i\sigma$.
A CTRS is \emph{quasi-decreasing} if there exists a well-founded
order $>$ with the subterm property (i.e., ${\rhd} \subseteq {>}$
where $s \rhd t$ if $t$ is a proper subterm of $s$) such that
both $\to$ and $\quasistep$ are included in $>$~\cite{DO90}.
We additionally define here that a term $s$ is quasi-decreasing if there
is no infinite sequence $s = u_0 \mathrel{(\to \cup \quasistep)} u_1
\mathrel{(\to \cup \quasistep)} \cdots$.
Clearly, a CTRS is quasi-decreasing if and only if all its terms are,
but individual terms may be quasi-decreasing even if the CTRS is not.
Quasi-decreasingness ensures termination and, for finite CTRSs,
computability of the rewrite relation.
Quasi-decreasingness coincides with
\emph{operational termination}~\cite{LMM05}.
We call a CTRS \emph{constructor-based} if the right-hand sides of
conditions as well as proper subterms of the left-hand sides of
rules are constructor terms.

\subparagraph*{\textbf{Limitations.}}
We restrict ourselves to constructor-based deterministic
3-CTRSs, where the right-hand sides of conditions
use only variables not occurring in the left-hand side
or in earlier conditions.
That is, for every rule $f(\seq{\ell}) \to r \Leftarrow
a_1 \approx b_1, \dots, a_k \approx b_k \in \RR$:
\begin{itemize}
\item
$\seq{\ell},\seq[k]{b}$ are
constructor terms without common variables,
\item
$\Var(r) \subseteq \Var(\seq{\ell},\seq[k]{b})$ and
$\Var(a_i) \subseteq \Var(\seq{\ell},\seq[i-1]{b})$ for
$1 \leqslant i \leqslant k$.
\end{itemize}
We will call such systems \emph{CCTRSs}. 
Furthermore, we will focus on \emph{strong} CCTRSs:
semi-finite CCTRSs such that, for every rule
$f(\seq{\ell}) \to r \Leftarrow a_1 \approx b_1, \dots, a_k
\approx b_k \in \RR$,
\begin{itemize}
\item
$f(\seq{\ell})$ and $\seq[k]{b}$ are linear terms: no variable
occurs more than once in them.
\end{itemize}
Note that, even in strong CCTRSs,
the left-hand sides of conditions are \emph{not} required to be linear.
We will develop a complexity notion for the general case of
CCTRSs, but limit the work on transformations (in
Section~\ref{sec:transformation} and beyond) to strong CCTRSs.
We will particularly consider \emph{confluent} CCTRSs.
While confluence is not needed for the formal development in this
paper, without it the complexity notion we define is not meaningful,
as discussed below.

\smallskip

To appreciate the limitations, note that in CTRSs which are not
deterministic 3-CTRSs, the rewrite relation is undecidable in
general, which makes it hard to define what complexity means.
The restrictions with regards to variables and constructors in
strong CCTRSs are the natural extension of the common restriction to
left-linear constructor TRSs in unconditional rewriting.
They closely correspond to \emph{pattern guards} \cite{EPJ01},
a language extension of Haskell.
Semi-finiteness actually \emph{weakens} the standard restriction
that $\RR$ must be finite.

The limitation to CCTRSs is important because, in confluent CCTRSs,
the approach to computation is unambiguous:
To evaluate whether a term $\ell\sigma$ reduces with
a rule $\ell \to r \Leftarrow a_1 \approx b_1, \dots, a_k \approx b_k$
of a CCTRS,
we start by reducing $a_1\sigma$ and, finding an instance of $b_1$,
extend $\sigma$ to the new variables in $b_1$ resulting in $\sigma'$,
continue with $a_2\sigma'$, and so on.
Assuming confluence, if there is an extension of $\sigma$
which satisfies all conditions then, no matter how we
reduce, this procedure will either find it or---if
$\ell\sigma$ is not quasi-decreasing---enter into an infinite reduction,
a possibility which is also interesting from a complexity standpoint.
However, if confluence or any of the restrictions on the conditions
were dropped, this would no longer be the case and we might be unable
to verify the applicability of a rule without enumerating all
possible reducts of its conditions.
The restrictions are needed to obtain Lemma~\ref{lem:local}, which
will be essential to justify the way we handle failure.

We do not limit interest to quasi-decreasing CCTRSs---which
would correspond to the usual approach of limiting interest to
terminating TRSs in the unconditional setting---but will rather
define the complexity of non-quasi-decreasing terms to be
infinite. This is done in order to unify proof efforts,
especially for Theorem~\ref{thm:transformcomplete}.

\begin{exa}
\label{ex:fibonacci}
The CTRS $\RR_{\m{fib}}$ consisting of the rewrite rules \\
\begin{minipage}{.3\textwidth}
\begin{align}
\m{0} + y &\to y \label{fibplus1} \\
\m{s}(x) + y &\to \m{s}(x + y) \label{fibplus2}
\end{align}
\end{minipage}
\begin{minipage}{.69\textwidth}
\begin{align}
\m{fib}(\m{0}) &\to \langle \m{0},\m{s}(\m{0}) \rangle \label{fib3} \\
\m{fib}(\m{s}(x)) &\to \langle z, w \rangle
~\Leftarrow~ \m{fib}(x) \approx \langle y, z \rangle,~ y + z \approx w
\label{fib4}
\end{align}
\end{minipage}

\medskip

\noindent
is a quasi-decreasing and confluent strong CCTRS. The requirements
for quasi-decreasingness are satisfied (e.g.) by the lexicographic path
order with precedence
$\m{fib} > {\langle \cdot,\cdot \rangle} > \m{+} > \m{s}$.
Because the 3-CTRS $\RR_{\m{fib}}$ is orthogonal, right-stable, and
properly oriented, confluence follows from the result of \cite{SMI95}.
\end{exa}

\begin{notation}
To simplify the notation and shorten proofs, we will use the
following convention throughout the paper.
Given a rule $\rho\colon \ell \to r \Leftarrow c$,
\begin{itemize}
\item
the conditional part $c$ consists of the conditions
$a_1 \approx b_1$, $\dots$, $a_k \approx b_k$ for some
$k \geqslant 0$ (which depends on $\rho$),
\smallskip
\item
for all $0 \leqslant j \leqslant k$,
$c_j^{\leqslant}$ denotes the
sequence $a_1 \approx b_1$, $\dots$, $a_j \approx b_j$.
\end{itemize}
In addition, we will sometimes refer to $\ell$ as $b_0$ and to
$r$ as $a_{k+1}$.
\end{notation}

With these conventions, the limitations on rules can be
reformulated as follows. For every rule $\ell \to r \Leftarrow c$:
\begin{itemize}
\item
$\seq[k]{b}$ and the proper subterms of $b_0$ are constructor terms,
\smallskip
\item
$\Var(b_i) \cap \Var(b_j) = \varnothing$
for all $0 \leqslant i, j \leqslant k$ with $i \neq j$
and, in a strong CCTRS, the terms $\nulseq[k]{b}$ are linear,
\smallskip
\item
$\Var(a_i) \subseteq \Var(\nulseq[i-1]{b})$
for all $1 \leqslant i \leqslant k+1$.
\end{itemize}\newpage

\section{Analysis}
\label{sec:analysis}

Before we can define a notion of complexity, we must consider a
model of computation. Unlike unconditional term rewriting, it is not
obvious how a term in a CTRS is reduced to normal form. Even taking
the approach for confluent CCTRSs sketched in
Section~\ref{sec:preliminaries} as a basis, some unresolved
questions remain.
In this section, we will study both computation and complexity by
an appeal to intuition. In the next section we will formalize the
results.

We start our analysis with a deceivingly simple CCTRS to illustrate
that the notion of complexity for conditional systems is not obvious.

\begin{exa}
\label{ex:evenodd}
The CCTRS $\RR_{\m{even}}$ consists of the following six rewrite rules:
\\
\hspace*{-2mm}
\begin{minipage}{.5\textwidth}
\begin{align}
\m{even}(\m{0}) &\to \m{true} \label{eo1}
\\
\m{even}(\m{s}(x)) &\to \makebox[7mm]{$\m{true}$} \,\Leftarrow\,
\m{odd}(x) \approx \m{true} \label{eo2}
\\
\m{even}(\m{s}(x)) &\to \makebox[7mm]{$\m{false}$} \,\Leftarrow\,
\m{even}(x) \approx \m{true} \label{eo3}
\end{align}
\end{minipage}
\begin{minipage}{.5\textwidth}
\begin{align}
\m{odd}(\m{0}) &\to \m{false} \label{eo4}
\\
\m{odd}(\m{s}(x)) &\to \makebox[7mm]{$\m{true}$} \,\Leftarrow\,
\m{even}(x) \approx \m{true} \label{eo5}
\\
\m{odd}(\m{s}(x)) &\to \makebox[7mm]{$\m{false}$} \,\Leftarrow\,
\m{odd}(x) \approx \m{true} \label{eo6}
\end{align}
\end{minipage}
\end{exa}

If, like in the unconditional case, we count the number of steps
needed to normalize a term, then a term
$t_n = \m{even}(\m{s}^n(\m{0}))$ has derivation height $1$, since
$t_n \to \m{false}$ or $t_n \to \m{true}$ in a single step.
To reflect actual computation, the
rewrite steps to verify the condition should be taken into account.
Viewed like this, normalizing $t_n$ takes $n+1$ rewrite steps.

However, this still seems unrealistic
as a rewriting engine cannot know in advance which
rule to attempt first. For example, when rewriting $t_9$,
rule \eqref{eo2} may be tried first, which requires normalizing
$\m{odd}(\m{s}^8(\m{0}))$ to verify the condition.
After finding that the condition fails, rule \eqref{eo3} is attempted.
Thus, for $\RR_{\m{even}}$, a tool implementing conditional term
rewriting with a random rule selection strategy would select a rule
with a failing condition about half the time. If we assume a worst
possible selection and count all rewrite steps performed during
the computation, we need $2^{n+1}-1$ steps to normalize $t_n$.

Although this exponential upper bound may come as a surprise, a powerful
rewrite engine like Maude~\cite{Maude} does not perform much better,
as can be seen from the data in Table~\ref{even/odd/Maude}.
Unlike term rewriting (which is non-deterministic by nature),
Maude employs a top-down rule selection strategy, so the order in
which the rules are presented makes a difference in the
outcome---although, as it turns out, not a substantial
one for Example~\ref{ex:evenodd} or other examples in this paper.
For rows three and four we presented
the rules to Maude in the order given in Example~\ref{ex:evenodd}.
If we change the order to \eqref{eo4}, \eqref{eo6}, \eqref{eo5},
\eqref{eo1}, \eqref{eo3}, \eqref{eo2} we obtain
the last two rows, showing an exponential number of steps in all cases.
Regardless of the order on the rules, we never obtain the optimal
linear bound for all tested terms.

\begin{table}[b]
\renewcommand{\arraystretch}{1.25}
\begin{tabular}{c|rrrrrrrrrrrrr}
$n$ & 0 & 1 & 2 & 3 & 4 & 5 & 6 & 7 & 8 & 9 & 10 & 11 & 12
\\
\hline
$2^{n+1}-1$ & 1 & 3 & 7 & 15 & 31 & 63 & 127 & 255 & 511 & 1023 &
2047 & 4095 & 8191
\\
\hline
$\m{even}(\m{s}^n(\m{0}))$ & 1 & 3 & 3 & 11 & 5 & 37 & 7 & 135 & 9 & 521 &
11 & 2059 & 13
\\
$\m{odd}(\m{s}^n(\m{0}))$ & 1 & 2 & 6 & 4 & 20 & 6 & 70 & 8 & 264 & 10 &
1034 & 12 & 4108
\\
\hline
$\m{even}(\m{s}^n(\m{0}))$ & 1 & 2 & 7 & 8 & 31 & 32 & 127 & 128 & 511 &
512 & 2047 & 2048 & 8191
\\
$\m{odd}(\m{s}^n(\m{0}))$ & 1 & 3 & 4 & 15 & 16 & 63 & 64 & 255 & 256 &
1023 & 1024 & 4095 & 4096
\end{tabular}
\renewcommand{\arraystretch}{1}
\medskip
\caption{Number of steps required to normalize
$\m{even}(\m{s}^n(\m{0}))$ and $\m{odd}(\m{s}^n(\m{0}))$ in Maude.}
\label{even/odd/Maude}
\end{table}

\medskip

From the above we conclude that \emph{a realistic definition of
conditional complexity should take failed computations into account}.
This conclusion opens new questions, however; most
pertinently, the question of how to handle repeated failed attempts.
It is obvious that we cannot allow repeatedly trying (and failing)
the same rule at the same position. For instance, it would be
foolish to attempt to reduce $\m{even}(\m{s}(\m{0}))$ with rule
\eqref{eo2}, fail, then try the same rule \emph{again} ten more times
before turning to \eqref{eo3} and count the steps for all the
failed attempts in the reduction cost. Thus, we must impose some
restrictions on duplicated attempts. To this end, let us consider
what constitutes a duplicated attempt.

\begin{exa}
\label{ex:poschoice}
The CCTRS $\RR_{\m{fg}}$ consists of the following two rewrite rules: \\
\begin{minipage}{.49\textwidth}
\begin{align}
\m{f}(x) &\to x &
\label{fg1}
\end{align}
\end{minipage}
\begin{minipage}{.5\textwidth}
\begin{align}
\m{g}(x) &\to \m{a} \,\Leftarrow\, x \approx \m{b}
\label{fg2}
\end{align}
\end{minipage}

\medskip

\noindent
Consider
$t_{n,m} = \m{f}^n(\m{g}(\m{f}^m(\m{a})))$. As we have not imposed an
evaluation strategy, one approach to evaluate this term could be as
follows.
We try using \eqref{fg2} on the subterm $\m{g}(\m{f}^m(\m{a}))$.
This fails in $m$ steps.
With \eqref{fg1} at the root position we obtain $t_{n-1,m}$.
We again attempt \eqref{fg2}, failing in $m$ steps. Repeating
this results in $n \cdot m$ rewrite steps before we reach
$t_{0,m}$.
\end{exa}

In this example we repeatedly attempt---and fail---to rewrite an
\emph{unmodified copy} of a subterm we tried before, with the same
rule. 
Although the position of the subterm $\m{g}(\m{f}^m(\m{a}))$ changes, we
already know that this reduction will fail.
Hence, once we fail a conditional rule on given subterms, it is reasonable
not to try the same rule again
on (copies of) the same subterms, even after a successful step.
In our model of computation we therefore wish to
keep track of previous failed attempts. This will be formalized in
Section~\ref{sec:complexity}.

\begin{exa}
\label{ex:poschoiceextended}
Continuing with $t_{0,m}$ from the preceding example, we could
try to use \eqref{fg2}, which fails in $m$ steps. Next, \eqref{fg1}
is applied on a subterm, and we obtain $t_{0,m-1}$. Again we try
\eqref{fg2}, failing after executing $m-1$ steps. Repeating this
alternation results eventually in the normal form $t_{0,0}$, but not
before computing $\frac{1}{2}(m^2 + 3m)$ rewrite steps in total.
\end{exa}

As in Example~\ref{ex:poschoice}, we keep returning to a subterm
which we have already tried before in an unsuccessful attempt. The
difference is that the subterm has been rewritten between successive
attempts.
According to the following general result, we need not reconsider a
failed attempt to apply a conditional rewrite rule if only the arguments
were changed.

\begin{lem}
\label{lem:local}
Given a CCTRS $\RR$,
suppose $s \xrightarrow{\smash{> \epsilon}}^* t$ and let
$\rho\colon \ell \to r \Leftarrow c$ be a rule such that
$s$ is an instance of $\ell$. If $t \xrightarrow{\epsilon}_\rho u$
then there exists a term $v$ such that 
$s \xrightarrow{\epsilon}_\rho v$ and $v \to^* u$.
\end{lem}

So if we can rewrite a term at the root position eventually,
and the term already matches the left-hand side of the rule with which
we can do so, then we can rewrite the term with this rule
immediately and obtain the same result.
Note that this lemma does \emph{not} assume confluence,
  quasi-decreasingness or left-linearity, and so is broadly applicable.

  \begin{proof}
  Suppose $s = \ell\sigma$ with
  $\dom(\sigma) \subseteq \Var(\ell)$, and let $\tau$ be a
  substitution such that $t = \ell\tau$, $u = r\tau$, and
  $\RR \vdash c\tau$, which exists since $\rho$ applies to $t$
  at the root position.
  Because $\ell$ is a basic term, all steps in
  $s \xrightarrow{\smash{> \epsilon}}^* t$ take place in
  the substitution part $\sigma$ of $\ell\sigma$ and thus
  $\sigma(x) \to^* \tau(x)$ for all $x \in \Var(\ell)$.
  Defining the substitution $\sigma'$ as follows,
  we have $s = \ell\sigma = \ell\sigma'$ and $\sigma' \to^* \tau$:
  \[
  \sigma'(x) = \begin{cases}
  \sigma(x) & \text{if $x \in \Var(\ell)$} \\
  \tau(x) & \text{if $x \notin \Var(\ell)$}
  \end{cases}
  \]
  Let $a \approx b$ be a condition in $c$. 
  From $\Var(b) \cap \Var(\ell) = \varnothing$ we infer
  $a\sigma' \to^* a\tau \to^* b\tau = b\sigma'$. It follows that
  $\RR \vdash c\sigma'$ and thus $s \xrightarrow{\epsilon}_\rho r\sigma'$.
  Hence we can take $v = r\sigma'$ as $r\sigma' \to^* r\tau = u$.
  \end{proof}

  From the above observations we conclude that, to avoid unnecessary
  repetitions, we can simply mark occurrences
  of defined symbols with the rules we have already tried without
  success---or, symmetrically, with the rules we have yet to try,
  as we will do in Section~\ref{sec:complexity}.

  Table~\ref{fg/Maude} compares these theoretical
  considerations to actual computations of $\RR_{\m{fg}}$ in Maude.
  Interestingly, Maude seems to perform worse on evaluating
  $\m{g}(\m{f}^m(\m{a}))$ than the realistic $m+1$ bound.
  Thus, it seems that Maude could benefit from incorporating
  the implications of Lemma~\ref{lem:local}. However, it should be
  remarked that when presenting $\RR_{\m{fg}}$ as a \emph{functional module}
  \cite[Chapter~6]{Maude}, Maude will switch to an innermost evaluation
  strategy and compute the normal form $\m{g}(\m{a})$ of
  $\m{f}^n(\m{g}(\m{f}^m(\m{a})))$ in $m+n$ steps.

  \begin{table}[tb]
  \renewcommand{\arraystretch}{1.25}
  \begin{tabular}{c|rrrrrrrrrrrrr}
  $n$ & 0 & 1 & 2 & 3 & 4 & 5 & 6 & 7 & 8 & 9 & 10 & 11 & 12
  \\ 
  $m$ & 0 & 1 & 2 & 3 & 4 & 5 & 6 & 7 & 8 & 9 & 10 & 11 & 12
  \\
  $n \cdot m$ & 0 & 1 & 4 & 9 & 16 & 25 & 36 & 49 & 64 & 81 & 100 & 121 &
  144 \\
  $\frac{1}{2}(m^2 + 3m)$ & 0 & 2 & 5 & 9 & 14 & 20 & 27 & 35 & 44 & 54 &
  65 & 77 & 90 \\
  \hline
  $\m{f}^n(\m{g}(\m{f}^m(\m{a})))$ & 0 & 3 & 8 & 16 & 28 & 45 & 68 & 98 &
  136 & 183 & 240 & 308 & 388 \\
  $\m{g}(\m{f}^m(\m{a}))$ & 0 & 2 & 6 & 13 & 24 & 40 & 62 & 91 & 128 & 174 &
  230 & 297 & 376
  \end{tabular}
  \renewcommand{\arraystretch}{1}
  \medskip
  \caption{Number of steps required to normalize
  $\m{f}^n(\m{g}(\m{f}^m(\m{a})))$ and $\m{g}(\m{f}^m(\m{a}))$ in Maude.}
  \label{fg/Maude}
  \end{table}

  In this paper, we will assume that rewriting takes
  Lemma~\ref{lem:local} into account, and thus avoids repeatedly
  reevaluating the same term. Also unlike Maude, we will not impose
  an evaluation order on the rules, nor a strategy for the position in
  a term that must be rewritten first, but allow free choice as is
  common in term rewriting.

  \medskip

  Another important aspect to consider is how to define a
  ``failed'' reduction. Intuitively, a rule $\ell \to r \Leftarrow c$
  should be considered not
  applicable on a term $\ell\sigma$ if there is no extension $\sigma'$
  of $\sigma$ such that $\RR \vdash c\sigma'$.
  Yet in
  Example~\ref{ex:evenodd} we already concluded that the second rule
  was not applicable to $t_9$ simply after reducing
  $\m{odd}(\m{s}^8(\m{0}))$ to its normal form $\m{false}$, because
  $\m{false}$ does not match the right-hand side $\m{true}$ of the
  condition.
  As remarked in Section~\ref{sec:preliminaries}, this is possible due
  to our restrictions.
  The following lemma makes this observation formal.

  \begin{lem}\label{lem:whatisfailure}
  Let $\rho\colon \ell \to r \Leftarrow c$
  be a rule in a confluent CCTRS $\RR$ and $\sigma$ a substitution 
  such that $\dom(\sigma) \subseteq \Var(\ell)$ and
  $\ell\sigma$ is quasi-decreasing.
Then $\rho$ is not applicable to $\ell\sigma$ if and only if there
is an extension $\sigma'$ of $\sigma$, and some
$1 \leqslant i \leqslant k$ such that
$\RR \vdash c_{i-1}^{\leqslant}\sigma$
and $a_i\sigma' \to^* u$ for some normal
form $u$ which is not an instance of $b_i$.
\end{lem}

\begin{proof}
(Recall that $c$ is $a_1 \approx b_1$, $\dots$, $a_k \approx b_k$ and
$c_{i-1}^\leqslant$ denotes $a_1 \approx b_1$, $\dots$,
$a_{i-1} \approx b_{i-1}$.)
We first prove the ``only if'' direction. So suppose that $\rho$ is
not applicable to $\ell\sigma$.
We define extensions $\sigma_0, \dots, \sigma_{i-1}$ of $\sigma$ such that
$\sigma_j(x) = \sigma(x)$ for all $x \in \Var(\ell)$,
$\dom(\sigma_j) \subseteq \Var(\ell,\seq[j]{b})$
for all $0 \leqslant j < i$, $\RR \vdash c_j^{\leqslant}\sigma_j$,
and $a_i\sigma_{i-1} \to^* u$ for some normal form $u$ which is not an
instance of $b_i$. Then $\sigma' = \sigma_{i-1}$ satisfies the
requirements of the lemma.
Let $\sigma_0 = \sigma$ and suppose $\seq[j-1]{\sigma}$ have been
defined. We have $\ell\sigma = \ell\sigma_{j-1} \quasistep a_j\sigma_{j-1}$
and hence $a_j\sigma_{j-1}$ is terminating by quasi-decreasingness.
Let $u$ be a normal form of $a_j\sigma_{j-1}$. If $u$ is an instance of
$b_j$, say $u = b_j\tau$ with $\dom(\tau) \subseteq \Var(b_j)$, then
we let $\sigma_j = \sigma_{j-1} \cup \tau$. Note that $\sigma_j$ is
well-defined as $\dom(\sigma_{j-1}) \cap \Var(b_j) = \varnothing$.
In this case $\sigma_j$ clearly satisfies the above conditions.
If $u$ is not an instance of $b_j$ then we are done by letting $i = j$.
Note that the latter must happen for some $j$ since we assumed that
$\rho$ is not applicable.

Next we prove the ``if'' direction. Suppose $\sigma'$, $i$, and
$u$ exist with the stated properties. For a proof by contradiction,
also suppose that the rule is applicable, so
there is an extension $\tau$ of $\sigma$ such that
$\RR \vdash c\tau$.
Define the substitution $\tau{\downarrow}$ as
$\{ x \mapsto \tau(x){\downarrow_\RR} \mid x \in \dom(\tau) \}$,
where $\tau(x){\downarrow_\RR}$ denotes the unique normal form of
$\tau(x)$.
(This is well-defined because $\ell\tau = \ell\sigma$ is quasi-decreasing
\and thus $a_j\tau$ and $b_j\tau$ are quasi-decreasing for all
$1 \leqslant j \leqslant k$. Therefore, all subterms of
$\ell\tau$, $b_1\tau$, $\dots$, $b_k\tau$ are terminating.
Since we may assume
$\dom(\tau) \subseteq \Var(\ell,\seq[k]{b})$---as each $a_j$ uses
only variables in $\Var(\ell,b_1,\dots,b_{j-1})$---confluence ensures
that $\tau(x)$ has a unique normal form for
every $x \in \dom(\tau)$.)
Fix $1 \leqslant j \leqslant k$.
We have $a_j\tau \to^* b_j\tau \to^* b_j(\tau{\downarrow})$
and $a_j\tau \to^* a_j(\tau{\downarrow})$. Since
$b_j$ is a constructor term, $b_j(\tau{\downarrow})$
is a normal form and thus
$a_j(\tau{\downarrow}) \to^* b_j(\tau{\downarrow})$ by confluence.
We claim that $\sigma'(x) \to^* \tau{\downarrow}(x)$ for all
$x \in \Var(\ell,\seq[i-1]{b})
$. If $x \in \Var(\ell)$ then
$\sigma'(x) = \sigma(x) = \tau(x)$.
Hence also $\sigma'(x) \to^* \tau{\downarrow}(x)$.
Suppose the claim holds for $x \in \Var(\ell,\seq[j-1]{b})$
with $1 \leqslant j < i$.
From $\Var(a_j) \subseteq \Var(\ell,\seq[j-1]{b})$ we infer
$a_j\sigma' \to^* a_j(\tau{\downarrow}) \to^* b_j(\tau{\downarrow})$.
Also $a_j\sigma' \to^* b_j\sigma'$ and thus
$b_j\sigma' \to^* b_j(\tau{\downarrow})$ by confluence. As $b_j$ is
a constructor term,
$\sigma'(x) \to^* \tau{\downarrow}(x)$ for all $x \in \Var(b_j)$.
This completes the proof of the claim. From the claim we find
$a_i\sigma' \to^* a_i(\tau{\downarrow}) \to^* b_i(\tau{\downarrow})$.
Using $a_i\sigma' \to^* u$ and confluence, we obtain
$b_i(\tau{\downarrow}) = u$, contradicting the assumption that
$u$ is not an instance of $b_i$.
\end{proof}

Thus, if we reduce the conditions of a rule and find a normal
form that does not instantiate the required right-hand side, we
can safely conclude that the rule does not apply.

A final aspect to consider is when to stop reducing a condition.
Should we stop once we obtain the right shape? Or should we
allow---or even enforce---reductions to normal form?

\begin{exa}\label{ex:whenstop}
Consider the following CCTRS implementing addition:
\begin{align*}
\m{plus}(x,y) &\to y \,\Leftarrow\, x \approx \m{0} &
\m{plus}(x,y) &\to \m{s}(\m{plus}(z,y)) \,\Leftarrow\, x \approx \m{s}(z)
\end{align*}
Let $t = \m{plus}(\m{plus}(\m{s}^9(\m{0}),\m{0}),\m{s}(\m{0}))$.
To reduce $t$ at the root with the second rule, we must
evalu- ate the condition $\m{plus}(\m{s}^9(\m{0}),\m{0}) \to^* \m{s}(z)$.
This is satisfied in a single step, reducing to
$\m{s}(z)\{ z \mapsto \m{plus}(\m{s}^8(\m{0}),\m{0}) \}$.
Should we therefore reduce to
$\m{s}(\m{plus}(\m{plus}(\m{s}^8(\m{0}),\m{0}),\m{s}(\m{0})))$ immediately?
Or should we continue reducing the
condition until we obtain a normal form $\m{s}^8(\m{0})$ and
then reduce to $\m{s}(\m{plus}(\m{s}^8(\m{0}),\m{s}(\m{0})))$?
Similarly, if we try to reduce $t$ at the root with the
first rule, we obtain in one step an instance of $\m{s}(z)$,
which does not unify with $\m{0}$. Since every reduct of $\m{s}(z)\sigma'$
is still an instance of $\m{s}(z)$, we could immediately
conclude that the condition will fail.
\end{exa}

Both questions are a matter of \emph{strategy}, and different
approaches might adopt different choices. One could argue that it
makes little sense to continue reducing a term for a condition when we
already know that it is satisfied, much like we said it makes no
sense to keep reevaluating the same failing condition.
However, since we aim for a
general definition, we have decided not to pursue this. That is, in
Example~\ref{ex:whenstop} we may choose to stop
evaluating the conditions and reduce with the rule (resp.\ conclude
failure) once we obtain an instance of the desired pattern (resp.\ a
term for which we can easily see that it will never reduce to such an
instance), but this is not compulsory. Specific evaluation
strategies can easily be added to the corresponding definitions and
transformations later.

Although a large part of our complexity notion deals with failed
reductions, there are many CCTRSs where this is not relevant. Consider
for instance Example~\ref{ex:fibonacci} in which the conditions of the
one conditional rule are not expected to fail;
they merely evaluate the result of a smaller term to a normal form
(or at least a constructor instance), and use its subterms.
Correspondingly, as can be seen in Table~\ref{fib/Maude},
the time needed to normalize terms in the Fibonacci CCTRS grows
roughly as fast as the Fibonacci sequence itself, with no
additional exponential growth for failed attempts.

\begin{table}[tb]
\renewcommand{\arraystretch}{1.25}
\begin{tabular}{c|rrrrrrrrrrrrr}
$n$ & 0 & 1 & 2 & 3 & 4 & 5 & 6 & 7 & 8 & 9 & 10 & 11 & 12
\\
\hline
$\m{fib}(\m{s}^n(\m{0}))$ & 1 & 3 & 7 & 13 & 23 & 40 & 69 & 119 & 205 &
353 & 607 & 1042 & 1785
\end{tabular}
\renewcommand{\arraystretch}{1}
\medskip
\caption{Number of steps required to normalize
$\m{fib}(\m{s}^n(\m{0}))$ in Maude.}
\label{fib/Maude}
\end{table}

\section{Conditional Complexity}
\label{sec:complexity}

In Section~\ref{sec:analysis} we have come to an intuitive
understanding of how a term $s$ in a (confluent) CTRS can be reduced,
and what the corresponding complexity should be:
\begin{itemize}
\item
In every step we select a position $p$ and a rule
$\ell \to r \Leftarrow c$ matching the corresponding subterm
(i.e., $s|_p = \ell\sigma$ for some $\sigma$).
\smallskip
\item
We then start evaluating the conditions in $c$ from left to right,
extending $\sigma$ as we go, until we have either
confirmed all conditions or obtain a failing condition.
\smallskip
\item
In the former case, we reduce $s|_p$ by this rule (obtaining
$s[r\sigma']_p$ for the extension $\sigma'$ of $\sigma$ found by
evaluating the conditions in $c$). In the latter case, we mark
the subterm $s|_p$ to indicate that we should not try the
rule $\ell \to r \Leftarrow c$ on this subterm again.
\smallskip
\item
The complexity of a conditional reduction is then obtained by
counting all rewrite steps, including those in successful and
failed condition evaluations.
\end{itemize}
In this section, we will formalize this intuition. A key aspect is the
ability to mark terms, so as to avoid continuously
repeating the same reduction attempt.
To achieve this, we will label defined function symbols by
subsets of the rules used to define them. Then, we define
a variation $\xrightharpoonup{}$ of the rewrite relation $\to$ which
explicitly includes failed computations. This relation is
used as the basis to define a complexity measure in a natural way.

We begin by defining labeled
terms and the labeled rewrite relation $\xrightharpoonup{}$ (Section
\ref{subsec:complexity:definition}). Then we analyze how
$\xrightharpoonup{}$ relates to the unlabeled conditional rewrite
relation $\to$ (Section~\ref{subsec:complexity:simulation}) and
define derivation height and complexity (Section
\ref{subsec:complexity:costs}).

\subsection{Labeled Terms and Reduction}
\label{subsec:complexity:definition}

\begin{defi}
Let $\RR$ be a CCTRS over a signature $\FF$. The labeled signature
$\GG$ is defined as $\FF_\CC \cup \{ f_R \mid \text{$f \in \FF_\DD$ and
$R \subseteq \RR{\restriction}f$} \}$. A labeled term is a term in
$\TT(\GG,\VV)$.
\end{defi}

Intuitively, the label $R$ in $f_R$ records the defining rules for $f$
which have not yet been tried. In examples we will generally
conflate the rules in $R$ with labels identifying them.

\begin{defi}
Let $\RR$ be a CCTRS over a signature $\FF$.
The mapping
$\m{label}\colon \TT(\FF,\VV) \to \TT(\GG,\VV)$ labels every defined
symbol $f$ with $\RR{\restriction}f$.
The mapping
$\m{erase}\colon \TT(\GG,\VV) \to \TT(\FF,\VV)$ removes the labels
from defined symbols.
\end{defi}

We obviously have $\m{erase}(\m{label}(t)) = t$ for
every $t \in \TT(\FF,\VV)$ and $\m{erase}(t) = \m{label}(t) = t$
for constructor terms $t$. The identity $\m{label}(\m{erase}(t)) = t$
does not hold for arbitrary $t \in \TT(\GG,\VV)$.

\begin{defi}
A \emph{labeled normal form} is a term in
$\TT(\FF_\CC \cup \{ f_\varnothing \mid f \in \FF_\DD \},\VV)$.
\end{defi}

\begin{exa}
\label{ex:labeledfib}
In $\RR_{\m{fib}}$ from Example~\ref{ex:fibonacci}, the labeled
signature $\GG$ consists of $\m{0}$, $\m{s}$, $\langle\cdot\rangle$,
$\m{+}_R$ for every subset $R$ of
$\{ \eqref{fibplus1}, \eqref{fibplus2} \}$,
and $\m{fib}_R$ for every subset $R$ of
$\{ \eqref{fib3}, \eqref{fib4} \}$.
We have
\[
\m{label}(\m{fib}(\m{s}(\m{0}) + \m{0})) =
\m{fib}_{\{ \eqref{fib3}, \eqref{fib4} \}}(\m{s}(\m{0})
\mathop{\m{+}_{\{ \eqref{fibplus1}, \eqref{fibplus2} \}}} \m{0})
\]
Examples of labeled normal forms are $\m{s}(\m{0})$ and
$\m{fib}_\varnothing(\m{0} \mathop{\m{+}_\varnothing}
\m{s}(\m{s}(\m{0})))$.
\end{exa}

The relation $\xrightharpoonup{}$ will be designed so that a
ground labeled term can be reduced if and only if it is not a labeled
normal form (see Lemma~\ref{lem:labelednfprop}).
First, with Definition~\ref{cc bot} we can remove rules from a
label if they will never apply due to an impossible
matching problem.

\begin{defi}
\label{cc bot}
Let $\RR$ be a CCTRS. For labeled terms $s$ and $t$ we write
$s \xrightharpoonup{\bot\,} t$ if there exist a position
$p \in \Pos(s)$ and a rewrite rule
$\rho\colon \ell \to r \Leftarrow c$ in $\RR$ such that
\begin{enumerate}
\item
$s|_p = f_R(\seq{s})$ with $\rho \in R$, 
\smallskip
\item
$t = s[f_{R \setminus \{ \rho \}}(\seq{s})]_p$, and
\smallskip
\item
there exist linear labeled normal forms $\seq{u}$ with fresh
variables and a substitu\-tion $\sigma$ such that
$s|_p = f_R(\seq{u})\sigma$ and $f(\seq{u})$ does not unify
with $\ell$.
\end{enumerate}
\end{defi}\smallskip

\noindent The last item ensures that rewriting (using $\xrightharpoonup{}$)
strictly below position $p$
cannot give a reduct that matches $\ell$,
since all such reducts will still be instances of $f_R(\seq{u})$.
Furthermore, if $s$ is ground and $s|_p = f_R(\seq{s})$ where $R$ is
non-empty and all $\seq{s}$ are labeled normal forms, then either
$f(\seq{s})$ is an instance of $\ell$, or
$\smash{\xrightharpoonup{\bot\,}}$ applies to $s|_p$.

\begin{exa}
\label{ex:labeledevenfollowup}
In Example~\ref{ex:labeledfib} we have
\[
\m{fib}_{\{ \eqref{fib3}, \eqref{fib4} \}}(\m{s}(\m{0})
\mathop{\m{+}_{\{ \eqref{fibplus1}, \eqref{fibplus2} \}}} \m{0})
~\xrightharpoonup{\bot\,}~
\m{fib}_{\{ \eqref{fib3}, \eqref{fib4} \}}(\m{s}(\m{0})
\mathop{\m{+}_{\{ \eqref{fibplus2} \}}} \m{0})
\]
because $\m{s}(\m{0}) + \m{0}$ does not unify with $\m{0} + y$,
and both $\m{s}(\m{0})$ and $\m{0}$ are linear labeled normal forms.
Also
\[
\m{fib}_{\{ \eqref{fib3}, \eqref{fib4} \}}(\m{s}(\m{0}
\mathop{\m{+}_{\{ \eqref{fibplus1} \}} \m{0}})
\mathop{\m{+}_{\{ \eqref{fibplus1}, \eqref{fibplus2} \}}} \m{0})
~\xrightharpoonup{\bot\,}~
\m{fib}_{\{ \eqref{fib3}, \eqref{fib4} \}}(\m{s}(\m{0}
\mathop{\m{+}_{\{ \eqref{fibplus1} \}}} \m{0})
\mathop{\m{+}_{\{ \eqref{fibplus2} \}}} \m{0})
\]
since $\m{s}(x)$ and $\m{0}$ are linear labeled normal forms
and $\m{s}(x) + \m{0}$ does not unify with $\m{0} + y$.
\end{exa}

Second, Definition~\ref{cc reduction} describes how to ``reduce''
labeled terms in general.
This definition is designed to reduce \emph{ground} terms in the
way roughly described in Section~\ref{sec:analysis}. Labeled
terms are reduced without any strategy, but subterms keep track of
which rules have not yet been attempted,
thus implicitly avoiding duplication.

\begin{defi}
\label{cc reduction}
A \emph{labeled} reduction is a sequence
$t_1 \xrightharpoonup{} t_2 \xrightharpoonup{} \:\cdots\:
\xrightharpoonup{} t_m$
of labeled terms where $s \xrightharpoonup{} t$ if either
$\smash{s \xrightharpoonup{\bot\,} t}$,
or there exist a position $p \in \Pos(s)$, a rewrite rule $\rho\colon
f(\seq{\ell}) \to r \Leftarrow a_1 \approx b_1, \dots, a_k \approx b_k$,
a substitution $\sigma$, and an index
$0 \leqslant j \leqslant k$
with
\begin{enumerate}
\item
$s|_p = f_R(\seq{s})$ with $\rho \in R$ and
$s_i = \ell_i\sigma$ for all $1 \leqslant i \leqslant n$,
\smallskip
\item\label{cc reduction conditions}
$\m{label}(a_i)\sigma \xrightharpoonup{}^* b_i\sigma$ for all
$1 \leqslant i \leqslant j$,
\end{enumerate}
and either
\begin{enumerate}[resume]
\item\label{cc reduction success}
$j = k$ and $t = s[\m{label}(r)\sigma]_p$
\end{enumerate}
in which case we speak of a \emph{successful} step, or
\begin{enumerate}[resume]
\item\label{cc reduction failure}
$j < k$ and there exist a linear labeled normal form $u$ and a 
substitution $\tau$ such that
\begin{enumerate}
\item
$\m{label}(a_{j+1})\sigma \xrightharpoonup{}^* u\tau$ and
$u$ does not unify with $b_{j+1}$, and
\smallskip
\item
$t = s[f_{R \setminus \{ \rho \}}(\seq{s})]_p$,
\end{enumerate}
\end{enumerate}
which is a \emph{failed} step.\newpage

A \emph{complexity-conscious reduction} is a labeled reduction
complete with proofs of the sub-requirements, i.e., a sequence
$(t_1 \xrightharpoonup{} t_2),\dots,(t_{m-1} \xrightharpoonup{} t_m)$
of \emph{complexity-conscious steps}, where each complexity-conscious
step $s \xrightharpoonup{} t$ is a tuple combining
$s$, $t$, $p$, $\rho$, $j$ and the complexity-conscious reductions
$\m{label}(a_i)\sigma \xrightharpoonup{}^* b_i\sigma$ for
$1 \leqslant i \leqslant j$
and possibly $\m{label}(a_{j+1})\sigma \xrightharpoonup{}^* u\tau$.
We will denote complexity-conscious reductions as labeled reductions,
and simply assume the underlying condition evaluations given.
\end{defi}

It is easy to see that for all ground labeled terms $s$ which
are not labeled normal forms, either
$\smash{s \xrightharpoonup{\bot} t}$ for some term $t$
or there are $p$, $\rho$, $\sigma$ such that 
$s|_p$ ``matches'' $\rho$ in the sense that the first
requirement in Definition~\ref{cc reduction} is satisfied.
In the latter case, the conditions are evaluated left-to-right;
as all $b_j$ are linear constructor terms on fresh
variables, $\m{label}(a_j)\sigma \xrightharpoonup{}^{*} b_j\sigma$
simply indicates that $a_j\sigma$---with labels added to allow reducing
defined symbols in $a_j$---reduces to an instance of $b_j$.
A successful reduction occurs when we manage to reduce each 
$\m{label}(a_i)\sigma$ to $b_i\sigma$. A failed reduction occurs
when we start reducing $\m{label}(a_i)\sigma$ and obtain a term that 
will never reduce to an instance of $b_i$.

\begin{exa}
\label{ex:labeledfibreduction}
Continuing Example~\ref{ex:labeledfib}, we
have the following complexity-conscious reduction:
\begin{align*}
\m{fib}_{\{ \eqref{fib3}, \eqref{fib4} \}}(\m{s}(\m{0})
\mathop{\m{+}_{\{ \eqref{fibplus1}, \eqref{fibplus2} \}}} \m{0})
&~\xrightharpoonup{\bot\,}~
\m{fib}_{\{ \eqref{fib3}, \eqref{fib4} \}}(\m{s}(\m{0})
\mathop{\m{+}_{\{ \eqref{fibplus2} \}}} \m{0})
\tag{Example~\ref{ex:labeledevenfollowup}} \\[-.5ex]
&~\xrightharpoonup{\phantom{\bot\,}}~
\m{fib}_{\{ \eqref{fib3}, \eqref{fib4} \}}(\m{s}(\m{0}
\mathop{\m{+}_{\{ \eqref{fibplus1}, \eqref{fibplus2} \}}} \m{0}))
\tag{successful step} \\[-.5ex]
&~\xrightharpoonup{\phantom{\bot\,}}~
\langle \m{s}(\m{0}), \m{s}(\m{0}) \rangle
\tag{successful step}
\intertext{The first successful step uses the unconditional rule
\eqref{fibplus2}. The second successful step uses rule \eqref{fib4}
and the complexity-conscious reductions}
\m{fib}_{\{ \eqref{fib3}, \eqref{fib4} \}}(\m{0}
\mathop{\m{+}_{\{ \eqref{fibplus1}, \eqref{fibplus2} \}}} \m{0})
&~\xrightharpoonup{}~
\m{fib}_{\{ \eqref{fib3}, \eqref{fib4} \}}(\m{0})
~\xrightharpoonup{}~
\langle \m{0}, \m{s}(\m{0}) \rangle
\intertext{and}
\m{0} \mathop{\m{+}_{\{ \eqref{fibplus1}, \eqref{fibplus2} \}}}
\m{s}(\m{0})
&~\xrightharpoonup{}~
\m{s}(\m{0})
\end{align*}
for the evaluation of the conditions,
all by successful steps without conditions.
\end{exa}

\begin{exa}
\label{ex:labeledevenreduction}
In the CCTRS of Example~\ref{ex:evenodd} we have the following
complexity-conscious reduction:
\begin{align*}
\m{even}_{\{ \eqref{eo1}, \eqref{eo2}, \eqref{eo3} \}}(\m{s}(\m{0}))
&~\xrightharpoonup{\phantom{\bot\,}}~
\m{even}_{\{ \eqref{eo1}, \eqref{eo3} \}}(\m{s}(\m{0}))
\tag{failed step} \\[-.5ex]
&~\xrightharpoonup{\bot\,}~
\m{even}_{\{ \eqref{eo3} \}}(\m{s}(\m{0}))
\tag{matching failure} \\[-.5ex]
&~\xrightharpoonup{\phantom{\bot\,}}~
\m{false}
\tag{successful step}
\end{align*}
The first step fails with
$j = 0$ because
\[
\m{label}(\m{odd}(\m{0})) =
\m{odd}_{\{ \eqref{eo4}), \eqref{eo5}, \eqref{eo6} \}}(\m{0})
~\xrightharpoonup{\bot\,}~
\m{odd}_{\{ \eqref{eo4}), \eqref{eo5} \}}(\m{0})
~\xrightharpoonup{\phantom{\bot\,}}~
\m{false}
\]
and $\m{false}$ is a linear labeled normal form which does
not unify with $\m{true}$. The third step succeeds because
$\m{label}(\m{even}(\m{0})) =
\m{even}_{\{ \eqref{eo1}), \eqref{eo2}, \eqref{eo3} \}}(\m{0})
~\xrightharpoonup{}~ \m{true}$.
\end{exa}

There is one possibility remaining which is not covered by
Definition~\ref{cc reduction}; in a non-quasi-decreasing setting, a
condition may give rise to an infinite reduction, neither failing nor
succeeding. To handle this case we introduce a third definition.

\begin{defi}
\label{cc infinite}
We write $s \rhrhd t$ if there exist a position
$p \in \Pos(s)$, a rewrite rule $\rho\colon f(\seq{\ell}) \to r
\Leftarrow a_1 \approx b_1, \dots, a_k \approx b_k$, a substitution
$\sigma$ and $1 \leqslant j < k$ such that
\begin{enumerate}
\item
\label{cc infinite 1}
$s|_p = f_R(\seq{s})$ with $\rho \in R$ and
$s_i = \ell_i\sigma$ for all $1 \leqslant i \leqslant n$,
\smallskip
\item
\label{cc infinite 2}
$\m{label}(a_i)\sigma \xrightharpoonup{}^* b_i\sigma$ for all
$1 \leqslant i \leqslant j$, and
\smallskip
\item
$t = \m{label}(a_{j+1})\sigma$
\end{enumerate}
We write $s \xrightharpoonup{\infty\,}$ if there is an infinite
sequence
$s = s_0 \mathrel{(\xrightharpoonup{} \cup \rhrhd)}
s_1 \mathrel{(\xrightharpoonup{} \cup \rhrhd)} \cdots$
\end{defi}

Definition~\ref{cc infinite}
completes labeled reduction;
all ground labeled terms $s$
are either labeled normal forms, or can be reduced using
$\xrightharpoonup{}$ or $\rhrhd$.
This is verified in the following lemma.

\begin{lem}\label{lem:labelednfprop}
For every ground labeled term $s$ one of the following
alternatives holds:
\begin{enumerate}
\item
\label{alternative 1}
$s \xrightharpoonup{\infty\,}$
\item
\label{alternative 2}
$s \xrightharpoonup{} t$ for some term $t$, or
\item
\label{alternative 3}
$s$ is a labeled normal form.
\end{enumerate}
\end{lem}

\begin{proof}
We non-deterministically construct a (finite or infinite) sequence
$s = s_0$, $s_1$, $s_2$, $\dots$ of ground terms as
follows. Assuming $s_i$ has been defined,
if there is some $u$ such that $s_i \xrightharpoonup{} u$ then we take any
such $u$ as $s_{i+1}$. Otherwise, if there is some $v$ with $s_i \rhrhd v$
then we take $s_{i+1} = v$. If there are multiple such $v$, we choose one
with the largest possible number $j$
(cf.~\eqref{cc infinite 2} in Definition~\ref{cc infinite})
of successful conditions with respect to a rule $\rho$
satisfying~\eqref{cc infinite 1} in Definition~\ref{cc infinite}.
If no $s_{i+1}$ has been defined, we terminate the construction and let
$N = i$.

If the constructed sequence is infinite then
$s \xrightharpoonup{\infty\,}$ and thus
statement~\eqref{alternative 1} holds.
So suppose the sequence is finite.

We claim that $s_N$ is a labeled normal form. For a proof by
contradiction, assume that $s_N$ is not a labeled normal form, so it
has a subterm whose root symbol has a non-empty label. Choosing a
minimal such subterm, we find a position $p \in \Pos(s_N)$ such that
$s_N|_p = f_R(\seq{u})$ where $\seq{u}$ are labeled normal forms and
$R \neq \varnothing$, say
$\rho\colon f(\seq{\ell}) \to r \Leftarrow c \in R$.
Now, if $f(\seq{u})$ does not instantiate $f(\seq{\ell})$, then
$\smash{s_N \xrightharpoonup{\bot\,} s_N[f_{R \setminus \{ \rho \}}
(\seq{u})]_p}$ because $\seq{u}$ are ground labeled normal forms,
contradicting the fact that $s_N$ is the last element of
the sequence. It follows that a substitution $\sigma$ exists such
that $u_i = \ell_i\sigma$ for $1 \leqslant i \leqslant n$. If $k = 0$
then $s_N \xrightharpoonup{} s_N[r\sigma]_p$, otherwise
$s_N \rhrhd \m{label}(a_1)\sigma$.
In both cases we obtain a contradiction to the choice of $N$.

Next we prove $s_i \xrightharpoonup{} s_{i+1}$ for
$0 \leqslant i < N$. Aiming for a contradiction, consider the
largest $i$ such that $s_i \xrightharpoonup{} s_{i+1}$ does not hold.
Then we have $s_i \rhrhd s_{i+1}$, so there
exist a position $p$, a rule
$\rho\colon f(\seq{\ell}) \to r \Leftarrow c \in R$,
an index $1 \leqslant j < k$, and a substitution
$\sigma$ such that $s_i|_p = f_R(\seq{u})$ with
$f(\seq{u}) = f(\seq{\ell})\sigma$,
$\m{label}(a_l)\sigma \xrightharpoonup{}^* b_l\sigma$ for all
$1 \leqslant l \leqslant j$, and $s_{i+1} = \m{label}(a_{j+1})\sigma$.
We have $s_{i+1} \xrightharpoonup{}^* s_N$ by the choice of $i$.
If $s_N$ is not an instance of $b_{j+1}$ then, since $s_N$ is a
ground labeled normal form, the condition has failed and we
have $s_i \xrightharpoonup{} s_i[f_{R \setminus \{ \rho \}}(\seq{u})]_p$,
contradicting the choice for $s_{i+1}$. So $s_N$ does instantiate
$b_{j+1}$. But then, if $j+1 < k$, we should have chosen
$s_{i+1} = \m{label}(a_{j+2})\sigma$ according to the construction of
$s_{i+1}$, and if $j+1 = k$ then
$s_{i+1} = s_i[r\sigma]_p$
should have been chosen instead.

Now, if $N = 0$ then $s = s_N$ is a labeled normal form and thus
statement~\eqref{alternative 3} holds. If $N > 0$
statement~\eqref{alternative 2} holds as $s = s_0 \xrightharpoonup{} s_1$.
\end{proof}

Note that the three alternatives in Lemma~\ref{lem:labelednfprop}
are not exclusive: It is possible to have 
$s \xrightharpoonup{} t$ as well as
$s \xrightharpoonup{\,\infty}$ for a term $s$.

\subsection{Labeled versus Unlabeled Reduction}
\label{subsec:complexity:simulation}

\mbox{}

\bigskip

\noindent
The relation $\xrightharpoonup{}$ provides an alternative approach to
evaluation which keeps track of failed rule application attempts,
whereas $\xrightharpoonup{\infty\,}$
is the counterpart of non-quasi-decreasingness.
As may be expected, there is a
strong connection between the relations $\to$ and $\xrightharpoonup{}$.
This connection is made formal in
Theorem~\ref{lem:simulate} and the subsequent lemmata.

\begin{defi}
Let $\RR$ be semi-finite and $t \in \TT(\GG,\VV)$.
We write $\|t\|$ for
the total number of rules occurring in all labels in $t$.
\end{defi}

Semi-finiteness ensures that $\|t\|$ is a well-defined
natural number.

\begin{thm}\label{lem:simulate}
Let $\RR$ be a CCTRS.
\begin{enumerate}
\item\label{lem:simulate:lr}
Let $s, t \in \TT(\FF,\VV)$.
\begin{enumerate}
\item\label{lem:simulate:lr:single}
If $s \to t$ then $\m{label}(s) \xrightharpoonup{} \m{label}(t)$.
\item\label{lem:simulate:lr:multi}
If $s \to^* t$ then $\m{label}(s) \xrightharpoonup{}^* \m{label}(t)$.
\end{enumerate}
\item\label{lem:simulate:rl}
Let $s, t \in \TT(\GG,\VV)$.
\begin{enumerate}
\item\label{lem:simulate:rl:single}
If $s \xrightharpoonup{} t$ then either
$\m{erase}(s) \to \m{erase}(t)$ or both $\m{erase}(s) = \m{erase}(t)$
and, if $\RR$ is semi-finite, $\|s\| > \|t\|$.
\item\label{lem:simulate:rl:multi}
If $s \xrightharpoonup{}^* t$ then $\m{erase}(s) \to^* \m{erase}(t)$.
\end{enumerate}
\end{enumerate}
\end{thm}

\begin{proof}
We use induction on the total number of
\emph{rewrite steps} of $\to$ and $\xrightharpoonup{}$, respectively.
This is the number of steps used both directly in the reduction, and
those needed to verify the conditions $a_i\sigma \to^* b_i\sigma$ or
$\m{label}(a_i)\sigma \xrightharpoonup{}^* b_i\sigma$.
\begin{enumerate}
\item
We derive cases~\eqref{lem:simulate:lr:single}
and~\eqref{lem:simulate:lr:multi} by simultaneous
induction on the total number of
rewrite steps needed to derive $s \to t$ and $s \to^* t$.
\begin{enumerate}
\item
There exist a position $p \in \Pos(s)$,
a rule $\rho\colon \ell \to r \Leftarrow c$,
and a substitution $\sigma$ such that $s|_p =
\ell\sigma$, $t = s[r\sigma]_p$, and
$\RR \vdash c\sigma$.
Let $\sigma'$ be the (labeled) substitution $\m{label} \circ \sigma$.
Fix $1 \leqslant i \leqslant k$.
We have $\m{label}(a_i\sigma) = \m{label}(a_i)\sigma'$ and
$\m{label}(b_i\sigma) = b_i\sigma'$ (as $b_i$ is a constructor term).
Because $a_i\sigma \to^* b_i\sigma$ is used in the derivation of
$s \to t$ we can apply the induction hypothesis for part \(b,
resulting in
$\m{label}(a_i\sigma) \xrightharpoonup{}^* \m{label}(b_i\sigma)$.
Furthermore, writing $\ell = f(\seq{\ell})$, we obtain
$\m{label}(\ell) = f_{\RR{\restriction}f}(\seq{\ell})$.
Hence $\m{label}(s) = \m{label}(s)[\m{label}(\ell)\sigma']_p
\xrightharpoonup{} \m{label}(s)[\m{label}(r)\sigma']_p = \m{label}(t)$
because conditions (1)--(3) in Definition~\ref{cc reduction} are
satisfied.
\item
If $s = t$ then the result is obvious.
If $s \to u \to^* t$ then $\m{label}(s) \xrightharpoonup{} \m{label}(u)$
follows by case~\eqref{lem:simulate:lr:single}, and the
induction hypothesis yields
$\m{label}(u) \xrightharpoonup{} \m{label}(t)$.
\end{enumerate}
\smallskip
\item
We prove both statements by simultaneous induction on the total
number of steps required to derive $s \xrightharpoonup{} t$ and
$s \xrightharpoonup{}^* t$.
For part \(a\ we distinguish two cases.
\begin{itemize}
\item
Suppose $s \xrightharpoonup{\bot\,} t$ or $s \xrightharpoonup{} t$
by a failed step. In either case we have
$\m{erase}(s) = \m{erase}(t)$. Moreover, if all labels have
finite size, also $\|s\| = \|t\| + 1$.
\item
Suppose $s \xrightharpoonup{} t$ by a successful step.
So there exist a
position $p \in \Pos(s)$, a rule $\rho\colon \ell \to r \Leftarrow c$
in $\RR$, a substitution
$\sigma$, and terms $\ell', a_1', \dots, a_k'$ such that
$s|_p = \ell'\sigma$ with $\m{erase}(\ell') = \ell$,
$a_i'\sigma \xrightharpoonup{}^* b_i\sigma$ with $\m{erase}(a_i') = a_i$
for all $1 \leqslant i \leqslant k$, and $t = s[\m{label}(r)\sigma]_p$.
Let $\sigma'$ be the (unlabeled) substitution $\m{erase} \circ \sigma$.
We have $\m{erase}(s) = \m{erase}(s)[\ell\sigma']_p$
and $\m{erase}(u) = \m{erase}(s)[r\sigma']_p$.
Since the sequence $a_i'\sigma \xrightharpoonup{}^* b_i\sigma$ is
used as a strict subpart of the derivation of
$s \xrightharpoonup{} t$, we obtain
$a_i\sigma' = \m{erase}(a_i'\sigma) \to^* \m{erase}(b_i\sigma) =
b_i\sigma'$ from the induction hypothesis, for all
$1 \leqslant i \leqslant k$. Hence $\RR \vdash c\sigma'$, so indeed
$\m{erase}(s) \to \m{erase}(t)$.
\end{itemize}
Again, part \(b\ easily follows from part \(a.
\qedhere
\end{enumerate}
\end{proof}

\begin{exa}
In Examples~\ref{ex:poschoice} and~\ref{ex:poschoiceextended} we
encountered the reduction
\begin{align*}
t_{n,m} = \m{f}^n(\m{g}(\m{f}^m(\m{a})))
&~\to^*~
\m{g}(\m{a}) = t_{0,0}
\intertext{By Theorem~\ref{lem:simulate}, we immediately obtain a
labeled reduction}
\m{label}(t_{n,m}) = \m{f}^n_{\{ \eqref{fg1} \}}
(\m{g}^{\phantom{n}}_{\{ \eqref{fg2} \}}(\m{f}_{\{ \eqref{fg1} \}}^m
(\m{a}))) 
&~\xrightharpoonup{}^*~
\m{g}^{\phantom{n}}_{\{ \eqref{fg2} \}}(\m{a}) = \m{label}(t_{0,0})
\end{align*}
Note that we can reduce this term further to
$\m{g}_{\varnothing}(\m{a})$ and obtain a labeled normal form.
\end{exa}


\begin{lem}
\label{lem:simulate-substep}
Let $\RR$ be a CCTRS.
\begin{enumerate}
\item\label{lem:simulate-substep:lr}
If $s, t \in \TT(\FF,\VV)$ and $s \quasistep t$ then
$\m{label}(s) \rhrhd \m{label}(t)$.
\item\label{lem:simulate-substep:rl}
If $s, t \in \TT(\GG,\VV)$ and $s \rhrhd t$
then $\m{erase}(s) \quasistep \m{erase}(t)$.
\end{enumerate}
\end{lem}

\begin{proof}\hfill
\begin{enumerate}
\item
There exist a position $p \in \Pos(s)$,
a rule $\rho\colon \ell \to r \Leftarrow c$,
a substitution $\sigma$, and an
index $1 \leqslant i \leqslant k$ such that
$s|_p = \ell\sigma$, $t = a_i\sigma$ and
$a_j\sigma \to^* b_j\sigma$ for all $1 \leqslant j < i$.
Write $\sigma' = \m{label} \circ \sigma$.
By Theorem~\ref{lem:simulate}(\ref{lem:simulate:lr}),
$\m{label}(a_j)\sigma' = \m{label}(a_j\sigma) \xrightharpoonup{}^*
\m{label}(b_j\sigma) = b_j\sigma'$ for all $1 \leqslant j < i$.
Let $\ell = f(\seq{\ell})$ and $R = \RR{\restriction}f$. Clearly,
$\rho \in R$ and therefore $\m{label}(s)|_p = f_R(\m{label}(s_1),
\dots,\m{label}(s_n)) = f_R(\ell_1\sigma',\dots,\ell_n\sigma')
\rhrhd \m{label}(a_i)\sigma' = \m{label}(t)$ as
required.
\item
There exists position
$p \in \Pos(s)$, a rule $\rho\colon f(\seq{\ell}) \to r \Leftarrow c$,
a substitution
$\sigma$, and an index $1 \leqslant i \leqslant k$ such that
$s|_p = f_R(\seq{s})$ with $\rho \in R$ and $s_j = \ell_j\sigma$ for all
$1 \leqslant j \leqslant n$, $\m{label}(a_j)\sigma \to^* b_j\sigma$ for all
$1 \leqslant j < i$, and $t = \m{label}(a_i)\sigma$.
Write $\sigma' = \m{erase} \circ \sigma$.
By Theorem~\ref{lem:simulate}(\ref{lem:simulate:rl}),
$\m{erase}(\m{label}(a_j)\sigma) = a_j\sigma' \to^* b_j\sigma' =
\m{erase}(b_j\sigma)$ for $1 \leqslant j < i$.
We obtain $\m{erase}(s) = \m{erase}(s)[f(\seq{\ell})\sigma']_p
\quasistep a_i\sigma' = \m{erase}(t)$.
\qedhere
\end{enumerate}
\end{proof}

\begin{lem}
\label{lem:quasidecreasing}
A term $s \in \TT(\FF,\VV)$ in a semi-finite CCTRS $\RR$ is
non-quasi-decreasing if and only
if $\m{label}(s) \xrightharpoonup{\infty\,}$.
\end{lem}

\begin{proof}
If $s$ is not quasi-decreasing then there exists an infinite sequence
$s = u_0 \mathrel{(\to \cup \quasistep)} u_1
\mathrel{(\to \cup \quasistep)} \cdots\:$. We obtain
$\m{label}(u_0)
\mathrel{(\xrightharpoonup{} \cup \rhrhd)}
\m{label}(u_1)
\mathrel{(\xrightharpoonup{} \cup \rhrhd)}
\cdots$ from Theorem~\ref{lem:simulate}(\ref{lem:simulate:lr})
and Lemma~\ref{lem:simulate-substep}(\ref{lem:simulate-substep:lr}).
Thus $\m{label}(s) = \m{label}(u_0) \xrightharpoonup{\infty\,}$.
Conversely, if $\m{label}(s) \xrightharpoonup{\infty\,}$
then there is an infinite sequence $\m{label}(s) = u_0
\mathrel{(\xrightharpoonup{} \cup \rhrhd)}
u_1 \mathrel{(\xrightharpoonup{} \cup \rhrhd)}
\cdots\:$.
From Theorem~\ref{lem:simulate}(\ref{lem:simulate:rl})
and Lemma~\ref{lem:simulate-substep}(\ref{lem:simulate-substep:rl})
we obtain $\m{erase}(u_i) \to^= \m{erase}(u_{i+1})$ or
$\m{erase}(u_i) \quasistep \m{erase}(u_{i+1})$
for every $i \geqslant 0$.
Since $\m{erase}(u_i) = \m{erase}(u_{i+1})$ implies
$\|u_i\| > \|u_{i+1}\|$, this gives an infinite sequence of
$\to$ and $\quasistep$ steps starting from $\m{erase}(u_0) = s$.
\end{proof}

\begin{exa}
Consider the CCTRS consisting of the single rule
$\rho\colon \m{a} \to \m{b} \,\Leftarrow\, \m{a} \approx \m{b}$.
We have
$\m{label}(\m{a}) = \m{a}_{\{ \rho \}} \xrightharpoonup{\rhd\,} 
\m{label}(\m{a})$. Hence
$\m{label}(\m{a}) \xrightharpoonup{\infty\,}$
and thus $\m{a}$ is non-quasi-decreasing.
\end{exa}

We have now transposed conditional rewriting to an essentially
equivalent relation on labeled terms, which enables us to keep track
of failed computations.\newpage

\subsection{Derivation Height and Complexity}
\label{subsec:complexity:costs}

\mbox{}

\bigskip

Now we show how labeled---or rather,
complexity-conscious---reduction gives rise to
conditional complexity. With failures now explicitly included in the
reduction relation, the only hurdle to defining derivation height
is the question of how exactly to handle the evaluation of
conditions. To this end, we assign an evaluation cost to individual
steps.

\begin{defi}
The \emph{cost} $\m{cost}(s \xrightharpoonup{}^* t)$ of a
complexity-conscious reduction $s \xrightharpoonup{}^* t$ is
the sum of the costs of its steps. The cost of a step
$s \xrightharpoonup{} t$ is $0$ if $s \xrightharpoonup{\bot\,} t$,
\[
1 + \sum_{i=1}^k 
\m{cost}(\m{label}(a_i)\sigma \xrightharpoonup{}^* b_i\sigma)
\]
in case of a successful step $s \xrightharpoonup{} t$, and
\[
\sum_{i=1}^j 
\m{cost}(\m{label}(a_i)\sigma \xrightharpoonup{}^* b_i\sigma)
+
\m{cost}(\m{label}(a_{j+1})\sigma \xrightharpoonup{}^* u\tau)
\]
in case of a failed step $s \xrightharpoonup{} t$.
\end{defi}

Intuitively, the cost of a reduction measures the number of
successful rewrite steps, both direct and in condition evaluations,
but does not count the mere removal of a rule from a label.
This is why the cost of a failed step is the cost to evaluate its
conditions and conclude failure, while for successful steps we add
one for the step itself.

\begin{exa}
The cost of the reduction in Example~\ref{ex:labeledfibreduction} is
$0 + 1 + 4 = 5$, where the $4 = 1 + 3$ includes the three steps
in the conditions.
The cost of the reduction in Example~\ref{ex:labeledevenreduction} is
$1 + 0 + 2 = 3$.
Note that in both cases, the cost is simply obtained by counting the
number of successful rewrite steps, including those occurring in a
condition evaluation.
\end{exa}

\begin{defi}
\label{def:dh}
The \emph{derivation height}
$\m{dh}(s)$ of a labeled term $s$
in a semi-finite CCTRS is defined as
\[
\max\,(\{\m{cost}(s \xrightharpoonup{}^* t) \mid t \in \TT(\GG,\VV)\}
\cup \{ \infty \mid s \xrightharpoonup{\infty\,} \})
\]
where $\infty > n$ for all $n \in \N$.
\end{defi}

That is, a labeled term $s$ has infinite derivation height if
$s \xrightharpoonup{\,\infty}$, and the maximum cost of any reduction
starting in $s$ otherwise.  Since $\RR$ is semi-finite,
the set of possible values $\m{cost}(s \xrightharpoonup{}^* t)$
can only be unbounded if $s \xrightharpoonup{\,\infty}$, in which
case $\m{dh}(s) = \max(\langle$some infinite set$\rangle \cup
\{ \infty \}) = \infty$. In other cases, the set of costs is necessarily
finite, and hence the derivation height is well-defined, and in
$\N$. Note that for $t \in \TT(\FF)$, the derivation
height of $\m{label}(t)$ is infinite if and only if $t$ is
quasi-decreasing, by Lemma~\ref{lem:quasidecreasing}.

\medskip

We have limited interest to semi-finite CCTRSs primarily to follow the
standard in complexity for unconditional term rewriting, where TRSs are
assumed to be finite.
It is certainly possible to extend the definition towards
non-semi-finite CCTRSs, simply by taking the \emph{infimum} instead
of the \emph{maximum} of the set in Definition~\ref{def:dh},
in which case we might obtain an infinite
derivation height even for the labeled version of a quasi-decreasing
term. This would happen both if there are reductions of arbitrarily
high cost starting in $\m{label}(s)$, or if we obtain an
infinite reduction of rule-removal steps,
e.g.\ $\smash{\m{label}(s) \xrightarrow{\bot\,}
s_1 \xrightarrow{\bot\,} s_2 \xrightarrow{\bot\,}} \cdots$.
One might argue that this is justified, as finding an
appropriate rule to apply may take arbitrarily long. However, in the
unconditional setting, it seems unnatural to assign an infinite
derivation height to, for instance, a normal form. Given
that non-semi-finite TRSs are of very little practical interest, we
prefer to leave this discussion to another work.

\begin{defi}
The \emph{conditional derivational complexity} of a semi-finite
CCTRS $\RR$ is defined as $\m{cdc}_\RR(n) =
\max\,\{\m{dh}(\m{label}(t)) \mid |t| \leqslant n \}$. If we restrict
$t$ to basic terms we arrive at the \emph{conditional runtime
complexity} $\m{crc}_\RR(n)$.
\end{defi}

Arguably, the case where the CCTRS $\RR$ is not
quasi-decreasing is not very interesting for complexity (unless
perhaps all terms of interest, e.g.\ all basic terms, are
quasi-decreasing). The main reason why we consider systems without
this restriction is to show that the transformation methods we use
preserve the fundamental properties of a CCTRS. Thus, we can for
instance guarantee that the TRS obtained in the next section is
terminating if and only if the original CCTRS is quasi-decreasing.
This allows us to obtain completeness results, and to use complexity
methods to prove quasi-decreasingness as well.

\medskip

Continuing the discussion in Section~\ref{sec:analysis}, we claim
that for a ground term $s \in \TT(\FF)$, the derivation height
$\m{dh}(\m{label}(s))$ gives a realistic and (in the absence of a
reduction strategy) narrow bound on the time needed to normalize
$s$. That is, we can always find a normal form of $s$ in
$\OO(\m{dh}(\m{label}(s)))$ steps (by rewriting $\m{label}(s)$ using
$\xrightharpoonup{}$). A worst-case derivation following
the intuition laid out at the start of this section requires
$\Omega(\m{dh}(\m{label}(s))$ steps.

\section{Complexity Transformation}
\label{sec:transformation}

The notion of complexity introduced in the preceding section
has the downside that we cannot easily reuse existing
complexity results and tools. Therefore, we will consider a
transformation to unconditional rewriting where, rather than tracking
rules in the labels of the defined function symbols, we will keep
track of them in separate arguments, but restrict reduction by
adopting a suitable \emph{context-sensitive replacement map}.
This transformation is based directly on the CCTRS $(\FF,\RR)$,
but in Section \ref{subsec:transformation:proofs} we will see how it
relates to the labeled system and the labeled rewrite relation
$\xrightharpoonup{}$. In particular, we will see that the
unconditional rewrite relation defined in
Section~\ref{subsec:transformation:definition} both preserves and
reflects complexity. To this end, however, we will have to limit
interest to \emph{strong} CCTRSs, as defined in
Section~\ref{sec:preliminaries}, since
we rely on (e.g.) left-linearity to be able to test when rules do
\emph{not} apply.

Our transformation builds on the ideas of the structure-preserving
transformations in~\cite{ABH03,SR06}, but differs in particular
by its use of context-sensitivity, by forcing that the conditions for
different rules are evaluated separately, and by using additional
symbols $f_i^j$ to mark when the evaluation of a condition is in
progress---a change which significantly simplifies for instance the
method of polynomial interpretations we shall employ in
Section~\ref{sec:polynomial}. Structure-preserving transformations
are discussed in Section~\ref{subsec:related work}.

Context-sensitive rewriting
restricts the positions in a term where rewriting is allowed.
A (C)TRS is combined with a
\emph{replacement map} $\mu$, which assigns to every $n$-ary
symbol $f \in \FF$ a subset $\mu(f) \subseteq \{ 1, \dots, n \}$.
A position $p$ is \emph{active} in a term $t$ if either
$p = \epsilon$, or $p = i\,q$, $t = f(\seq{t})$, $i \in \mu(f)$,
and $q$ is active in $t_i$. The set of active
positions in a term $t$ is denoted by $\Pos_\mu(t)$,
and $t$ may only be reduced at active positions.

\subsection{The Unconditional TRS $\Xi(\RR)$.}
\label{subsec:transformation:definition}

\begin{defi}
\label{def:contextsignature}
Let $\RR$ be a strong CCTRS over a signature $\FF$. For $f \in \FF$,
let $m_f$ be the number of rules in $\RR{\restriction}f$ (so $m_f = 0$
for constructor symbols $f$) and fix an order
$\RR{\restriction}f = \{ \rho^f_1, \dots, \rho^f_{m_f} \}$.
The context-sensitive signature $(\HH,\mu)$ is defined as follows:
\begin{itemize}
\item
$\HH$ contains two constants $\bot$ and $\top$,
\item
for every symbol $f \in \FF$ of arity $n$, $\HH$ contains
a symbol $f$ of arity $n + m_f$ with $\mu(f) = \{ 1, \dots, n \}$,
\item
for every defined symbol $f \in \FF_\DD$ of arity $n$, rule
$\smash{\rho^f_i}\colon \ell \to r \Leftarrow a_1 \approx b_1,
\dots,a_k \approx b_k$
in
$\RR{\restriction}f$, and $1 \leqslant j \leqslant k$, $\HH$
contains a symbol $f_i^j$ of arity $n + m_f + j - 1$ with
$\mu(f_i^j) = \{ n + i + j - 1 \}$.
\end{itemize}
\end{defi}

\noindent Terms in $\TT(\HH,\VV)$ that are involved in reducing
$f(\seq{s}) \in \TT(\FF,\VV)$ will have one of two forms:
$f(\seq{s},\seq[m_f]{t})$ with each $t_i \in \{ \top, \bot \}$,
indicating that rule $\rho_i^f$ has been attempted (and
failed) if and only if $t_i = \bot$, and
\[
f_i^j(\seq{s},\seq[i-1]{t},b_1\sigma,\dots,b_{j-1}\sigma,
u,t_{i+1},\dots,t_{m_f})
\]
indicating that rule $\rho_i^f$ is currently being evaluated and the
first $j-1$ conditions of $\rho_i^f$ have succeeded; $u$ records the
current progress on the condition $a_j \approx b_j$.

In the following we drop the superscript $f$ from $\rho_i^f$ if no
confusion arises.

\begin{defi}\label{def:xi}
The maps
$\xi_\star\colon \TT(\FF,\VV) \to \TT(\HH,\VV)$ with
$\star \in \{ \bot, \top \}$
are inductively defined as follows:
\[
\xi_\star(t) = \begin{cases}
t & \text{if $t$ is a variable,} \\
f(\xi_\star(t_1),\dots,\xi_\star(t_n)) & \text{if $t = f(\seq{t})$ and
$f$ is a constructor symbol,} \\
f(\xi_\star(t_1),\dots,\xi_\star(t_n),\star,\dots,\star) &
\text{if $t = f(\seq{t})$ and $f$ is a defined symbol.}
\end{cases}
\]
Linear terms in the set
$\{ \xi_\bot(t) \mid t \in \TT(\FF,\VV) \}$ are called
\emph{$\bot$-patterns}.
\end{defi}

In the transformed system that we will define, a ground term is in
normal form if and only if it is a $\bot$-pattern. This allows for
syntactic ``normal form'' tests. Most importantly, it allows for
purely syntactic anti-matching tests: If $s$ does not reduce to
an instance of some linear constructor term $t$, then
$s \to^* u\sigma$ for some substitution $\sigma$ and $\bot$-pattern $u$
that does not unify with $t$. What is more, we only need to consider a
finite number of $\bot$-patterns $u$.

\begin{defi}
\label{def:AP}
Let $t$ be a linear constructor term. The set of \emph{anti-patterns}
$\m{AP}(t)$ is inductively defined as follows.
If $t$ is a variable then $\m{AP}(t) = \varnothing$. If $t = f(\seq{t})$
then $\m{AP}(t)$ consists of the following $\bot$-patterns:
\begin{itemize}
\item
$g(\seq[m]{x})$ for every $m$-ary constructor symbol $g$ different from
$f$,
\item
$g(\seq[m]{x},\bot,\dots,\bot)$ for every defined symbol $g$ of arity
$m$ in $\FF$, and
\item
$f(\seq[i-1]{x},u,x_{i+1},\dots,x_n)$ for all
$1 \leqslant i \leqslant n$ and $u \in \m{AP}(t_i)$.
\end{itemize}
Here the $x_j$ are fresh and pairwise distinct variables.
\end{defi}

\begin{exa}
Consider the CCTRS of Example~\ref{ex:fibonacci}. The set
$\m{AP}(\langle z, w \rangle)$ consists of the $\bot$-patterns $\m{0}$,
$\m{s}(x)$, $\m{fib}(x,\bot,\bot)$, and $\m{+}(x,y,\bot,\bot)$.
\end{exa}

\begin{lem}
\label{AP lemma}
Let $s$ be a $\bot$-pattern and $t$ a linear constructor term
with $\Var(s) \cap \Var(t) = \varnothing$. If $s$ and $t$ are not
unifiable then $s$ is an instance of an anti-pattern in $\m{AP}(t)$.
\end{lem}

\begin{proof}
We use induction on the size of $t$.
If $s$ and $t$ are not unifiable, neither can be a variable. So
let $t = f(\seq{t})$. If $s = g(\seq{s})$ or
$s = {} g(\seq{s},\bot,\dots,\bot)$ for
some $g \neq f$ then $s$ instantiates $g(\seq{x})$ or
$g(\seq{x}, \bot,\dots,\bot)$ in $\m{AP}(t)$. Otherwise,
$s = f(\seq{s})$. If $s_i$ and $t_i$ are not unifiable for some $i$,
then by
the induction hypothesis $s_i$ is an instance of some
$u \in \m{AP}(t_i)$, so $s$ instantiates
$f(\seq[i-1]{x},u,x_{i+1},\dots,x_n) \in \m{AP}(t)$. If no such $i$
exists, there are substitutions $\seq{\sigma}$
such that $s_i\sigma_i = t_i\sigma_i$
for all $1 \leqslant i \leqslant n$.
Since $s$ and $t$ are linear terms without common variables, this
implies that $s$ and $t$ are unifiable by the substitution
$\sigma = \sigma_1 \cup \cdots \mathop{\cup} \sigma_n$,
contradicting the assumption.
\end{proof}

We are now ready to define the transformation
from a CCTRS $(\FF,\RR)$ to a context-sensitive TRS $(\HH,\mu,\Xi(\RR))$.
Here, we will use the notation
$\langle \seq{t} \rangle[\seq[j]{u}]_i$ to denote the sequence
$\seq[i-1]{t}, \seq[j]{u}, t_{i+1}, \dots, t_n$ and we
occasionally write $\vec{t}$ for a sequence $\seq{t}$.

\begin{defi}
\label{def:transformrules}
Let $\RR$ be a strong CCTRS over a signature $\FF$. The
TRS $\Xi(\RR)$ is defined over the context-sensitive signature
$(\HH,\mu)$ from Definition~\ref{def:contextsignature}
as follows. Let $\rho_i\colon f(\seq{\ell}) \to r \Leftarrow
a_1 \approx b_1, \dots, a_k \approx b_k$ be
the $i$-th rule in $\RR{\restriction}f$
(where $1 \leqslant i \leqslant m_f$).
\begin{itemize}
\item
If $k = 0$ then $\Xi(\RR)$ contains the rule
\begin{align}
f(\seq{\ell},\langle \seq[m_f]{x} \rangle[\top]_i)
&\to \xi_\top(r)
\tag{$1_\rho$}\label{xi1}
\end{align}
\item
If $k > 0$ then $\Xi(\RR)$ contains the rules
{\allowdisplaybreaks\begin{align}
f(\vec{\ell},\langle \seq[m_f]{x} \rangle[\top]_i)
&\to f_i^1(\vec{\ell},\langle \seq[m_f]{x} \rangle[\xi_\top(a_1)]_i)
\tag{$2_\rho$}\label{xi2}
\\
f_i^k(\vec{\ell},\langle \seq[m_f]{x} \rangle[\seq[k]{b}]_i)
&\to \xi_\top(r),
\tag{$3_\rho$}\label{xi3}
\intertext{the rules}
f_i^j(\vec{\ell},\langle \seq[m_f]{x} \rangle[\seq[j]{b}]_i)
&\to
\notag \\
\makebox[1cm][l]{$f_i^{j+1}(\vec{\ell},\langle \seq[m_f]{x}
\rangle[\seq[j]{b},\xi_\top(a_{j+1})]_i)$} &
\tag{$4_\rho$}\label{xi4}
\end{align}}
for all $1 \leqslant j < k$,
and the rules
\begin{align}
f_i^j(\vec{\ell},\langle \seq[m_f]{x} \rangle[\seq[j-1]{b},v]_i)
&\to
f(\vec{\ell},\langle \seq[m_f]{x} \rangle[\bot]_i)
\tag{$5_\rho$}\label{xi5}
\end{align}
for all $1 \leqslant j \leqslant k$ and $v \in \m{AP}(b_j)$
(where $\Var(v) \cap \Var(f(\vec{\ell},\vec{b},\vec{x})) =
\varnothing$).
\smallskip
\item
Regardless of $k$, $\Xi(\RR)$ contains the rules
\begin{align}
f(\langle \seq{y} \rangle[v]_j,\langle \seq[m_f]{x} \rangle[\top]_i)
&\to
f(\langle \seq{y} \rangle[v]_j,\langle \seq[m_f]{x} \rangle[\bot]_i)
\tag{$6_\rho$}\label{xi6}
\end{align}
for all $1 \leqslant j \leqslant n$ and $v \in \m{AP}(\ell_j)$
(where $\Var(v) \cap \Var(f(\vec{y},\vec{x})) = \varnothing$).
\end{itemize}
Here $\seq[m_f]{x}, \seq[n]{y}$ are fresh and pairwise distinct
variables.
A step using rule \eqref{xi1} or rule \eqref{xi3} has cost 1;
other rules---also called \emph{administrative rules}---have
cost 0.
\end{defi}

Rule~\eqref{xi1} simply adds the $\top$ labels to the right-hand
sides of unconditional rules. To apply a conditional rule $\rho_i$,
we mark the current function symbol as ``in progress for
$\rho_i$'' with rule~\eqref{xi2} and start evaluating the
first condition of $\rho_i$ by steps inside the argument for this
condition.
With rules~\eqref{xi4} we move to the next condition and, 
after all conditions have succeeded, an application of
rule~\eqref{xi3} results in the right-hand side with $\top$ labels.
If a condition fails \eqref{xi5} or 
the left-hand side of the rule does not match
and will never match \eqref{xi6}, then we replace the label
for $\rho_i$ by $\bot$, indicating that we do not need to try it again.

Note that the rules that do not produce the right-hand side of the
originating conditional rewrite rule are
considered \emph{administrative} and hence
do not contribute to the cost of a reduction.
The anti-pattern sets result in many rules \eqref{xi5} and
\eqref{xi6}, but all of these are simple. We could generalize the
system by replacing each $v \in \m{AP}(\ell_j)$
by a fresh variable; the complexity of the resulting (smaller) TRS
gives an upper bound for the original complexity. Indeed, all
methods proposed in Sections~\ref{sec:polynomial}--\ref{sec:splitsize},
also apply to the transformation using variables instead.
The primary purpose of anti-patterns is to ensure
\emph{completeness} (Theorem~\ref{thm:transformcomplete}); by
using anti-patterns instead of variables, we guarantee that a rule
is only marked as unsuccessful (by replacing its parameter by
$\bot$) if it truly cannot succeed anymore.

Note also that the resulting system $\Xi(\RR)$ is left-linear,
which is advantageous for
the potential applicability of various termination and complexity
techniques.

\begin{exa}
\label{ex:Xi(even/odd)}
The (context-sensitive) TRS $\Xi(\RR_{\m{even}})$ consists of the
rules below, with the numbers in square brackets indicating the
cost of the rule: 0 for administrative rules and 1 for the others.
{\allowdisplaybreaks
\begin{alignat*}{2}
[1] &&
\m{even}(\m{0},\top,y,z) &\to \m{true}
\tag{$1_1$}\label{ev11} \\
[0] &&
\m{even}(\star_1,\top,y,z) &\to \m{even}(\star_1,\bot,y,z)
\tag{$6_1$}\label{ev61} \\
[0] &&
\m{even}(\m{s}(x),y,\top,z) &\to
\m{even_2^1}(\m{s}(x),y,\m{odd}(x,\top,\top,\top),z)
\tag{$2_2$}\label{ev22} \\
[1] &&
\m{even_2^1}(\m{s}(x),y,\m{true},z) &\to \m{true}
\tag{$3_2$}\label{ev32} \\
[0] &&
\m{even_2^1}(\m{s}(x),y,\star_2,z) &\to \m{even}(\m{s}(x),y,\bot,z)
\tag{$5_2$}\label{ev52} \\
[0] &&
\m{even}(\star_3,y,\top,z) &\to \m{even}(\star_3,y,\bot,z)
\tag{$6_2$}\label{ev62} \\
[0] &&
\m{even}(\m{s}(x),y,z,\top) &\to
\m{even_3^1}(\m{s}(x),y,z,\m{even}(x,\top,\top,\top))
\tag{$2_3$}\label{ev23} \\
[1] &\qquad&
\m{even_3^1}(\m{s}(x),y,z,\m{true}) &\to \m{false}
\tag{$3_3$}\label{ev33} \\
[0] &&
\m{even_3^1}(\m{s}(x),y,z,\star_2) &\to \m{even}(\m{s}(x),y,z,\bot)
\tag{$5_3$}\label{ev53} \\
[0] &&
\m{even}(\star_3,y,z,\top) &\to \m{even}(\star_3,y,z,\bot)
\tag{$6_3$}\label{ev63} \\
[1] &&
\m{odd}(0,\top,y,z) &\to \m{false}
\tag{$1_4$}\label{ev14} \\
[0] &&
\m{odd}(\star_1,\top,y,z) &\to \m{odd}(\star_1,\bot,y,z)
\tag{$6_4$}\label{ev64} \\
[0] &&
\m{odd}(\m{s}(x),y,\top,z) &\to
\m{odd_2^1}(\m{s}(x),y,\m{odd}(x,\top,\top,\top),z)
\tag{$2_5$}\label{ev25} \\
[1] &&
\m{odd_2^1}(\m{s}(x),y,\m{true},z) &\to \m{false}
\tag{$3_5$}\label{ev35} \\
[0] &&
\m{odd_2^1}(\m{s}(x),y,\star_2,z) &\to \m{odd}(\m{s}(x),y,\bot,z)
\tag{$5_5$}\label{ev55} \\
[0] &&
\m{odd}(\star_3,y,\top,z) &\to \m{odd}(\star_3,y,\bot,z)
\tag{$6_5$}\label{ev65} \\
[0] &&
\m{odd}(\m{s}(x),y,z,\top) &\to
\m{odd_3^1}(\m{s}(x),y,z,\m{even}(x,\top,\top,\top))
\tag{$2_6$}\label{ev26} \\
[1] &&
\m{odd_3^1}(\m{s}(x),y,z,\m{true}) &\to \m{true}
\tag{$3_6$}\label{ev36} \\
[0] &&
\m{odd_3^1}(\m{s}(x),y,z,\star_2) &\to \m{odd}(\m{s}(x),y,z,\bot)
\tag{$5_6$}\label{ev56} \\
[0] &&
\m{odd}(\star_3,y,z,\top) &\to \m{odd}(\star_3,y,z,\bot)
\tag{$6_6$\label{ev66}}
\end{alignat*}}
for all
\begin{align*}
\star_1 &\in \m{AP}(0) = \{ \m{true}, \m{false}, \m{s}(x),
\m{even}(x,\bot,\bot,\bot), \m{odd}(x,\bot,\bot,\bot) \}
\\
\star_2 &\in \m{AP}(\m{true}) = \{ \m{false}, 0, \m{s}(x),
\m{even}(x,\bot,\bot,\bot), \m{odd}(x,\bot,\bot,\bot) \}
\\
\star_3 &\in \m{AP}(\m{s}(x)) = \{ \m{true}, \m{false}, 0,
\m{even}(x,\bot,\bot,\bot), \m{odd}(x,\bot,\bot,\bot) \}
\end{align*}
Following Definition~\ref{def:contextsignature}, this TRS is
equipped  with the following replacement map $\mu$:
\begin{gather*}
\mu(\m{even}) = \mu(\m{odd}) = \{ 1 \}
\qquad
\mu(\m{even}_2^1) = \mu(\m{odd}_2^1) = \{ 3 \}
\qquad
\mu(\m{s}) = \{ 1 \}
\\
\mu(\m{even}_3^1) = \mu(\m{odd}_3^1) = \{ 4 \}
\qquad
\mu(\m{0}) = \mu(\m{false}) = \mu(\m{true}) = \varnothing
\end{gather*}
\end{exa}

Instead of the current rules, which pass along
the various $\ell_i$ and $b_j$ unmodified throughout condition
evaluation, we could have opted for a more fine-grained approach
where we pass on their \emph{variables}, and then only those which are
needed later on, similar to what is done in
the \emph{optimized} unraveling~\cite{O99}.
Doing so, the example above would for instance have rules
\begin{align}
\m{even}(\m{s}(x),y,\top,z) &\to
\m{even}_2^1(x,y,\m{odd}(x,\top,\top,\top),z)
\tag{\ref{ev22}}
\intertext{and}
\m{even}_2^1(x,y,\star_2,z) &\to \m{even}(\m{s}(x),y,\bot,z)
\tag{\ref{ev52}}
\end{align}
However, this would complicate the presentation for no easily
discernible gain.

In an early version of this work~\cite{KMS15}, we employed a
slightly different transformation in which the symbols $f^i_\rho$
were \emph{constructor symbols}, used in
subterms corresponding to the rule whose conditions they evaluated.
For instance, the above rules were rendered as
\begin{align}
\m{even}(\m{s}(x),y,\top,z)
&\to \m{even}_{\m{active}}
(x,y,\m{even}^1_2(\m{s}(x),\m{odd}(x,\top,\top,\top)),z)
\tag{\ref{ev22}}
\intertext{and}
\m{even}_{\m{active}}(\m{s}(x),y,\m{even}^1_2(u,\star_2),z)
&\to \m{even}(\m{s}(x),y,\bot,z)
\tag{\ref{ev52}}
\end{align}
We simplifed this to our current definition because it is easier to work
with when looking for interpretations to establish termination as in
Section~\ref{sec:polynomial}.

\begin{defi}
We define the derivation height of a terminating term $s$ in the
context-sensitive TRS $(\HH,\Xi(\RR))$ as the greatest number of
non-administrative steps in any reduction starting in $s$, taking
the replacement map into account:
\begin{align*}
\m{dh}(s) &= {}
\max\,(\{ \m{cost}(s \to_{\RR,\mu}^* t) \mid t \in \TT(\HH,\VV) \})
\intertext{Letting $\m{dh}(s) = \infty$ if $s$ is non-terminating,
the derivation and runtime complexities are defined
accordingly:}
\m{dc}_{\Xi(\RR)}(n) &= \max\,\{ \m{dh}(s) \mid
\text{$s \in \TT(\HH)$ and $|s| \leqslant n$} \}
\\
\m{rc}_{\Xi(\RR)}(n) &= \max\,\{ \m{dh}(s) \mid
\text{$s \in \TT(\HH)$, $|s| \leqslant n$, and $s$ is basic}
\}
\end{align*}
\end{defi}

\subsection{Labeled reduction versus $\Xi(\RR)$}
\label{subsec:transformation:proofs}

\mbox{}

\bigskip

\noindent
In order to use the translated TRS $\Xi(\RR)$, we must understand
how the conditional complexity of the original CCTRS relates to the
unconditional complexity of $\Xi(\RR)$. To this end, we will define
a translation $\zeta$ from \emph{labeled} terms to terms over $\HH$,
which has the following properties:
\begin{enumerate}
\item
if $s$ has (conditional) derivation height $N$ then
$\zeta(s)$ has (unconditional) derivation height at least $N$
(Theorem~\ref{thm:transformsound}),
\item
if $\zeta(s)$ has (unconditional) derivation height $N$
then $s$ has (conditional) derivation height at least $N$
(Theorem~\ref{thm:transformcomplete}).
\end{enumerate}
Thus, we will be able to use the transformed system $\Xi(\RR)$ to
obtain both upper and lower bounds for conditional complexity.

While $\to_{\Xi(\RR),\mu}$ and $\xrightharpoonup{}$ were designed to
be intuitively equivalent, the proofs are rather technical.
Before proving the first result, we define the
mapping $\zeta$ from terms in $\TT(\GG,\VV)$ to terms in $\TT(\HH,\VV)$.
It resembles the earlier
definition of $\xi_\star$, but also handles the labels.

\begin{defi}
For $t \in \TT(\GG,\VV)$ we define
\[
\zeta(t) = \begin{cases}
t & \text{if $t \in \VV$,} \\
f(\zeta(t_1),\dots,\zeta(t_n)) &
\text{if $t = f(\seq{t})$ with $f$ a constructor symbol,} \\
f(\zeta(t_1),\dots,\zeta(t_n),\seq[m_f]{c}) &
\text{if $t = f_R(\seq{t})$ with $R \subseteq \RR{\restriction}f$} \\
\end{cases}
\]
where $c_i = \top$ if $\rho_i$ belongs to $R$ and $c_i = \bot$
otherwise, for $1 \leqslant i \leqslant m_f$.
For a substitution $\sigma \in \Sigma(\GG,\VV)$ we denote
the substitution $\zeta \circ \sigma$ by $\sigma_\zeta$.
\end{defi}

It is easy to see that $p \in \Pos_\mu(\zeta(t))$
if and only if $p \in \Pos(t)$,
if and only if $p \in \Pos(\zeta(t))$ and
$\zeta(t)|_p \notin \{ \bot, \top \}$,
for any $t \in \TT(\GG,\VV)$.

\begin{lem}
\label{zeta xi}
If $t \in \TT(\FF,\VV)$ then $\zeta(\m{label}(t)) = \xi_\top(t)$.
If $t \in \TT(\GG,\VV)$ and $\sigma \in \Sigma(\GG,\VV)$ then
$\zeta(t\sigma) = \zeta(t)\sigma_\zeta$.
Moreover, if $t$ is a linear labeled normal form then 
$\zeta(t) = \xi_\bot(\m{erase}(t))$ is a $\bot$-pattern, and if
$\zeta(t)$ is a $\bot$-pattern then $t$ is a linear labeled normal form.
\end{lem}

\begin{proof}
All four properties are easily proved by induction on the size
of $t$.
\end{proof}

We are now ready for the first main result, which states
that $\Xi$ reflects complexity.

\begin{thm}
\label{thm:transformsound}
Let $\RR$ be a strong CCTRS.
\begin{enumerate}
\item
If $s \xrightharpoonup{}^* t$ is a
complexity-conscious reduction with cost $N$ then
there exists a context-sensitive reduction
$\zeta(s) \to_{\Xi(\RR),\mu}^* \zeta(t)$ with cost $N$.
\item
If $s \xrightharpoonup{\infty\,}$ then there is an
infinite $(\Xi(\RR),\mu)$ reduction starting from $\zeta(s)$.
\end{enumerate}
\end{thm}

\begin{proof}
We prove the first statement by induction on the number of steps
in $s \xrightharpoonup{}^* t$.
The result is obvious when this number is zero, so suppose
$s \xrightharpoonup{} u \xrightharpoonup{}^* t$ and let $M$ be the
cost of the step $s \xrightharpoonup{} u$
and $N - M$ the cost of $u \xrightharpoonup{}^* t$. The induction
hypothesis yields a context-sensitive reduction
$\zeta(u) \to_{\Xi(\RR),\mu}^* \zeta(t)$ of cost $N - M$ and so
it remains to show that there exists a context-sensitive reduction
$\zeta(s) \to_{\Xi(\RR),\mu}^* \zeta(u)$ of cost $M$. Let
$\rho\colon f(\seq{\ell}) \to r \Leftarrow c$ be the rule in
$\RR$ that gives rise to the step
$s \xrightharpoonup{} u$ and let $i$ be its index in $\RR{\restriction}f$. 
There exist a position $p \in \Pos(s)$, terms $\seq{s}$, and
a subset $R \subseteq \RR{\restriction}f$ such that
$s|_p = f_R(\seq{s})$ and $\rho \in R$. We have
$\zeta(s)|_p = \zeta(s|_p) =
f_R(\zeta(s_1),\dots,\zeta(s_n),\seq[m_f]{c})$
where $c_j = \top$ if the $j$-th rule of $\RR{\restriction}f$ belongs
to $R$ and $c_j = \bot$ otherwise, for $1 \leqslant j \leqslant m_f$.
In particular, $c_i = \top$.
Note that $p$ is an active position in $\zeta(s)$.
We distinguish three cases.
\begin{itemize}
\item
First suppose that
$s \xrightharpoonup{\bot\,} u$. So $M = 0$,
$u = s[f_{R \setminus \{ \rho \}}(\seq{s})]_p$, and---by linearity
of $f(\seq{\ell})$---there
exist a linear labeled normal form $v$, a substitution $\sigma$,
and an index $1 \leqslant j \leqslant n$ such that $s_j = v\sigma$ and
$\m{erase}(v)$ does not unify with $\ell_j$.
By Lemma~\ref{zeta xi}, $\zeta(s_j) = \zeta(v\sigma) =
\zeta(v)\sigma_\zeta = \xi_\bot(\m{erase}(v))\sigma_\zeta$.
By definition, $\xi_\bot(\m{erase}(v))$ is a $\bot$-pattern, which
cannot unify with $\ell_j$ because $\m{erase}(v)$ does not.
From Lemma~\ref{AP lemma} we obtain an anti-pattern $v' \in
\m{AP}(\ell_j)$ such that $\xi_\bot(\m{erase}(v))$ is an instance of
$v'$.
Hence
$\zeta(s) = \zeta(s)[f(\zeta(s_1),\dots,\zeta(s_n),\seq[m_f]{c})]_p$
with $\zeta(s_j)$ an instance of $v' \in \m{AP}(\ell_j)$ and
$c_i = \top$. Consequently, $\zeta(s)$ reduces to
$\zeta(s)[f(\zeta(s_1),\dots,\zeta(s_n),
\langle \seq[m_f]{c} \rangle[\bot]_i)]_p$
by an application of rule \eqref{xi6}, which has cost zero. The latter
term equals
\[
\zeta(s[f_{R \setminus \{ \rho \}}(\seq{s})]_p) = \zeta(u)
\]
and hence we are done.
\smallskip
\item
Next suppose that $s \xrightharpoonup{} u$ is a successful step. So there
exists a substitution $\sigma$ such that
$\m{label}(a_i)\sigma \xrightharpoonup{}^* b_i\sigma$ with cost $M_i$
for all $1 \leqslant i \leqslant k$, and $M = 1 + M_1 + \cdots + M_k$.
The induction hypothesis yields reductions
$\zeta(\m{label}(a_i)\sigma) \to_{\Xi(\RR),\mu}^* \zeta(b_i\sigma)$
with cost $M_i$.
By Lemma~\ref{zeta xi}, $\zeta(\m{label}(a_i)\sigma) =
\zeta(\m{label}(a_i))\sigma_\zeta = \xi_\top(a_i)\sigma_\zeta$
and $\zeta(b_i\sigma) = b_i\sigma_\zeta$.
Moreover, $\zeta(s)|_p = \zeta(s|_p) =
f(\vec{\ell},\langle \seq[m_f]{c} \rangle[\top]_i)\sigma_\zeta$ and
$\zeta(u) = \zeta(s)[\zeta(\m{label}(r)\sigma)]_p$ with
$\zeta(\m{label}(r))\sigma_\zeta = \xi_\top(r)\sigma_\zeta$ by
Lemma~\ref{zeta xi}.
So it suffices if
$f(\vec{\ell},\langle \seq[m_f]{c} \rangle[\top]_i)\sigma_\zeta
\to_{\Xi(\RR),\mu}^* \xi_\top(r)\sigma_\zeta$
with cost $M$. If $k = 0$, we can use rule \eqref{xi1}.
Otherwise, we use the
reductions $\xi_\top(a_i)\sigma_\zeta \to_{\Xi(\RR),\mu}^*
b_i\sigma_\zeta$, rules \eqref{xi2} and \eqref{xi3}, and $k-1$ times a
rule of type \eqref{xi4} to obtain
\begin{align*}
f(\vec{\ell},\langle\seq[m_f]{c}\rangle[\top]_i)\sigma_\zeta
&\,\to_{\Xi(\RR),\mu}\,
f_i^1(\vec{\ell},\langle\seq[m_f]{c}\rangle
[\xi_\top(a_1)]_i)\sigma_\zeta
\\
&\,\to_{\Xi(\RR),\mu}^*\,
f_i^1(\vec{\ell},\langle\seq[m_f]{c}\rangle
[b_1]_i)\sigma_\zeta
\\
&\,\to_{\Xi(\RR),\mu}\,
f_i^2(\vec{\ell},\langle\seq[m_f]{c}\rangle
[b_1,\xi_\top(a_2)]_i)\sigma_\zeta
\\
&\,\to_{\Xi(\RR),\mu}^*\,
\cdots
\\
&\,\to_{\Xi(\RR),\mu}\,
f_i^k(\vec{\ell},\langle\seq[m_f]{c}\rangle
[b_1,\dots,b_k]_i)\sigma_\zeta
\\
&\,\to_{\Xi(\RR),\mu}\,
\xi_\top(r)\sigma_\zeta
\end{align*}
Note that all steps take place at active positions,
and that the steps with rules \ref{xi2} and \ref{xi4} are
administrative. Therefore, the cost of this reduction equals $M$.
\smallskip
\item
The remaining case is a failed step $s \xrightharpoonup{} u$. So there
exist substitutions $\sigma$ and $\tau$, an index
$1 \leqslant j < k$, and a linear labeled normal form $v$ which
does not unify with $b_{j+1}$ such that
$\m{label}(a_i)\sigma \xrightharpoonup{}^* b_i\sigma$ with cost $M_i$
for all $1 \leqslant i \leqslant j$ and
$\m{label}(a_{j+1})\sigma \xrightharpoonup{}^* v\tau$ with cost $M_{j+1}$.
We obtain
$\zeta(\m{label}(a_i)\sigma) = \xi_\top(a_i)\sigma_\zeta$,
$\zeta(b_i\sigma) = b_i\sigma_\zeta$, and $\zeta(s)|_p =
f(\vec{\ell},\langle \seq[m_f]{c} \rangle[\top]_i)\sigma_\zeta$ like
in the preceding case. Moreover, like in the first case, we obtain
an anti-pattern $v' \in \m{AP}(b_{j+1})$ such that
$\xi_\bot(\m{erase}(v))$ is an instance of $v'$. We have
$\zeta(v\tau) = \zeta(v)\tau_\zeta = \xi_\bot(\m{erase}(v))\tau_\zeta$ by
Lemma~\ref{zeta xi}. Hence
$\zeta(v\tau)$ is an instance of $v'$. Consequently,
\begin{align*}
f(\vec{\ell},\langle\seq[m_f]{c}\rangle[\top]_i)\sigma_\zeta
&\,\to_{\Xi(\RR),\mu}^*\,
f_i^{j+1}(\vec{\ell},\langle\seq[m_f]{c}\rangle
[b_1,\dots,b_j,\xi_\top(a_{j+1})]_i)\sigma_\zeta
\\
&\,\to_{\Xi(\RR),\mu}^*\,
f_i^{j+1}(\vec{\ell},\langle\seq[m_f]{c}\rangle
[b_1,\dots,b_j,\zeta(v\tau)]_i)\sigma_\zeta
\\
&\,\to_{\Xi(\RR),\mu}\,
f(\vec{\ell},\langle\seq[m_f]{c}\rangle[\bot]_i)\sigma_\zeta
\end{align*}\enlargethispage{2\baselineskip}%
where the last step uses an administrative rule of type \eqref{xi5}.
Again, all steps take place at active positions.
Note that
$f(\vec{\ell},\langle\seq[m_f]{c}\rangle[\bot]_i)\sigma_\zeta =
\zeta(f_{R \setminus \{ \rho \}}(\seq{s})) = \zeta(u|_p)$.
Hence $\zeta(s) \to_{\Xi(\RR),\mu}^* \zeta(u)$ as desired.
The cost of this reduction is $M_1 + \cdots + M_{j+1}$, which
coincides with the cost $M$ of the step $s \xrightharpoonup{} u$.
\end{itemize}
This concludes the proof of the first statement.
As for the second statement, suppose $s \xrightharpoonup{\infty\,}$, so
there exists an infinite sequence $(s)_{i \geqslant 0}$ of terms such
that $s = s_0$ and $s_i \xrightharpoonup{} s_{i+1}$ or
$s_i \xrightharpoonup{\rhd\,} s_{i+1}$ for all $i \geqslant 0$.
Fix $i \geqslant 0$. If $s_i \xrightharpoonup{} s_{i+1}$ then
$\smash{\zeta(s_i) \to^+_{\Xi(\RR),\mu} \zeta(s_{i+1})}$ follows from
the first statement.
Suppose $s_i \xrightharpoonup{\rhd\,} s_{i+1}$. We show that
$\zeta(s_i) \to^+_{\Xi(\RR),\mu} C[\zeta(s_{i+1})]$ for some context $C$
whose hole is at an active position.
there exist an active position $p \in \Pos(s_i)$, a rule
$\rho\colon f(\seq{\ell}) \to r \Leftarrow c$ in $R$,
a substitution $\sigma$, and an index $j$ such that
$s_i|_p = f_R(\ell_1\sigma,\dots,\ell_n\sigma)$,
$\m{label}(a_1)\sigma \xrightharpoonup{}^* b_1\sigma$, $\dots$,
$\m{label}(a_j)\sigma \xrightharpoonup{}^* b_j\sigma$, and $s_{i+1} =
\m{label}(a_{j+1})\sigma$, so
$\zeta(s_{i+1}) = \xi_\top(a_{j+1})\sigma_\zeta$.
Let $l$ be the index of $\rho$ in $\RR{\restriction}f$.
We obtain
$\zeta(s_i)|_p = \zeta(s_i|p) =
f(\ell_1\sigma_\zeta,\dots,\ell_n\sigma_\zeta,\seq[m_f]{c})$ where
$c_l = \top$, and $\xi_\top(a_d)\sigma_\zeta =
\zeta(\m{label}(a_d)\sigma) \to^*_{\Xi(\RR),\mu} \zeta(b_d\sigma) =
b_d\sigma_\zeta$ for $1 \leqslant d \leqslant j$,
by the first statement. Hence
\[
s_i =
s_i[f(\vec{\ell},\langle\seq[m_f]{c}\rangle[\top]_j)\sigma_\zeta]_p
\,\to_{\Xi(\RR),\mu}^*\,
s_i[f_l^{j+1}(\vec{\ell},\langle\seq[m_f]{c}\rangle
[\seq[j]{b},\xi_\top(a_{j+1})]_l)\sigma_\zeta]_p
\]
and thus we can take the context
\[
C = {} s_i[f_l^{j+1}(\ell_1\sigma,\dots,\ell_n\sigma,
c_1,\dots,c_{l-1},b_1,\dots,b_j,\Box,c_{l+1},\dots,c_{m_f})]_p
\]
The hole is at an active position, since $p$ is active
in $s_i$ and $n + l + j$ in $\mu(f_l^{j+1})$.
\end{proof}

Theorem~\ref{thm:transformsound} provides a way to establish
conditional complexity: If $\Xi(\RR)$ has complexity $\OO(\varphi(n))$
then the conditional complexity of $\RR$ is at $\OO(\varphi(n))$.
This is the important direction as it allows us to obtain an upper
bound for complexity by transforming the conditional system into an
unconditional one. However, we have more. The following result
shows that complexity bounds thus obtained can be \emph{sharp}.

\begin{thm}
\label{thm:transformcomplete}
Let $\RR$ be a strong CCTRS and $s \in \TT(\GG)$.
\begin{enumerate}
\item
If $\zeta(s)$ is terminating and there exists a
context-sensitive reduction $\zeta(s) \to_{\Xi(\RR),\mu}^* t$ for
some $t$ with cost $N$, then there exists a complexity-conscious
reduction $s \xrightharpoonup{}^* t'$
for some $t'$ with cost at least $N$.
\item
If there exists an infinite $(\Xi(\RR),\mu)$ reduction
starting from $\zeta(s)$ then $s \xrightharpoonup{\infty\,}$.
\end{enumerate}
\end{thm}

\begin{proof}[Proof Idea]
First of all, we may safely assume that $t$ is in normal form;
if it is not, we
simply extend the reduction (which can only increase the cost). Due
to the context-sensitivity restrictions and the form of the rules
$\Xi(\RR)$, any such normal form $t$ must be a $\bot$-pattern.

\smallskip

Next we transform the reduction $\zeta(s) \to_{\Xi(\RR),\mu}^* t$
(resp.\ $\zeta(s) \to_{\Xi(\RR)} \dots$) to a reduction with at least
the same cost (resp.\ an infinite reduction) which is well-behaved in
the sense that for any rule application
$u[\ell\sigma]_p \to u[r\sigma]_p$, the substitution $\sigma$ can be
written as $\zeta \circ \tau$.
This is done by a reordering argument, either postponing steps in
subterms (if the result of the step is used later), or eagerly
evaluating the corresponding subterm to normal form.

\smallskip

Having a well-behaved reduction, steps using rules \eqref{xi1} can
be translated directly to unconditional $\xrightharpoonup{}$ steps,
and \eqref{xi6} translates to $\smash{\xrightharpoonup{\bot\,}}$. Combined
steps \eqref{xi2} followed by some \eqref{xi4} applications and ending
with \eqref{xi3} or \eqref{xi5} correspond to successful or failed
applications; the restrictions of context-sensitivity guarantee that
any reduction steps in between these rule applications are either at
independent positions---in which case they can be postponed---or
inside the argument for the condition in progress. Since $t$ is
assumed to be in normal form, all such combinations are either
completed---in which case they can be transformed---or give rise to
an infinite reduction inside the accessible argument of a $f_i^j$
symbol---in which case we can reduce with a
$\smash{\xrightharpoonup{\rhd\,}}$
step to a non-terminating term $\zeta(a_j)$. Either way we are done.
\end{proof}

We refer to Appendix~\ref{app:transformcomplete} for the
full and rather intricate proof.
Note that Theorems~\ref{thm:transformsound}
and~\ref{thm:transformcomplete} together with Lemma~\ref{zeta xi}
tell us that for terms $s$ in $\TT(\FF,\VV)$, the ``conditional
complexity cost'' of $\m{label}(s)$ is the same as the derivation
height of $\xi_\top(s)$. Consequently, complexity notions between
the original CTRS and the resulting context-sensitive TRS are
interchangeable, but only so long as we limit interest to starting
terms where the additional $m_f$ arguments of every defined symbol
$f$ are set to $\top$:
\begin{align*}
\m{cdc}(n) &= \max\,\{ \m{dh}(\xi_\top(t)) \mid |t| \leqslant n \} \\
\m{crc}(n) &= \max\,\{ \m{dh}(\xi_\top(t)) \mid \text{$|t| \leqslant n$
and $t$ is basic} \}
\end{align*}
What is more, we have gained an additional result: The transformation
does not merely relate \emph{complexity} notions, but conservatively
translates \emph{quasi-decreasingness} to termination.

\begin{cor}
A strong CCTRS $\RR$ is quasi-decreasing if and only if the
corresponding
context-sensitive TRS $\Xi(\RR)$ is terminating on all terms in the
set $\{ \zeta(s) \mid s \in \TT(\GG,\VV) \}$.
\end{cor}

Thus, we can use the same transformation to prove
quasi-decreasingness of CCTRSs.
Although there are no complexity tools yet which take
context-sensitivity into account, we can obtain an upper bound
by simply ignoring the replacement map. Similarly, although existing
tools do not accommodate administrative rules
we can count all rule applications equally.
Since for every non-administrative step reducing a term
$f_R(\cdots)$ at a position $p$, at most (number of rules)
$\times$ (greatest number of conditions + 1)
administrative steps at position $p$
can be done, the difference is only a constant factor.
Moreover, these rules are an instance of \emph{relative rewriting}, for
which advanced complexity methods do exist. Thus, it
is likely that there will be direct tool support in the future.

\section{Interpretations in $\N$}
\label{sec:polynomial}

A common method to derive complexity bounds for a TRS is the use of
interpretations in $\N$. Such an interpretation $\I$ maps function
symbols of arity $n$ to functions from $\N^n$ to $\N$, giving a value
$\interpret{t}{}$ for every ground term $t$, which is shown to
decrease in each reduction step. The method is easily adapted to
support context-sensitive rewriting and administrative rules.

As we will consider interpretations on different domains later on,
we define interpretations in a general way. Let $\A$ be a set (such
as $\N$) and let $>$ be a well-founded order on this set, and
$\geqslant$ a quasi-order compatible with $>$ (i.e.,
${\geqslant \cdot >} \subseteq {>}$ and
${> \cdot \geqslant} \subseteq {>}$).
A function $f$ from $\A^n$ to $\A$ is
\emph{strictly monotone} in its $i$-th argument, if
$f(s_1,\dots,s_i,\dots,s_n) > f(s_1,\dots,s_i',\dots,s_n)$ whenever
$s_i > s_i'$ and \emph{weakly monotone} in its $i$-th argument,
provided that
$f(s_1,\dots,s_i,\dots,s_n) \geqslant f(s_1,\dots,s_i',\dots,s_n)$
whenever $s_i \geqslant s_i'$.

\begin{defi}
\label{def:interpretation}
A context-sensitive interpretation over $\A$ is a function $\I$
mapping each symbol $f \in \FF$ of arity $n$ to a function $\I_f$
from $\A^n$ to $\A$, such that $\I_f$ is strictly monotone
in its $i$-th argument for all $i \in \mu(f)$.
Given a valuation $\alpha$ mapping each variable to an element of
$\A$, the value $\interpret{t}{\alpha} \in \A$ of a term
$t$ is defined as usual:
\begin{itemize}
\item
$\interpret{x}{\alpha} = \alpha(x)$ for $x \in \VV$,
\smallskip
\item
$\interpret{f(\seq{s})}{\alpha} =
\I_f(\interpret{s_1}{\alpha}, \dots, \interpret{s_n}{\alpha})$
for $f \in \FF$.
\end{itemize}
We say $\I$ is compatible with a set of unconditional rules $\RR$ if
for all rules $\ell \to r \in \RR$ and valuations $\alpha$, 
$\interpret{\ell}{\alpha} > \interpret{r}{\alpha}$ if $\ell \to r \in
\RR$ is non-administrative and
$\interpret{\ell}{\alpha} \geqslant \interpret{r}{\alpha}$ otherwise.
\end{defi}

We easily see that if $s \to_{\RR,\mu} t$ then
$\interpret{s}{\alpha} \geqslant \interpret{t}{\alpha}$, and
$\interpret{s}{\alpha} > \interpret{t}{\alpha}$ if the employed
rule is non-administrative. Consequently, if $\A = \N$, then
$\m{dh}(s,\to_{\RR,\mu}) \leqslant \interpret{s}{\alpha}$
for any valuation $\alpha$.
Having a derivation height for all terms, we can obtain the
derivational and runtime complexity of the original system. To take
advantage of the fact that we only need to consider terms
$\xi_\top(s)$, we can limit interest to ``$\top$-terms'': ground terms
which have the property that
$t_1 = \dots = t_{m_f} = \top$ and $\seq{s} \notin \{ \bot, \top \}$
for all subterms
$f(\seq{s},\seq[m_f]{t})$.
For runtime
complexity, we only have to consider basic $\top$-terms. We let
$|s|$ denote the number of function symbols in $s$
\emph{not counting $\top$}. Then $|s| = |\xi_\top(s)|$.

\begin{exa}
\label{ex:evenpolynomial}
Continuing Example~\ref{ex:Xi(even/odd)}, we define the following
interpretation over $\N$:
\begin{align*}
\I_\top &= 1
\qquad
\I_\bot = \I_{\m{true}} = \I_{\m{false}} = \I_{\m{0}} = 0
\qquad
\I_{\m{s}}(x) = x+1
\\
\I_{\m{even}}(x,u,v,w) &= \I_{\m{odd}}(x,u,v,w) = 1 + x + v \cdot 3^x
  + w \cdot 3^x
\\
\I_{\m{even}_2^1}(x,u,v,w) &= \I_{\m{odd}_2^1}(x,u,v,w) = 1 + x +
  v + w \cdot 3^x
\\
\I_{\m{even}_3^1}(x,u,v,w) &= \I_{\m{odd}_3^1}(x,u,v,w) = 1 + x +
  v \cdot 3^x + w
\end{align*}
One easily checks that $\I$ satisfies the required monotonicity
constraints: $\I_{\m{s}}$ is monotone in its only argument,
$\I_{\m{even}}$ and $\I_{\m{odd}}$ are monotone in $x$, while
$\I_{\m{even}_2^1}$, $\I_{\m{odd}_2^1}$ are monotone in $v$ and
$\I_{\m{even}_3^1}$, $\I_{\m{odd}_3^1}$ in $w$.
Moreover, all rules in $\Xi(\RR_{\m{even}})$ are oriented as required.
For example, the rules generated by the unconditional rule
\eqref{eo1} give the following obligations:
\begin{itemize}[label=($1_\rho$)]
\item[($1_\rho$)]
$\interpret{\m{even}(\m{0},\top,y,z)}{} =
  1 + 0 + y \cdot 3^0 + z \cdot 3^0 > 0 = \interpret{\m{true}}{}$,
\smallskip
\item[($6_\rho$)]
$\interpret{\m{even}(\star_1,\top,y,z)}{} =
  1 + \varphi + y \cdot 3^\varphi + z \cdot 3^\varphi \geqslant
  1 + \varphi + y \cdot 3^\varphi + z \cdot 3^\varphi =
  \interpret{\m{even}(\star_1,\bot,y,z)}{}$ with
$\varphi = \interpret{\star_1}{}$.
\end{itemize}
The rules corresponding to the unconditional $\m{odd}$ rule
\eqref{eo4} give
the same inequalities. As for the other four rules, their
translations and interpretations are all very similar, so we will
show only the interpretations and proof obligations for rule
\eqref{eo2}:
\begin{itemize}[label=($1_\rho$)]
\item[($2_\rho$)]
$\interpret{\m{even}(\m{s}(x),y,\top,z)}{} =
  2 + x + 3^{x+1} + z \cdot 3^{x+1} \geqslant 2 + x + (1 + x + 2 \cdot
  3^x) + z \cdot 3^{x+1} = \interpret{\m{even}_2^1(\m{s}(x),y,
  \m{odd}(x,\top,\top,\top),z)}{}$,
which follows from
$3^{x+1} = 3^x + 2 \cdot 3^x \geqslant (1 + x) + 2 \cdot 3^x$,
\smallskip
\item[($3_\rho$)]
$\interpret{\m{even_2^1}(\m{s}(x),y,\m{true},z)}{}
  = 2 + x + 0 + z \cdot 3^{x+1} > 0 = \interpret{\m{true}}{}$,
\smallskip
\item[($5_\rho$)]
$\interpret{\m{even_2^1}(s(x),y,\star_2,z)}{} =
  2 + x + \varphi + z \cdot 3^{x+1} \geqslant 2 + x + 0 + z \cdot
  3^{x+1} = \interpret{\m{even_2^1}(s(x),y,\bot,z)}{}$
with $\varphi = \interpret{\star_2}{} \geqslant 0$,
\smallskip
\item[($6_\rho$)]
$\interpret{\m{even}(\star_3,y,\top,z)}{} =
  1 + \varphi + 3^\varphi + z \cdot 3^\varphi \geqslant
  1 + \varphi + z \cdot 3^\varphi =
\interpret{\m{even}(\star_3,y,\bot,z)}{}$.
\end{itemize}
\medskip
Now, towards runtime complexity, we observe that for all ground
constructor terms $s$ with $|s| \leqslant n$ we also have
$\interpret{s}{} \leqslant n$
as $\I_f(\seq[m]{x}) \leqslant x_1 + \dots + x_m + 1$ for all
constructor symbols $f$. Therefore, the conditional runtime
complexity $\m{crc}_{\RR_{\m{even}}}(n)$ is bounded by $\OO(3^n)$:
\begin{align*}
&\max( \{ \interpret{f(\seq[m]{s},\top,\dots,\top)}{} \mid
\text{$f \in \FF_\DD$ and $\seq[m]{s}$ are ground constructor} \\[-.5ex]
&\phantom{\max( \{ \interpret{f(\seq[m]{s},\top,\dots,\top)}{} \mid {}}
\text{terms with $|s_1| + \dots + |s_m| < n$} \} ) \\
{} \leqslant {} &\max( \{ \I_f(\seq[m]{x},1,\dots,1) \mid
\text{$f \in \FF_\DD$ and $x_1 + x_2 + x_3 + x_4 < n$} \} ) \\
{} = {} &\max( \{ 1 + x + 2 \cdot 3^x \mid x < n \} ) =
n + 2 \cdot 3^{n-1} \leqslant 3^n \quad \text{for $n \geqslant 1$}
\end{align*}
As to derivational complexity, we observe that
$\interpret{t}{} \leqslant {}^n3$ (tetration,%
\footnote{Tetration is the next hyperoperation after exponentiation,
defined as iterated exponentiation.}
or $3 \mathop{\uparrow\uparrow} n$ in Knuth's up-arrow notation)
when $t$ is an arbitrary ground $\top$-term of size $n$. 
\end{exa}

To obtain a more elementary bound we will need more sophisticated
methods, for instance assigning a compatible sort system and using
the fact that all terms of sort $\m{int}$ are necessarily constructor
terms. A method based on separating size and space complexity is
discussed in Section~\ref{sec:splitsize}.

The interpretations in Example~\ref{ex:evenpolynomial} may appear somewhat
arbitrary, but in fact there is a recipe that we can most likely
apply to many TRSs obtained from CCTRSs using
Definitions~\ref{def:contextsignature} and~\ref{def:transformrules}.
The idea is to define the
interpretation $\I$ as an extension of a ``basic'' interpretation $\J$
over $\N$ with a fixed way of handling the additional arguments.

\begin{defi}[Recipe \refstepcounter{recipe}\label{def:recipeA}%
\therecipe]
Given
\begin{itemize}
\item
a strictly monotone interpretation function
$\J_f^0\colon \N^n \to \N$
for every symbol $f$ of arity $n$ in the
original signature $\FF$,
\smallskip
\item
weakly monotone interpretation functions
$\J_f^1, \dots, \J_f^{m_f}\colon \N^n \to \N$
for every $f \in \FF_\DD$,
\item
interpretation functions
$\J_{f,i}^1, \dots, \J_{f,i}^k$
with $\J_{f,i}^j\colon \N^{n + j} \to \N$ that are strictly monotone in
their last argument position ($n + j$),
for each rule $\rho_i \in \RR{\restriction}f$ with $k > 0$
conditions,
\end{itemize}
\noindent
we construct an interpretation $\I$ for $\HH$
as follows:
$\I_\top = 1$ and $\I_\bot = 0$,
$
\I_f(\seq{x}) = \J_f^0(\seq{x})
$
for every $f \in \FF_\CC$ of arity $n$,
\[
\I_f(\seq{x},\seq[m_f]{c}) =
\J_f^0(\seq{x}) + \sum_{k = 1}^{m_f} c_k \cdot \J_f^k(\seq{x})
\]
for every $f \in \FF_\DD$ of arity $n$, and finally
\begin{gather*}
\I_{f_i^j}(\seq{x},\seq[i-1]{c},\seq[j]{y},c_{i+1},\dots,c_{m_f}) =
\\
\J_f^0(\seq{x})
+ \J_{f,i}^j(\seq{x},\seq[j]{y}) +
\sum_{k=1,\,k \neq i}^{m_f} c_k \cdot \J_f^k(\seq{x})
\end{gather*}
for every symbol $f_i^j$.
\end{defi}

Using the interpretation of Recipe~\ref{def:recipeA} for the rules in
Definition~\ref{def:transformrules},
the inequalities we obtain can be greatly simplified, and in many
cases removed.

\begin{defi}
\label{compatibility constraints}
The \emph{compatibility constraints} for $\J$ comprise the following
inequalities, for every rule
$\rho_i\colon f(\seq{\ell}) \to r \Leftarrow a_1 \approx b_1, \dots,
a_k \approx b_k$ in the \emph{original}
system $\RR$:
\begin{itemize}[label=($1_\rho$)]
\item[($1_\rho$)]
$\J_f^0(\overrightarrow{\interpret{\ell}{\alpha}}) +
\J_f^i(\overrightarrow{\interpret{\ell}{\alpha}})
> \interpret{\xi_\top(r)}{\alpha}$ \ if $ k = 0$,
\smallskip
\item[($2_\rho$)]
$\J_f^i(\overrightarrow{\interpret{\ell}{\alpha}})
\geqslant \J_{f,i}^1(\overrightarrow{\interpret{\ell}{\alpha}},
\interpret{\xi_\top(a_1)}{\alpha})$ \ if $k > 0$,
\smallskip
\item[($3_\rho$)]
$\J_f^0(\overrightarrow{\interpret{\ell}{\alpha}}) +
\J_{f,i}^k(\overrightarrow{\interpret{\ell}{\alpha}},
\interpret{b_1}{\alpha},\dots,\interpret{b_k}{\alpha})
> \interpret{\xi_\top(r)}{\alpha}$,
\smallskip
\item[($4_\rho$)]
$\J_{f,i}^j(\overrightarrow{\interpret{\ell}{\alpha}},
\interpret{b_1}{\alpha},\dots,\interpret{b_j}{\alpha})
\geqslant
\J_{f,i}^{j+1}(\overrightarrow{\interpret{\ell}{\alpha}},
\interpret{b_1}{\alpha},\dots,\interpret{b_j}{\alpha},
\interpret{\xi_\top(a_{j+1})}{\alpha})$ \ for $1 \leqslant j < k$.
\end{itemize}
Here $\overrightarrow{\interpret{\ell}{\alpha}}$ denotes the
sequence $\interpret{\ell_1}{\alpha},\dots,\interpret{\ell_n}{\alpha}$.
\noindent
\end{defi}

\begin{lem}
\label{lem:recipe1}
The interpretation $\I$ from Recipe~\ref{def:recipeA} is a
context-sensitive interpretation for $(\HH,\mu)$.
If its interpretation functions satisfy the compatibility
constraints then $\I$ is compatible with $\HH$, so
\begin{align*}
\m{cdc}_\RR(n) &= \max \{ \interpret{\xi_\top(t)}{} \mid
\text{$t \in \TT(\FF)$ and $|t| \leqslant n$} \} \\
\m{crc}_\RR(n) &= \max \{ \interpret{\xi_\top(t)}{} \mid
\text{$t \in \TT(\FF)$, $|t| \leqslant n$, and $t$ is basic} \}
\end{align*}
Moreover,
\[
\interpret{\xi_\top(f(\seq{t}))}{\alpha} =
\sum_{i=0}^{m_f} \J_f^i(\interpret{\xi_\top(t_1)}{\alpha},\dots,
\interpret{\xi_\top(t_n)}{\alpha})
\]
\end{lem}

\begin{proof}
It is not hard to see that $\I$ satisfies the monotonicity
requirements of Definition~\ref{def:interpretation}. Hence it is
a context-sensitive interpretation for $(\HH,\mu)$.
The statements on $\m{cdc}_\RR$ and $\m{crc}_\RR$ follow by
compatibility and the
observations at the end of Section~\ref{sec:transformation}, because
of the inequality
$\m{dh}(s,\to_{\Xi(\RR),\mu}) \leqslant \interpret{s}{}$.
The final equality claim is obtained by writing out definitions.
For the compatibility claim, note that rules obtained from clause
\eqref{xi6} are obviously oriented as $\interpret{\bot}{} = 0$.
Compatibility is also satisfied for rules obtained from clause
\eqref{xi5}, as
\[
\J_{f,i}^j(\seq{s},\seq[j]{t}) \geqslant
\interpret{\bot}{\alpha} \cdot \J_f^i(\seq{s}) = 0
\]
always holds. The requirements for the other rules follow from the
compatibility constraints, by expanding the inequality
($\interpret{\ell}{\alpha} \geqslant \interpret{r}{\alpha}$ or
$\interpret{\ell}{\alpha} > \interpret{r}{\alpha}$) and removing
unhelpful terms on the left.
For instance, rules obtained from \eqref{xi1} impose the inequality
\[
\J_f^0(\overrightarrow{\interpret{\ell}{\alpha}})
+ \J_f^i(\overrightarrow{\interpret{\ell}{\alpha}})
+ \sum_{k=1,k\neq i}^{m_f} x_k \cdot
\J_f^k(\overrightarrow{\interpret{\ell}{\alpha}})
> \interpret{\xi_\top(r)}{\alpha}
\]
which follows from clause ($1_\rho$) in
Definition~\ref{compatibility constraints}; we
omitted the summation because the $x_i$ do not appear on the right,
and could well be $0$.
\end{proof}

By the final part of Lemma~\ref{lem:recipe1}, which recursively
defines $\interpret{\xi_\top(f(\seq{t}))}{\alpha}$ purely in 
terms of $\J$, we can obtain bounds on derivation heights without ever
calculating $\xi_\top(t)$. Thus, we do not even need to consider the
labeled or translated systems.

\begin{exa}\label{exa:recipe1}
To demonstrate the use of the recipe, recall the CCTRS from Example
\ref{ex:poschoice}:
\begin{xalignat*}{2}
\m{f}(x) &\to x &
\m{g}(x) &\to \m{a} \,\Leftarrow\, x \approx \m{b}
\end{xalignat*}
The recipe gives the following proof obligations:
\begin{align*}
\J_{\m{f}}^0(x) + \J_{\m{f}}^1(x) &> x &
\J_{\m{g}}^1(x) &\geqslant \J_{\m{g},1}^1(x,x) \\
& & \J_{\m{g}}^0(x) + \J_{\m{g},1}^1(x,\J_{\m{b}}) &> \J_{\m{a}}
\end{align*}
Here, $\J_{\m{f}}^0$ and $\J_{\m{g}}^0$ must be strictly monotone in
their first argument, $\J_{\m{f}}^1$ and $\J_{\m{g}}^1$ weakly
monotone, and $\J_{\m{g},1}^1$ must be strictly monotone in its
\emph{second} argument. These monotonicity requirements are satisfied
by choosing
\begin{align*}
\J_{\m{a}} &= 0 &
\J_{\m{f}}^0(x) &= x &
\J_{\m{g}}^0(x) &= x &
\J_{\m{g},1}^1(x,y) &= y \\
\J_{\m{b}} &= 1 &
\J_{\m{f}}^1(x) &= 1 &
\J_{\m{g}}^1(x) &= x
\end{align*}
With these interpretations, the
proof obligations are simplified to
\begin{align*}
x + 1 &> x &
x &\geqslant x \\
& & x + 1 &> 0
\end{align*}
and obviously satisfied.
\end{exa}

In order to bound the derivational complexity in
Example~\ref{exa:recipe1}, we make the following general observation.

\begin{lem}
\label{general}
If for every symbol $h$ of arity $n$ in some strong CCTRS we have
\[
\J_h^0(\seq{x}) + \dots +
\J_h^{m_h}(\seq{x}) \leqslant K \cdot (x_1 + \dots + x_n) + M
\]
then
$\interpret{\xi_\top(s)}{} \leqslant M \cdot (K^0 + \dots + K^{|s|-1})$
for all ground terms $s$.
\end{lem}

Recall that $m_h = 0$ for constructor symbols, so the above
requirement is well-defined.

\begin{proof}
We use
induction on $|s|$. If $|s| = 1$ then $s$ is a constant and
\[
\interpret{\xi_\top(s)}{} = \J_s^0 + \dots + \J_s^{m_s} \leqslant
K \cdot 0 + M = M = M \cdot K^0
\]
If $|s| = m+1$ then
$s = h(\seq{t})$ with $|t_1| + \dots + |t_n| = m$ and
\begin{align*}
\interpret{\xi_\top(s)}{} &= \sum_{i = 0}^{m_f}
\J_h^i(\interpret{\xi_\top(t_1)}{},\dots,\interpret{\xi_\top(t_n)}{}) \\
&\leqslant K \cdot
(\interpret{\xi_\top(t_1)}{} + \dots + \interpret{\xi_\top(t_n)}{}) + M \\
&\leqslant K \cdot M \cdot (K^0 + \dots + K^m) + M \\
&= M \cdot (K^1 + \dots + K^{m+1}) + M \\
&= M \cdot (K^0 + \dots + K^{m+1})
\end{align*}\vspace{-31 pt}\enlargethispage{\baselineskip}

\end{proof}\smallskip

\noindent Since, for $K \geqslant 2$, we have $K^0 + \dots + K^m \leqslant K^{m+1}$,
a linear interpretation satisfying the premise of
Lemma~\ref{general}
gives $\m{cdc}_\RR(n) = \OO(K^n)$ by Lemma~\ref{lem:recipe1}.
With this understanding, we can complete the example.

\noindent
\textbf{Example~\ref{exa:recipe1} (continued).}\ \ 
We thus obtain an exponential $\OO(2^n)$ bound.
This may not seem like an impressive result, but in fact, this bound
is tight! Consider a term $\m{g}^n(\m{b})$. To evaluate this term
to normal form, we obtain a cost of $2^{n-1}$ if we simply
evaluate outside-in:
\begin{itemize}
\item
$\m{g}_{\{\eqref{fg2}\}}(\m{b}) \mathrel{\xrightharpoonup{}} \m{a}$
with cost $1 = 2^{1-1}$,
\smallskip
\item
$\m{g}_{\{\eqref{fg2}\}}(\m{g}_{\{\eqref{fg2}\}}^n(\m{b}))
\mathrel{\xrightharpoonup{}}
\m{g}_\varnothing(\m{g}_{\{\eqref{fg2}\}}^n(\m{b}))$ with cost
$2^{n-1}$ (the cost to reduce the left-hand side of the condition,
$\m{g}_{\{\eqref{fg2}\}}^n(\m{b})$, to normal form), and
$\m{g}_\varnothing(\m{g}_{\{\eqref{fg2}\}}^n(\m{b}))$ reduces to
normal form with cost $2^{n-1}$ (the cost to evaluate the subterm),
amounting to a total cost of $2^{n-1} + 2^{n-1} = 2^n$.
\end{itemize}
However, we do have
\begin{align*}
\m{dh}(\xi_\top(\m{f}^n(\m{g}(\m{f}^m(\m{a})))))
&\leqslant \interpret{\xi_\top(\m{f}^n(\m{g}(\m{f}^m(\m{a}))))}{} \\
&= \interpret{\xi_\top(\m{g}(\m{f}^m(\m{a})))}{} + n \\
&= 2 \cdot \interpret{\xi_\top(\m{f}^m(\m{a}))}{} + n \\
&= 2 \cdot m + n
\end{align*}
which gives the expected linear bound for the collection of
terms considered in Example~\ref{ex:poschoice}.

\section{Using Context-Sensitivity to Improve Runtime Complexity
Bounds}
\label{sec:ucs}

As observed before, the actual runtime complexity for the system
in Example~\ref{ex:evenpolynomial} is $\OO(2^n)$. In order to obtain this
more realistic bound, we will need more sophisticated methods than
simply polynomial interpretations. This is not a problem specific to
our transformed systems $\Xi(\RR)$; rather, giving tight
complexity bounds is a hard problem, which has been studied
extensively in the literature. Consequently, many different
complexity methods have been developed (e.g.\
matrix interpretations~\cite{MSW08,W10,MMNWZ11},
arctic interpretations~\cite{KW09},
polynomial path orders~\cite{AM08,AM13},
match bounds~\cite{GHWZ07}, dependency tuples~\cite{NEG13})
and it seems likely that most of these methods can easily be adapted to
context-sensitive and relative rewriting.

In order to demonstrate that the systems we obtain using our
transformation are not inherently problematic, we will show two
improvements which allow us to obtain better bounds. The first one,
which is treated in this section, employs a technique from~\cite{HM14}.

\begin{defi}
A replacement map $\NU$ is \emph{usable} for a strong CCTRS
$(\FF,\RR)$ if for every rewrite rule
$b_0 \to a_{k+1} \Leftarrow a_1 \approx b_1, \dots, a_k \approx b_k$ in
$\RR$ and all $1 \leqslant i \leqslant k+1$ and $p \in \Pos(a_i)$ we have
$p \in \Pos_{\NU}(a_i)$ if either
$p \in \Pos_{\FF_\DD}(a_i)$, or $p$ is a variable position in $a_i$
and there exist $0 \leqslant j < i$
and $q \in \Pos_{\NU}(b_j)$ such that $(a_i)|_p = (b_j)|_q$.
\end{defi}

Note that the requirement on $p \in \Pos_{\NU}(a_i)$
is a sufficient condition only;
it is allowed for $\Pos_{\NU}(a_i)$ to contain also $p$ which satisfy
neither premise. Therefore, the full replacement map, with
$\nu(f) = \{ 1, \dots, n \}$ for $f$ of arity $n$, is always usable.

\begin{exa}
\label{ex:fibonacci2}
We derive a usable replacement map $\NU$ for the CCTRS
$\RR_{\m{fib}}$ of Example~\ref{ex:fibonacci}:
\begin{xalignat*}{2}
\m{0} + y &\to y &
\m{fib}(\m{0}) &\to \langle \m{0},\m{s}(\m{0}) \rangle \\
\m{s}(x) + y &\to \m{s}(x + y) &
\m{fib}(\m{s}(x)) &\to \langle z, w \rangle
~\Leftarrow~ \m{fib}(x) \approx \langle y, z \rangle,~ y + z \approx w
\end{xalignat*}
From the rule $\m{s}(x) + y \to \m{s}(x + y)$ we obtain
$1 \in \NU(\m{s})$. The other constraints are obtained from the
conditional rule for $\m{fib}$. The variable $w$ appears at an active
(root) position in the right-hand side of a condition and also at
position $2$ in $\langle z, w \rangle$. Hence
we obtain $2 \in \NU(\langle \cdot , \cdot \rangle)$, which causes
the variable $z$ to appear at an active position in $\langle y, z \rangle$ 
and thus $2 \in \NU(+)$ and $1 \in \NU(\langle \cdot , \cdot \rangle)$.
The latter activates the variable $y$ in $\langle y, z \rangle$ and
thus we also need $1 \in \NU(+)$. There are no other demands and
hence the replacement map $\NU$ defined by $\NU(\m{s}) = \{ 1 \}$,
$\NU(\m{+}) = \NU(\langle \cdot , \cdot \rangle) = \{ 1, 2 \}$,
and $\NU(\m{fib}) = \varnothing$ is usable.
\end{exa}

\begin{defi}
Let $\NU$ be a usable replacement map for a strong CCTRS
$(\FF,\RR)$. Let
$\mu$ be the replacement map defined in
Definition~\ref{def:contextsignature} for the signature $\HH$. We
define a new replacement map $\mu\NU$ for $\HH$ as follows:
$\mu\NU(f) = \NU(f)$ for every $f \in \HH \cap \FF$ and
$\mu\NU(f) = \mu(f)$ for every $f \in \HH \setminus \FF$.
\end{defi}

\begin{thm}
\label{thm:mapreplace}
If $\NU$ is a usable replacement map for a strong CCTRS
$(\FF,\RR)$ then
$\m{crc}_{\RR}(n) \leqslant \m{rc}_{\Xi(\RR),\mu\NU}(n)$
for all $n \geqslant 0$.
\end{thm}

\begin{proof}
We define an intermediate replacement map $\NU'$ as follows:
$\NU'(f) = \NU(f)$ for every $f \in \HH \cap \FF$ and
$\NU'(f_i^j) = \NU(f) \cup \{ n+i, \dots, n+i+j-1 \}$ for every
$f_i^j \in \HH \setminus \FF$ such that the arity of $f$ in $\FF$
is $n$. It is not difficult to prove that
$\mu\NU(f) = \NU'(f) \cap \mu(f)$ for every $f \in \HH$.

We prove that $\Pos_{\HH_\DD}(t) \subseteq \Pos_{\NU'}(t)$ whenever
$s \to_{\Xi(\RR),\mu}^* t$ and $s$ is basic, by induction on
the length.
Since $\Pos_{\NU'}(t) \cap \Pos_\mu(t) = \Pos_{\mu\NU}(t)$,
this implies that any $(\Xi(\RR),\mu)$ reduction
sequence starting from a basic term is a reduction
sequence in $(\Xi(\RR),\mu\NU)$, and hence the statement of the
theorem follows from Theorem~\ref{thm:transformsound}.
The base case is obvious since
$\Pos_{\HH_\DD}(t) = \{ \epsilon \} \subseteq
\Pos_{\NU'}(s)$ if $t$ is basic. For the induction step we consider
\[
s \to_{\Xi(\RR),\mu}^* s' \to_{\Xi(\RR),\mu} t
\]
We obtain $\Pos_{\HH_\DD}(s') \subseteq \Pos_{\NU'}(s')$
from the induction hypothesis. Suppose the step from $s'$ to $t$ employs
the rule $u \to v$ from $\Xi(\RR)$ at position $p \in \Pos_\mu(s')$ with
substitution $\sigma$. We have
$p \in \Pos_{\HH_\DD}(s')$ and thus also $p \in \Pos_{\NU'}(s')$.
Since $s'|_p = u\sigma$ we also have
\begin{gather}
\label{property1}
\Pos_{\HH_\DD}(u\sigma) \subseteq \Pos_{\NU'}(u\sigma)
\end{gather}
Furthermore, because $t = s'[v\sigma]_p$,
$\Pos_{\HH_\DD}(t) \subseteq \Pos_{\NU'}(t)$ follows from
$\Pos_{\HH_\DD}(v\sigma) \subseteq \Pos_{\NU'}(v\sigma)$.
The latter inclusion we prove by a case analysis on $u \to v$.
Let $q \in \Pos_{\HH_\DD}(v\sigma)$.
\begin{itemize}[label=($1_\rho$)]
\item[\eqref{xi1}]
We have $u = f(\vec{\ell},\langle \seq[m_f]{x} \rangle[\top]_i)$
and $v = \xi_\top(r)$ with $\ell = f(\vec{\ell}) \to r$ a rule of $\RR$.
If $q \in \Pos_{\HH_\DD}(\xi_\top(r))$ then $q \in \Pos_{\FF_\DD}(r)$
and thus $q \in \Pos_{\NU}(r) = \Pos_{\NU'}(\xi_\top(r)) \subseteq
\Pos_{\NU'}(\xi_\top(r)\sigma)$ since $\NU$ is a usable replacement map.
Otherwise $q = q_1 q_2$ with $q_1 \in \Pos_\VV(\xi_\top(r)) = \Pos_\VV(r)$
and $q_2 \in \Pos_{\HH_\DD}(z\sigma)$ where $z = \xi_\top(r)|_{q_1}$.
Since $z \in \Var(r) \subseteq \Var(\ell)$, there exists a
position $q_3 \in \Pos_\VV(\ell) \subseteq \Pos_\VV(u)$ such
that $\ell|_{q_3} = u|_{q_3} = z$.
We have
$q_3 q_2 \in \Pos_{\HH_\DD}(u\sigma)$ and thus
$q_3 q_2 \in \Pos_{\NU'}(u\sigma)$ by \eqref{property1}.
Hence both $q_3 \in \Pos_{\NU'}(u)$ and
$q_2 \in \Pos_{\NU'}(z\sigma)$. Since
$q_3 \in \Pos_{\NU'}(u) = \Pos_{\NU}(\ell)$, we obtain
$q_1 \in \Pos_{\NU}(r) = \Pos_{\NU'}(r) = \Pos_{\NU'}(v)$
from the usability of $\NU$. Hence $q \in \Pos_{\NU'}(v\sigma)$
as desired.
\item[\eqref{xi2}]
We have $u = f(\vec{\ell},\langle \seq[m_f]{x} \rangle[\top]_i)$
and
$v = f_i^1(\vec{\ell},\langle \seq[m_f]{x} \rangle[\xi_\top(a_1)]_i)$.
Comparing $u\sigma$ and
$v\sigma$ and observing that $\NU'(f)$ and $\NU'(f_i^1)$ agree
on $\{ 1, \dots, n + m_f \} \setminus \{ n+i \}$, the only interesting
case is $q = (n+i)\,q'$ with $q' \in \Pos_{\HH_D}(\xi_\top(a_1)\sigma)$.
We distinguish two subcases.
If $q' \in \Pos_{\HH_\DD}(\xi_\top(a_1)) = \Pos_{\HH_\DD}(a_1)$ then
$q' \in \Pos_{\NU}(a_1) = \Pos_{\NU'}(\xi_\top(a_1))$ and,
since $n+i \in \NU'(f_i^1)$, $q \in \Pos_{\NU'}(v\sigma)$.
Otherwise
$q' = q_1 q_2$ with $q_1 \in \Pos_\VV(\xi_\top(a_1)) = \Pos_\VV(a_1)$
and $q_2 \in \Pos_{\HH_\DD}(z\sigma)$ where $z = (a_1)|_{q_1}$.
Sinze $z$ must occur in $\Var(\vec{\ell}) = \Var(b_0)$, we conclude
as in case \eqref{xi1}.
\item[\eqref{xi3}]
We have $u = f_i^k(\vec{\ell},\langle \seq[m_f]{x} \rangle[\seq[k]{b}]_i)$
and $v = \xi_\top(r)$.
Like in case \eqref{xi1} we distinguish two cases.
The case $q \in \Pos_{\HH_\DD}(\xi_\top(r))$ is dealt with as before.
Suppose $q = q_1 q_2$ with $q_1 \in \Pos_\VV(\xi_\top(r)) = \Pos_\VV(r)$
and let $z = r|_{q_1}$. The variable $z$ occurs in
$\seq{\ell}$ or $\seq[k]{b}$. The former is treated as before.
Suppose $z = (b_l)|{q_3}$ with $1 \leqslant l \leqslant k$.
We have $(n+i+l-1)\,q_3 q_2 \in \Pos_{\NU'}(u\sigma)$
according to \eqref{property1}.
Hence $q_3 \in \Pos_{\NU'}(b_l) = \Pos_{\NU}(b_l)$ and thus
$q_1 \in \Pos_{\NU}(r)$ by the usability of $\NU$. We obtain
$q \in \Pos_{\NU'}(v\sigma)$ as before.
\item[\eqref{xi4}]
We have
$u = f_i^j(\vec{\ell},\langle \seq[m_f]{x} \rangle[\seq[j]{b}]_i)$
and
$v = f_i^{j+1}(\vec{\ell},
\langle \seq[m_f]{x} \rangle[\seq[j]{b},\newline\xi_\top(a_{j+1}]_i)$.
Comparing $u\sigma$ and
$v\sigma$ as well as $\NU'(f_i^j)$ and $\NU'(f_i^{j+1})$ allows us
to focus on the interesting case: $q = (n+i+j)\,q'$ with
$q' \in \Pos_{\HH_D}(\xi_\top(a_{j+1})\sigma)$.
We obtain $q' \in \Pos_{\NU'}(\xi_\top(a_{j+1})\sigma)$
and thus $q \in \Pos_{\NU'}(v\sigma)$ by repeating the reasoning
performed in the preceding cases.
\item[\eqref{xi5}]
We have
$u = f_i^j(\vec{\ell},\langle \seq[m_f]{x} \rangle[\seq[j-1]{b},v']_i)$
and $v = f(\vec{\ell},\langle \seq[m_f]{x} \rangle[\bot]_i)$.
We distinguish three cases.
If $q = \epsilon$ then
obviously $q \in \Pos_{\NU'}(v\sigma)$.
Let $q = i' q'$. If $i' \in \{ 1, \dots, n \}$ then
$q \in \Pos_{\HH_\DD}(u\sigma)$ and thus
$q \in \Pos_{\NU'}(u\sigma)$ by \eqref{property1}.
Hence also $q \in \Pos_{\NU'}(v\sigma)$ since
$\NU'(f) = \NU(f) = \NU'(f_i^j) \cap \{ 1, \dots, n \}$.
In the remaining case we have
$i' \in \{ n+1, \dots, n+m_f \} \setminus \{ n+i \}$ and thus
$v\sigma|_q = \sigma(x_{i'-n})|_{q'}$.
However, this subterm appears in $u\sigma$ at a position
not in $\Pos_{\NU'}(u\sigma)$ and thus cannot contain
defined symbols according to \eqref{property1},
contradicting the assumption $q \in \Pos_{\HH_\DD}(v\sigma)$.
\item[\eqref{xi6}]
We have $u =
f(\langle \seq{y} \rangle[v']_j,\langle \seq[m_f]{x} \rangle[\top]_i)$
for some $v' \in \m{AP}(\ell_j)$ and $v =
f(\langle \seq{y} \rangle[v']_j,\langle \seq[m_f]{x} \rangle[\bot]_i)$.
In this case we obviously have
$\Pos_{\HH_\DD}(u\sigma) = \Pos_{\HH_\DD}(v\sigma)$
and
$\Pos_{\NU'}(u\sigma) = \Pos_{\NU'}(v\sigma)$.
Hence $q \in \Pos_{\NU'}(v\sigma)$ is a consequence of \eqref{property1}.
\qedhere
\end{itemize}
\end{proof}

\noindent Using Theorem~\ref{thm:transformcomplete},
the inequality in Theorem~\ref{thm:mapreplace} becomes
an equality if we restrict the terms to consider for
$\m{rc}_{\Xi(\RR),\mu\NU}(n)$ to those that correspond to
labeled basic terms.

Theorem~\ref{thm:mapreplace} is highly relevant when using
interpretations since the (strong or weak) monotonicity
requirements are only imposed on the active arguments of the
interpretation functions.

\begin{exa}
\label{ex:evenfinal}
For the CCTRS $\RR_\m{even}$ we can take the (empty) usable
replacement map $\NU(f) = \varnothing$ for all function symbols
because $\m{even}$ and $\m{odd}$ do not appear
below the root in the right-hand side or left-hand side of
a condition of any rule. This implies that $\I_\m{even}$ and
$\I_\m{odd}$ do not need to be monotone in their first arguments.
Hence we can simplify the interpretation of
Example~\ref{ex:evenpolynomial} to
\begin{align*}
\I_\top &= 1
\qquad
\I_\bot = \I_{\m{true}} = \I_{\m{false}} = \I_{\m{0}} = 0
\qquad
\I_{\m{s}}(x) = x+1
\\
\I_{\m{even}}(x,u,v,w) &= \I_{\m{odd}}(x,u,v,w) = 1 + v \cdot (2^x - 1)
  + w \cdot (2^x - 1)
\\
\I_{\m{even}_2^1}(x,u,v,w) &= \I_{\m{odd}_2^1}(x,u,v,w) = 1 +
  v + w \cdot (2^x - 1)
\\
\I_{\m{even}_3^1}(x,u,v,w) &= \I_{\m{odd}_3^1}(x,u,v,w) = 1 +
  v \cdot (2^x - 1) + w
\end{align*}
The rules are still oriented; for example,
rule ($2_2$) gives rise to the inequality
\[
1 + 2^{x+1} - 1 + z \cdot (2^{x+1} - 1)
\geqslant 1 + (1 + 2 \cdot (2^x - 1)) + z \cdot (2^{x+1} - 1)
\]
which holds because $2^{x+1} - 1 = 1 + 2 \cdot (2^x - 1)$.
The above interpretation induces a runtime complexity of $\OO(2^n)$.
This is a \emph{tight} bound, as we observed earlier.
\end{exa}

\begin{defi}[Recipe \refstepcounter{recipe}\label{def:recipeB}%
\therecipe: Extension for Runtime Complexity]
Recipe~\ref{def:recipeA} is altered as follows, assuming we are
given a usable replacement map $\NU$ for $(\FF,\RR)$.
Rather than demanding (strict or weak) monotonicity of the functions
$\J_f^i$ in \emph{all} arguments, we merely demand that
\begin{itemize}
\item
$\J_f^0$ is strictly monotone in the arguments in $\NU(f)$,
for all $f \in \FF$,
\item
$\J_f^i$ is weakly monotone in the arguments in $\NU(f)$,
for all $f \in \FF_\DD$ and $1 \leqslant i \leqslant m_f$,
\item
as before,
$\J_{f,i}^j$
is strictly monotone in argument $n + i + j - 1$,
where $n$ is the arity of $f \in \FF_\DD$.
\end{itemize}
Given $\J$, the definition of $\I$ remains the same.
\end{defi}

Recipe~\ref{def:recipeB} can be used like Recipe~\ref{def:recipeA},
but only for runtime complexity.

\begin{lem}
\label{lem:recipe2}
The interpretation $\I$ from Recipe~\ref{def:recipeB} is a
context-sensitive interpretation for $(\HH,\mu\NU)$.
If its interpretation functions satisfy the compatibility constraints
from Definition~\ref{compatibility constraints},
then $\I$ is compatible with $\HH$ and
\[\m{crc}_\RR(n) = \max \{ \interpret{\xi_\top(t)}{} \mid
\text{$t \in \TT(\FF)$, $|t| \leqslant n$, and $t$ is basic} \}\,.\]
Moreover,
\[
\interpret{\xi_\top(f(\seq{t}))}{\alpha} =
\sum_{i=0}^{m_f} \J_f^i(\interpret{\xi_\top(t_1)}{\alpha},\dots,
\interpret{\xi_\top(t_n)}{\alpha})
\]
\end{lem}

\begin{proof}
It is not hard to see that if the restrictions in the recipe are
satisfied, then indeed all interpretation functions $\I_f$ are
strictly monotone in all arguments $i \in \mu\NU(f)$.
The result for $\m{crc}_\RR$ follows by Theorem~\ref{thm:mapreplace}.
Compatibility and equivalence are obtained from
Lemma~\ref{lem:recipe1}, as changing the monotonicity requirements
does not affect either property.
\end{proof}

\begin{exa}
\label{ex:fibruntimesimple}
We use Recipe~\ref{def:recipeB} to derive an upper bound for the runtime
complexity of $\RR_{\m{fib}}$. From Example~\ref{ex:fibonacci2} we
know that the replacement map $\nurc$ defined by
$\nurc(\m{s}) = \{ 1 \}$,
$\nurc(\m{+}) = \nurc(\langle \cdot , \cdot \rangle) = \{ 1, 2 \}$,
and $\nurc(\m{fib}) = \varnothing$ is usable.
For the interpretations, we assign:
\begin{gather*}
\J_{\m{0}}^0 = 0 \qquad
\J_{\m{s}}^0(x) = x + 1 \qquad
\J_{\langle \cdot , \cdot \rangle}^0(x,y) = x + y + 1 \qquad
\J_{+}^0(x,y) = 2x + y + 1
\\
\J_{+}^1(x,y) = \J_{+}^2(x,y) = 0 \qquad
\J_{\m{fib}}^0(x) = 3 \qquad
\J_{\m{fib}}^1(x) = 0 \qquad
\J_{\m{fib}}^2(x) = 5 \cdot (3^x - 1)
\\
\J_{\m{fib},2}^1(x,a) = 3a \qquad
\J_{\m{fib},2}^2(x,a,b) = a + b
\end{gather*}
One easily verifies
that these interpretations are strictly monotone
in the required argument positions, and weakly monotone in all
argument positions.
Omitting the (automatically satisfied) proof obligations for rules
\eqref{xi5} and \eqref{xi6}, this leaves
\[
\begin{array}{r@{~}c@{~}r@{~}c@{~}l@{~}c@{~}l}
\interpret{+(\m{0},y,\top,z)}{} &= & y + 1 & > & y
& = & \interpret{y}{} \\[.5ex]
\interpret{+(\m{s}(x),y,z,\top)}{} & = & 2x+y+3 & > & 2 x + y + 2
& = & \interpret{\m{s}(+(x,y,\top,\top))}{} \\[.5ex]
\interpret{\m{fib}(\m{0},\top,u)}{} & = & 3 + u \cdot 5 \cdot 0 & > & 2
& = & \interpret{\langle \m{0},\m{s}(\m{0}) \rangle}{} \\[.5ex]
\interpret{\m{fib}(\m{s}(x),u,\top)}{} & = & 5 \cdot 3^{x+1} - 2
& \geqslant & 5 \cdot 3^{x+1} - 3
& = & \interpret{\m{fib}_2^1(\m{s}(x),u,\m{fib}(x,\top,\top))}{} \\[.5ex]
\interpret{\m{fib}_2^1(\m{s}(x),u,\langle y, z \rangle)}{}
& = & 3y + 3z + 6 & \geqslant & 3y + 2z + 5 & = &
\interpret{\m{fib}_2^2(\m{s}(x),u,\langle y, z \rangle,
+(y,z,\top,\top))}{} \\[.5ex]
\interpret{\m{fib}_2^2(\m{s}(x),u,\langle y, z \rangle,w)}{} & = &
y + z + w + 4 & > & z + w + 1 & = & \interpret{\langle z, w \rangle}{}
\end{array}
\]
which holds for all values of $x$, $y$, $z$, $u$, and $w$.
From this we conclude $\OO(3^n)$ runtime complexity
by Lemma~\ref{lem:recipe2}.
\end{exa}

Note that Recipe~\ref{def:recipeB} may not be used for derivational
complexity.

\begin{exa}
\label{exa:odd}
The system $\RR_{\m{odd}}$ is a variation of $\RR_{\m{even}}$
defined by the following rules:
\begin{xalignat*}{2}
\m{odd}(\m{0}) &\to \m{false} &
\m{not}(\m{true}) &\to \m{false} \\
\m{odd}(\m{s}(x)) &\to \m{not}(y) ~\Leftarrow~ \m{odd}(x) \approx y &
\m{not}(\m{false}) &\to \m{true}
\end{xalignat*}
We will use Recipe~\ref{def:recipeB} to derive an upper bound for the
runtime complexity of this CCTRS, giving a bit more detail as to
how the interpretations are chosen. The replacement map $\NU$ with
$\NU(\m{odd}) = \NU(\m{s}) = \varnothing$ and $\NU(\m{not}) = \{ 1 \}$
is usable.
Since the unconditional rules will be taken care of by
the choice of $\J_{\m{odd}}^0$ and $\J_{\m{not}}^0$, we
let $\J_{\m{odd}}^1(x) = \J_{\m{not}}^1(x) = \J_{\m{not}}^2(x) = 0$.
For clarity, we assign different names
to the remaining interpretation functions:
\begin{align*}
\J_{\m{true}}^0 &= T &
\J_{\m{0}}^0 &= Z &
\J_{\m{odd}}^0 &= O &
\J_{\m{odd},2}^1 &= C \\
\J_{\m{false}}^0 &= F &
\J_{\m{s}}^0 &= S &
\J_{\m{odd}}^2 &= D &
\J_{\m{not}}^0 &= N
\end{align*}
Here $T$, $F$, and $Z$ are (unknown) constants,
$C$, $D$, $N$, $O$, and $S$ are (unknown) unary functions, and
$N$ must be strictly monotone.
The recipe gives rise to the following constraints:
\begin{align*}
O(Z) &> F &
N(T) &> F &
N(F) &> T
\end{align*}
for the unconditional rules and
\begin{align*}
D(S(x)) &\geqslant C(S(x),O(x) + D(x))
\\
O(S(x)) + C(S(x),y) &> N(y)
\end{align*}
for the conditional rule.
The constraints $N(T) > F$ and $N(F) > T$ are satisfied
by setting
$F = T = 0$ and $N(x) = x + 1$
(recall that $N$ must be strictly monotone).
As $O$ is not required to be (strictly) monotone, and the constraints
give little reason for $O$ to regard its argument, we let $O(x) = A$
for some constant $A$.
Hence the remaining constraints reduce to

\begin{align*}
A &> 0 \\
D(S(x)) &\geqslant C(S(x),A + D(x)) \\
A + C(S(x),y) &> y + 1
\end{align*}
By taking $A = 2$ and $C(x,y) = y$ we are left with
\begin{align*}
D(S(x)) &\geqslant 2 + D(x)
\end{align*}
which is easily satisfied by choosing $D(x) = x$
and $S(x) = x + 2$.
With these choices, we have $\interpret{s}{} \leqslant 2 \cdot |s|$
for all terms $s$, so we obtain linear runtime complexity
by Lemma~\ref{lem:recipe2}.
\end{exa}

Note that the use of the replacement map $\NU$ was essential to
obtain linear runtime complexity; if $\J_{\m{odd}}^0 = O$ was
required to be monotone in its first argument, we would have had to
choose $O(x) = x + 1$ or worse. While this would allow us to choose
the tighter interpretation $S(x) = x + 1$,
it would have produced the constraint
$D(x+1) \geqslant D(x) + x + 1$, which can be satisfied with a
\emph{quadratic} interpretation $D(y) = y^2$, but not with a linear
one.

\section{Splitting Time and Space Complexity}
\label{sec:splitsize}

Another method to improve interpretations is to separate
\emph{time} and \emph{space} complexity. To understand the
motivation, consider Example~\ref{ex:fibruntimesimple}. Since the
rules for addition had to be oriented strictly, the interpretation
$\J_+^0(x,y) = 2x + y + 1$ was chosen rather than the simpler
$\J_+^0(x,y) = x + y$. However, this does not accurately reflect the
number of steps it takes to evaluate an addition.
Rather, it reflects the sum of the
number of steps \emph{plus} the size of the result. This high value
for the interpretation also affects the interpretations for
other symbols. And while the difference is only a constant factor,
which is not an issue in \emph{polynomial} interpretations, it is a
cause for concern when considering \emph{exponential}
complexities; compare $\OO(2^n)$ and $\OO(2^{(an)}) = \OO((2^a)^n)$.

Thus, as an alternative, let us consider interpretations not in
$\N$, but rather in $\N^2$: pairs $(n,m)$, where $n$ records
the number of steps to evaluate a term to
constructor normal form, and $m$ the size of the result.
These pairs are equipped with the following orders:
$(n_1,m_1) > (n_2,m_2)$ if $n_1 > n_2$ and $m_1 \geqslant m_2$, and
$(n_1,m_1) \geqslant (n_2,m_2)$ if $n_1 \geqslant n_2$ and
$m_1 \geqslant m_2$. We suggestively write $\Cost((n,m)) = n$ and
$\Size((n,m)) = m$, and note that $\Cost(x) > \Cost(y)$ if $x > y$.
Consequently, $\m{dh}(s,\to_{\Xi(\RR),\mu}) \leqslant
\Cost(\interpret{s}{\alpha})$ for any valuation $\alpha$ over $\N^2$.

\begin{exa}
\label{ex:evenreallyfinal}
We revisit Example~\ref{ex:Xi(even/odd)} and define
\begin{align*}
\I_\top &= (0,1)
\qquad
\I_\bot = \I_{\m{true}} = \I_{\m{false}} = \I_{\m{0}} = (0,0)
\qquad
\I_{\m{s}}((c,s)) = (c,s+1)
\\
\I_{\m{even}}(x,u,v,w) &= \I_{\m{odd}}(x,u,v,w) = (1 + \Cost(x) +
(\Size(v) + \Size(w)) \cdot A(x), 0)
\\
\I_{\m{even}_2^1}(x,u,v,w) &= \I_{\m{odd}_2^1}(x,u,v,w) =
(1 + \Cost(x) + \Cost(v) + \Size(w) \cdot A(x), 0)
\\
\I_{\m{even}_3^1}(x,u,v,w) &= \I_{\m{odd}_3^1}(x,u,v,w) =
(1 + \Cost(x) + \Size(v) \cdot A(x) + \Cost(w), 0)
\end{align*}
where
\[
A(x) = (\Cost(x) + 1) \cdot (2^{\Size(x)} - 1)
\]
All interpretations are weakly monotone in all arguments because
$x \geqslant y$ implies both $\Cost(x) \geqslant \Cost(y)$ and
$\Size(x) \geqslant \Size(y)$, and in all interpretation functions
$\Cost(\cdot)$ and $\Size(\cdot)$ are only used positively.
If $x > x'$ then
$\I_\m{even}(x,u,v,w) > \I_\m{even}(x',u,v,w)$
since $\Cost(\I_\m{even}(x,u,v,w))$ has a $\Cost(x)$ summand.
The same holds for
$\I_\m{odd}(x,u,v,w)$ and $\I_\m{s}(x)$.
Both
$\I_{\m{even}_2^1}(x,u,v,w)$ and $\I_{\m{odd}_2^1}(x,u,v,w)$
have a $\Cost(v)$ summand and hence are strictly monotone in
their third arguments, and $\I_{\m{even}_3^1}(x,u,v,w)$ and
$\I_{\m{odd}_3^1}(x,u,v,w)$ have a $\Cost(w)$ summand.
Hence $\I$ satisfies the monotonicity requirements.

Furthermore, all rules of $\Xi(\RR_{\m{even}})$ are oriented as
required. For the size component this is clear as
$\Size(\interpret{\ell}{\alpha}) = 0 = \Size(\interpret{r}{\alpha})$
for all rules $\ell \to r$.
For the cost component, we see that
rules of the form \eqref{xi5} are oriented because
$\Cost(v) \geqslant 0 = \Size(\interpret{\bot}{\alpha}) \cdot A(x)$,
and rules of the form \eqref{xi6} are oriented by monotonicity since
$\interpret{\top}{\alpha} = (0,1) \geqslant (0,0) =
\interpret{\bot}{\alpha}$.
Rules \eqref{ev11}, \eqref{ev14}, \eqref{ev32}, \eqref{ev33},
\eqref{ev35}, and \eqref{ev36} are strictly oriented
since their left-hand sides evaluate to $1$ whereas the right-hand
sides evaluate to $0$.
The only rules where the orientation is non-trivial are \eqref{ev22},
\eqref{ev23}, \eqref{ev25}, and \eqref{ev26}.
We consider \eqref{ev22}:
\begin{gather*}
1 + \Cost(x) + (1 + \Size(w)) \cdot A((\Cost(x),\Size(x) + 1)) \\
\qquad \geqslant
1 + \Cost(x) + (1 + \Cost(x) + 2 \cdot A(x)) + \Size(w) \cdot
A((\Cost(x),\Size(x) + 1))
\end{gather*}
Removing equal parts from both sides and
inserting the definition of $A$ yields
\begin{gather*}
(\Cost(x)+1) \cdot (2^{\Size(x) + 1} - 1)
\geqslant
1 + \Cost(x) + 2 \cdot (\Cost(x) + 1) \cdot (2^{\Size(x)} - 1)
\end{gather*}
and one easily checks that both sides are equal.

Now, towards runtime complexity, an easy induction proof shows that
$\Cost(\interpret{s}{}) = 0$ and $\Size(\interpret{s}{}) \leqslant n$
for all ground constructor terms $s$ with $|s| \leqslant n$.
Therefore, the conditional
runtime complexity $\m{crc}_{\RR_{\m{even}}}(n)$ is bounded by
\begin{align*}
&\max \{ \Cost(\interpret{f(\seq[m]{s},\top,\dots,\top)}{}) \mid
\text{$f \in \FF_\DD$ and $\seq[m]{s}$ are ground constructor} \\[-.5ex]
&\phantom{\max \{ \Cost(\interpret{f(\seq[m]{s},\top,\dots,\top)}{})
\mid {}} \text{terms with $|s_1| + \dots + |s_m| < n$} \} \\
{} = {} &\max \{ \Cost(\I_f((0,x_1),\dots,(0,x_m),(0,1),\dots,(0,1))) \mid
\text{$f \in \FF_\DD$ and $x_1 + x_2 + x_3 + x_4 < n$} \} \\
{} = {} &\max \{ 1 + 0 + 2 \cdot 1 \cdot (2^x - 1) \mid x < n \} =
2^n - 1 \leqslant 2^n
\end{align*}
This is the same bound that we obtained in Example~\ref{ex:evenfinal},
but without employing context-sensitivity.
\end{exa}

Interestingly, we can obtain the same tight bound for
derivational complexity.

\begin{exa}
We prove by induction that for all ground $\top$-terms $s$ there
exist $K, N \geqslant 0$ with $K + N \leqslant |s|$ such that
$\Cost(\interpret{s}{}) \leqslant 2^N - 1$ and
$\Size(\interpret{s}{}) \leqslant K$.
\begin{itemize}
\item
If $s$ is $\m{0}$, $\m{true}$, or $\m{false}$
then $\Cost(\interpret{s}{}) = 0 = \Size(\interpret{s}{})$,
so we can take $K = N = 0$.
\smallskip
\item
If $s = \m{s}(t)$ and $t$ is bounded by $(K,N)$, then
$\Cost(\interpret{s}{}) = \Cost(\interpret{t}{}) \leqslant 2^N - 1$
and $\Size(\interpret{s}{})) = 1 + \Size(\interpret{t}{}) \leqslant
K + 1$, so we can take $(K+1,N)$.
\smallskip
\item
If $s = \m{even}(t,\top,\top,\top)$ or
$s = \m{odd}(t,\top,\top,\top)$ with $t$ bounded by $(K,N)$ then
$\Size(\interpret{s}{}) = 0$ and
\begin{align*}
\Cost(\interpret{s}{})
&= 1 + \Cost(\interpret{t}{}) + 2 \cdot (\Cost(\interpret{t}{}) + 1)
\cdot (2^{\Size(\interpret{t}{})} - 1) \\
&\leqslant 1 + (2^N - 1) + 2 \cdot 2^N \cdot (2^K - 1) \\
&= 2^N + 2^{N + K + 1} - 2^{N + 1} \\
&= 2^{N + K + 1} - 2^N \\
&\mathrel{\leqslant} 2^{N + K + 1} - 1
\end{align*}\enlargethispage{\baselineskip}
and so we can take $(0,N + K + 1)$.
\end{itemize}
It follows that $\m{cdc}_{\RR_{\m{even}}}(n) = \OO(2^n)$.
\end{exa}

Separating the ``cost'' and ``size'' component made it possible to
obtain an exponential bound for the derivational
complexity of $\RR_\m{even}$.
However, 
the derivation of this bound
is ad-hoc, and it would require a more systematic analysis of various
systems with the separated $\Cost/\Size$ approach to obtain a
strategy to find such bounds.
For runtime complexity, the approach is more straightforward. If for
all $f \in \FF_\CC$ the result $\I_f(\seq{x})$ has the form
\[
(c(\Cost(x_1),\dots,\Cost(x_n)),s(\Size(x_1),\dots,\Size(x_n)))
\]\enlargethispage{\baselineskip}
where $c$ is a linear polynomial with coefficients in $\{ 0, 1 \}$ and
constant part $0$, and $s$ is a linear polynomial with coefficients
in $\{ 0, 1 \}$ and a constant part at most $K$, then
all ground constructor terms $s$ have cost $0$ and size at most
$K \cdot |s|$, so $\m{crc}_\RR(n)$ is bounded
by the maximum value of $\I_f((0,s_1),\dots,(0,s_{m}),(0,1),\dots,
(0,1))$ where $f \in \FF_\DD$ and $s_1 + \dots + s_{m} < K \cdot n$.
This mirrors the corresponding notion of ``strongly linear
polynomials'' in the setting with interpretations over $\N$, and is what
we used in Example~\ref{ex:evenreallyfinal} (with $K = 1$).

As before, we will use a standard recipe to find such interpretations.
To this end, we adapt the ideas from Recipes~\ref{def:recipeA}
and~\ref{def:recipeB}.

\begin{defi}[Recipe \refstepcounter{recipe}\label{def:recipeC}%
\label{def:recipe:costsize}\therecipe: Cost/Size Version]
Given a usable replacement map $\NU$,
we consider the replacement map $\mu\NU$ where, for $f$ of
arity $n$ in the original signature $\FF$,
$\mu\NU(f) = \NU(f)$ when considering runtime complexity
and $\mu\NU(f) = \{ 1, \dots, n \}$ otherwise.
Given interpretation functions
\begin{itemize}
\item
$\IS_f\colon \N^n \to \N$ and
$\IC_f^0, \dots, \IC_f^{m_f}\colon \N^{2n} \to \N$ 
for every symbol $f$ of arity $n$ in
$\FF$ such that $\RR{\restriction}f$ consists of $m_f$ rules,
\smallskip
\item
$\IS_{f,i}^1, \dots, \IS_{f,i}^k$
with $\IS_{f,i}^j\colon \N^{n+j} \to \N$ and
$\IC_{f,i}^1, \dots, \IC_{f,i}^k$
with $\IC_{f,i}^j \colon \N^{2(n+j)} \to \N$
for every rule $\rho_i \in \RR{\restriction}f$ with $k > 0$ conditions
\end{itemize}
such that the following monotonicity constraints are satisfied:
\begin{itemize}
\item
$\IS_f$ is weakly monotone in all arguments in $\mu\NU(f)$,
\smallskip
\item
$\IC_f^0$ is strictly monotone in all arguments in $\mu\NU(f)$
and weakly monotone in all arguments in
$\{ n + j \mid j \in \mu\NU(f) \}$,
\smallskip
\item
$\IC_f^i$ is weakly monotone in all arguments in
$\{ j, n + j \mid j \in \mu\NU(f) \}$,
\smallskip
\item
$\IS_{f,i}^j$ is weakly monotone in its last argument $n + j$,
\smallskip
\item
$\IC_{f,i}^j$ is strictly monotone in argument $n + j$ and weakly
monotone in argument $2(n+j)$,
\end{itemize}
we construct an interpretation $\I$ for $\HH$ as follows:
$\I_\top = (0,1)$ and $\I_\bot = (0,0)$,
\begin{align*}
\I_f(\seq{x},\seq[m_f]{c}) = \bigl(\:&
\IC_f^0(\Cost(\vec{x}),\Size(\vec{x})) +
\sum_{k=1}^{m_f} \Size(c_k) \cdot
\IC_f^k(\Cost(\vec{x}),\Size(\vec{x})), \\[-1ex]
&\IS_f(\Size(\vec{x}))\:\bigr)
\end{align*}
for every $f \in \FF_\CC \cup \FF_\DD$ of arity $n$, and finally
\begin{align*}
\I_{f_i^j}(&\seq{x},\seq[i-1]{c},\seq[j]{y},c_{i+1},\dots,c_{m_f}) = \\
&\bigl(\:
\IC_f^0(\Cost(\vec{x}),\Size(\vec{x})) + 
\IC_{f,i}^j(\Cost(\vec{x}),\Cost(\vec{y}),
\Size(\vec{x}),\Size(\vec{y})) \\
&+ \sum_{k=1,\,k \neq i}^{m_f}
\Size(c_k) \cdot \IC_f^k(\Cost(\vec{x}),\Size(\vec{x})),\:
\max(\IS_f(\Size(\vec{x})),\IS_{f,i}^j(\Size(\vec{x}),
\Size(\vec{y})))\:\bigr)
\end{align*}
Here $\Cost(\vec{x})$ and $\Size(\vec{x})$ stand for
$\Cost(x_1),\dots,\Cost(x_n)$
and $\Size(x_1),\dots,\Size(x_n)$,
and similar for
$\Cost(\vec{y})$ and $\Size(\vec{y})$.
\end{defi}

The following remarks are helpful to understand the intuition behind
the interpretations defined in the above recipe.
\begin{itemize}
\item
The ``size'' of a term $s$ is intended to reflect---or at least
bound---how large a normal form of $s$ may be, where different
constructor symbols count differently towards the size.
In a term $f(\seq{s},\seq[m_f]{t})$, the size is
only affected by the sizes of $\seq{s}$; the
additional arguments merely indicate our progress in trying to
reduce the term.
In a term of the shape
$f_i^j(\seq{s},\langle \seq[m_f]{t} \rangle[\seq[j]{y}]_i)$
the size should similarly not be affected by the
progress on testing the applicability of the rule
$\rho_i \in \RR{\restriction}f$.
However, here a
rule-specific size function is included in a $\max$ expression
for technical reasons; in practice, we will always have 
$\IS_f(\cdots) \geqslant \IS_{f,i}^j(\cdots)$, but the latter will
have more variables that can be used to orient rules of the form
\eqref{xi3}.
\item
The ``cost'' of $f(\seq{s},\seq[m_f]{t})$ reflects how many steps we may
take to reach a normal form. This is affected by the cost of evaluating
each of the rule conditions where $t_i = (0,1)$
is the value of $\top$, as well as the
cost of evaluating whatever we may reduce to; the sizes of the
arguments may affect both those costs (since it will take longer
to evaluate $\m{even}(\m{s}^{100}(\m{0}))$ than
$\m{even}(\m{0})$, for instance).
\end{itemize}
As before, using this interpretation for the rules in
Definition~\ref{def:transformrules}, the obtained inequalities can be
greatly simplified.

\begin{defi}
\label{compatibility constraints 2}
The \emph{compatibility constraints} for $\IC$ and $\IS$ comprise the
following inequalities, for every rule
$\rho_i\colon f(\seq{\ell}) \to r \Leftarrow a_1 \approx b_1,
\dots, a_k \approx b_k$ in $\RR$:
\begin{itemize}[label=($1_\rho$)]
\item[($1_\rho$)]
$\IS_f(\overrightarrow{[\ell]_\IS}) \geqslant [\xi_\top(r)]_\IS$,
\smallskip
\item[($2_\rho$)]
$\IS_f(\overrightarrow{[\ell]_\IS})
\geqslant \IS_{f,i}^1(\overrightarrow{[\ell]_\IS},[\xi_\top(a_1)]_\IS)$,
\smallskip
\item[($3_\rho$)]
$\IS_{f,i}^k(\overrightarrow{[\ell]_\IS},[b_1]_\IS,\dots,[b_k]_\IS)
\geqslant [\xi_\top(r)]_\IS$,
\smallskip
\item[($4_\rho$)]
$\IS_{f,i}^j(\overrightarrow{[\ell]_\IS},[b_1]_\IS,\dots,[b_j]_\IS)
\geqslant 
\IS_{f,i}^{j+1}(\overrightarrow{[\ell]_\IS},[b_1]_\IS,\dots,[b_j]_\IS,
[\xi_\top(a_{j+1})]_\IS)$
\smallskip
\end{itemize}
and
\begin{itemize}[label=($1_\rho$)]
\item[($1_\rho$)]
$\IC_f^0(\overrightarrow{[\ell]_\IC},\overrightarrow{[\ell]_\IS})
+ \IC_f^i(\overrightarrow{[\ell]_\IC},\overrightarrow{[\ell]_\IS})
> [\xi_\top(r)]_\IC$,
\smallskip
\item[($2_\rho$)]
$\IC_f^i(\overrightarrow{[\ell]_\IC},\overrightarrow{[\ell]_\IS})
\geqslant
\IC_{f,i}^1(\overrightarrow{[\ell]_\IC},[\xi_\top(a_1)]_\IC,
\overrightarrow{[\ell]_\IS},[\xi_\top(a_1)]_\IS)$,
\smallskip
\item[($3_\rho$)]
$\IC_{f,i}^k(\overrightarrow{[\ell]_\IC},[b_1]_\IC,\dots,[b_k]_\IC,
\overrightarrow{[\ell]_\IS},[b_1]_\IS,\dots,[b_k]_\IS)
+ \IC_f^0(\overrightarrow{[\ell]_\IC},\overrightarrow{[\ell]_\IS})
> [\xi_\top(r)]_\IC$,
\smallskip
\item[($4_\rho$)]
$\IC_{f,i}^j(\overrightarrow{[\ell]_\IC},[b_1]_\IC,\dots,[b_j]_\IC,
\overrightarrow{[\ell]_\IS},[b_1]_\IS,\dots,[b_j]_\IS)
\geqslant {}$ \\
$\IC_{f,i}^{j+1}(\overrightarrow{[\ell]_\IC},[b_1]_\IC,\dots,[b_j]_\IC,
[\xi_\top(a_{j+1})]_\IC,
\overrightarrow{[\ell]_\IS},[b_1]_\IS,\dots,[b_j]_\IS,
[\xi_\top(a_{j+1})]_\IS)$
\end{itemize}
for the same cases of $k$ and $j$ as in
Definition~\ref{def:transformrules}.
Here $[s]_\IS = \Size(\interpret{s}{\alpha})$,
$[s]_\IC = \Cost(\interpret{s}{\alpha})$, and
$\overrightarrow{[\ell]_\IS}$ and $\overrightarrow{[\ell]_\IC}$
denotes the sequences $[\ell_1]_\IS, \dots, [\ell_n]_\IS$
and $[\ell_1]_\IC, \dots, [\ell_n]_\IC$.
\end{defi}

\begin{lem}
\label{lem:recipe3}
The interpretation $\I$ from Recipe~\ref{def:recipeC} is a
context-sensitive interpretation for $(\HH,\mu\NU)$.
If the corresponding functions $\IC$ and $\IS$
satisfy the compatibility
constraints from Definition~\ref{compatibility constraints 2}, then
\begin{align*}
[\xi_\top(f(\seq{t}))]_\IS &=
\IS_f([\xi_\top(t_1)]_\IS,\dots,[\xi_\top(t_n)]_\IS) \\
[\xi_\top(f(\seq{t}))]_\IC &=
\sum_{i=0}^{m_f} \IC_f^i(
[\xi_\top(t_1)]_\IC,\dots,[\xi_\top(t_n)]_\IC,
[\xi_\top(t_1)]_\IS,\dots,[\xi_\top(t_n)]_\IS)
\end{align*}
Moreover, $\I$ is compatible with $\HH$. Therefore
\begin{align*}
\m{cdc}_\RR(n) &= \max \{ \Cost(\interpret{\xi_\top(t)}{}) \mid
\text{$t \in \TT(\FF)$ and $|t| \leqslant n$} \} \\
\m{crc}_\RR(n) &= \max \{ \Cost(\interpret{\xi_\top(t)}{}) \mid
\text{$t \in \TT(\FF)$, $|t| \leqslant n$, and $t$ is basic} \}
\end{align*}
\end{lem}

\begin{proof}
For the first part of the claim, it is not hard to see that $\I$
satisfies the monotonicity requirements: Every
interpretation function $\I_f$ is strictly
monotone in each argument position belonging to $\mu\NU(f) = \NU(f)$
(or $\{ 1, \dots, n \}$ for derivational complexity), and every
$\I_{f_i^j}$ is strictly monotone in argument position $n+i+j-1$.
The second part of the claim is obtained by writing out definitions.
As for compatibility, minimality of $\interpret{\bot}{\alpha}$
ensures that all constraints obtained from
clause \eqref{xi6} are satisfied, while those obtained from
clause \eqref{xi5} are oriented because 
\begin{align*}
\IC_{f,i}^j(\cdots) &\geqslant 0 = [\bot]_\IS
\cdot \IC_f^i(\cdots)
\intertext{and}
\max \{ \IS_f(\Size(\vec{x})),\IS_{f,i}^j(\cdots) \}
&\geqslant \IS_f(\Size(\vec{x}))
\end{align*}
always hold.
The requirements for the other rules follow from the compatibility
constraints, by expanding the inequalities $[\ell]_\IS \geqslant [r]_\IS$
and $[\ell]_\IC \geqslant [r]_\IC$ or $[\ell]_\IC > [r]_\IC$
depending on the cost of the rule.
For instance,
the actual size constraint for \eqref{xi3} is
\begin{align*}
\max(\IS_f(\overrightarrow{[\ell]_\IS}),
\IS_{f,i}^k(\overrightarrow{[\ell]_\IS},[b_1]_\IS,\dots,[b_k]_\IS))
&> [\xi_\top(r)]_\IS
\intertext{while for \eqref{xi4} we obtain}
\max(\IS_f(\overrightarrow{[\ell]_\IS}),
\IS_{f,i}^j(\overrightarrow{[\ell]_\IS},[b_1]_\IS,\dots,[b_j]_\IS))
&\geqslant \\
\max(\IS_f(\overrightarrow{[\ell]_\IS}), {}&
\IS_{f,i}^{j+1}(\overrightarrow{[\ell]_\IS},[b_1]_\IS,\dots,[b_j]_\IS,
[\xi_\top(a_{j+1})]_\IS))
\end{align*}
Both constraints are clearly implied by the compatibility
constraints of Definition~\ref{compatibility constraints 2}.
The claims on $\m{cdc}_\RR$ and $\m{crc}_\RR$ hold because
$\m{dh}(s,\to_{\Xi(\RR),\mu\NU}) \leqslant \Cost(\interpret{s}{})$.
\end{proof}

As with Lemma~\ref{lem:recipe1}, we can find bounds
on derivation heights without calculating $\xi_\top(t)$.

\begin{exa}
We derive an upper bound for the runtime complexity of
$\RR_{\m{fib}}$, detailing how we arrive at the chosen interpretation.
Recall the rules:
\begin{xalignat*}{2}
\m{0} + y &\to y &
\m{fib}(\m{0}) &\to \langle \m{0},\m{s}(\m{0}) \rangle \\
\m{s}(x) + y &\to \m{s}(x + y) &
\m{fib}(\m{s}(x)) &\to \langle z, w \rangle
~\Leftarrow~ \m{fib}(x) \approx \langle y, z \rangle,~ y + z \approx w
\end{xalignat*}
We take the same usable replacement map $\NU$ as in
Example~\ref{ex:fibruntimesimple}: $\NU(\m{s}) = \{ 1 \}$,
$\NU(+) = \NU(\langle \cdot, \cdot \rangle) = \{ 1, 2 \}$, and
$\NU(\m{fib}) = \varnothing$.
To facilitate understanding of the following constraints, we
present the rules in $\Xi(\RR_{\m{fib}})$ that derive from the conditional
rule
(but note that they are not necessary to apply the recipe): 
\begin{align*}
\m{fib}(\m{s}(x),c_1,\top) &\to
\m{fib}_2^1(\m{s}(x),c_1,\m{fib}(x,\top,\top)) \\
\m{fib}_2^1(\m{s}(x),c_1,\langle y, z \rangle) &\to
\m{fib}_2^2(\m{s}(x),c_1,\langle y, z \rangle,+(y,z,\top,\top)) \\
\m{fib}_2^2(\m{s}(x),c_1,\langle y, z \rangle,w) &\to
\langle z, w \rangle
\end{align*}
Following the recipe, let $N = \IS_{\m{0}}$, $S = \IS_{\m{s}}$,
$P = \IS_{\langle \cdot, \cdot \rangle}$,
$A = \IS_{+}$, $F = \IS_{\m{fib}}$,
$B = \IS_{\m{fib},2}^1$ and $C = \IS_{\m{fib},2}^2$.
The interpretation functions $S$, $P$ and $A$ must be weakly monotone
in all arguments, $B$ and $C$ only in the last argument, and $F$ does not
need to be weakly monotone due to $\NU$.
The requirements on the size component give rise to the constraints
\begin{align}
A(N,y) &\geqslant y
\label{8.1} \\
A(S(x),y) &\geqslant S(A(x,y))
\label{8.2} \\
F(N) &\geqslant P(N,S(N))
\label{8.3}
\intertext{for the unconditional rules and}
F(S(x)) &\geqslant B(S(x),F(x))
\label{8.4} \\
B(S(x),P(y,z)) &\geqslant C(S(x),P(y,z),A(y,z))
\label{8.5} \\
C(S(x),P(y,z),w) &\geqslant P(z,w)
\label{8.6}
\end{align}
for the conditional rule of $\RR_{\m{fib}}$.
For the cost component we will follow the guiding principle that
$\IC_f^0(\seq{x},\seq{y}) \leqslant x_1 + \dots + x_n$ for all
constructor symbols $f \in \FF_\CC$, which gives cost $0$ for ground
constructor terms. As $\IC_f^0$ must be strictly monotone in the
first $n$ arguments for $f \in \FF_\CC$, we fix
$\IC_{\m{0}} = 0$, $\IC_{\m{s}}(x,y) = x$ and
$\IC_{\langle \cdot, \cdot \rangle}(cx,cy,sx,sy) = cx + cy$. We also
fix $\IC_{\m{+}}^1(cx,cy,sx,sy) = \IC_{\m{+}}^2(cx,cy,sx,sy) =
\IC_{\m{fib}}^1(cx,sx) = 0$ since these are the
``conditional evaluation'' components for the unconditional rules.
For the remaining interpretation functions,
write $Q = \IC_{\m{+}}^0$, $G = \IC_{\m{fib}}^0$, $H = \IC_{\m{fib}}^2$,
$D = \IC_{\m{fib},2}^1$, and $E = \IC_{\m{fib},2}^2$,
which yields
\begin{align}
Q(0,cy,N,sy) &> cy
\label{8.7} \\
Q(cx,cy,S(sx),sy) &> Q(cx,cy,sx,sy)
\label{8.8} \\
G(0,N) &> 0
\label{8.9}
\intertext{for the unconditional rules and}
H(cx,S(sx)) &\geqslant D(cx,G(cx,sx) + {} \notag \\
&H(cx,sx),S(sx),F(sx))
\label{8.10} \\
D(cx,cy+cz,S(sx),P(sy,sz)) &\geqslant {} \notag \\
E(cx,cy+cz,Q(cy,cz,sy,sz),S(sx),&P(sy,sz),A(sy,sz))
\label{8.11} \\
G(cx,S(sx)) + E(cx,cy+cz,cw,S(sx),P(sy,sz),sw) &> cz + cw
\label{8.12}
\end{align}
for the conditional rule.
Here, $Q$ is strictly monotone in its
first two arguments and weakly in the last two, $D$ is
strictly monotone in argument $2$ and weakly in $4$, while $E$ is
strictly monotone in argument $3$ and weakly in $6$.
There is no monotonicity constraint for $G$ or $H$.

Choosing minimal polynomials to satisfy the 
constraints deriving from the rules for $+$,
we set $N = 0$, $S(x) = x + 1$, $A(x,y) = x + y$, and
$Q(cx,cy,sx,sy) = cx + cy + sx + 1$.
Since $G$ need not be monotone, we simply take $G(x,y) = 1$
to satisfy \eqref{8.9}. Further choosing
$P(x,y) = x+y$, the constraints simplify to
\begin{align}
F(0) &\geqslant 1
\tag{\ref{8.3}} \\
F(x+1) &\geqslant B(x+1,F(x))
\tag{\ref{8.4}} \\
B(x+1,y+z) &\geqslant C(x+1,y+z,y+z)
\tag{\ref{8.5}} \\
C(x+1,y+z,w) &\geqslant z+w
\tag{\ref{8.6}} \\
\qquad\qquad\qquad\qquad\qquad\qquad H(cx,sx+1) &\geqslant D(cx,H(cx,sx)+1,sx+1,F(sx))
\tag{\ref{8.10}} 
\end{align}

\begin{align}
D(cx,cy+cz,sx+1,sy+sz) &\geqslant \notag \\
E(cx,cy+cz,cy&{}+cz+sy+1,sx+1,sy+sz,sy+sz)
\tag{\ref{8.11}} \\
1 + E(cx,cy+cz,cw,sx+1,sy+sz,sw) &> cz + cw
\tag{\ref{8.12}}
\end{align}
The size constraints are satisfied if we choose
$C(x,y,z) = y + z$, $B(x,y) = 2y$, and $F(x) = 2^x$.
Choosing $E(cx,cy,cz,sx,sy,sz) = cy + cz$ and
$D(cx,cy,sx,sy) = 2cy + sy + 1$
takes care of \eqref{8.11} and~\eqref{8.12},
leaving only
\begin{align}
H(c,s+1) &\geqslant 2 \cdot (H(c,s) + 1) + 2^s + 1
\tag{\ref{8.10}}
\end{align}
This final constraint is satisfied for
$H(c,s) = (s+1) \cdot (2^{s+1} - 2)$ since
\begin{align*}
H(c,s+1)
&= (s+2) \cdot (2^{s+2} - 2)
= s \cdot 2^{s+2} + 8 \cdot 2^s - 2s - 4 \\
&= s \cdot 2^{s+2} + 5 \cdot 2^s - 2s - 4 + 3 \cdot 2^s
\geqslant s \cdot 2^{s+2} + 5 \cdot 2^s - 4s - 4 + 3 \\
&= 2 \cdot (s+1) \cdot 2^{s+1} - 4 \cdot (s+1) + 2^s + 3 
= 2 \cdot (s+1) \cdot (2^{s+1}-2) + 2^s + 3 \\
&= 2 \cdot H(c,s) + 2^s + 3 =
2 \cdot (H(c,s) + 1) + 2^s + 1
\end{align*}
Since all ground constructor terms $s$ have cost $0$ and size
at most $|s|$, for ground basic terms $s$ with $|s| \leqslant n$,
$\Cost(\interpret{s}{})$ is bounded by
$G(0,n-1) + H(0,n-1) = 1 + n \cdot 2^n - 2n$.
We conclude a runtime complexity of $\OO(n \cdot 2^n)$
by Lemma~\ref{lem:recipe3}.
\end{exa}

\section{Conclusions}
\label{sec:conclusion}

In this paper we have improved and extended the notion of complexity
for conditional term rewriting first introduced in \cite{KMS15}. This
notion takes failed calculations into account as any automatic
rewriting engine would. We have defined a transformation to 
unconditional left-linear context-sensitive TRSs whose complexity is the
same as the conditional complexity of the original system, and shown
how this transformation can be used to find bounds for conditional
complexity using traditional interpretation-based methods.

\subsection{Implementation and Experiments}
\label{subsec:experiments}

At present, we have not implemented the results of
Sections~\ref{sec:polynomial}, \ref{sec:ucs},
and~\ref{sec:splitsize}.
However, we did
implement the transformation from Section~\ref{sec:transformation}.
The resulting (context-sensitive) TRSs can be used as input to a
conventional TRS complexity tool, which by
Theorem~\ref{thm:transformsound} gives an upper bound for
conditional complexity. Although existing tools do not take advantage
of either information regarding the replacement map, nor of the
specific shape of the rules or the fact that only terms of the form
$\xi_\top(s)$ need to be considered, the results are often tight
bounds.

\smallskip

\setlength{\intextsep}{0mm}
\begin{wrapfigure}[6]{r}{0.18\textwidth}
\begin{tabular}{|l|r|@{}}
\hline
$\OO(1)$ & 16 \\
$\OO(n)$ & 10 \\
$\OO(n^2)$ & 3 \\
MAYBE & 33 \\
\hline
\end{tabular}
\end{wrapfigure}
We have used this approach with \TCT~\cite{MM13} as the
underlying complexity tool, to analyze
the runtime complexity of the 57 strong CCTRSs in the
current version of the termination problem database (TPDB 10.3),%
\footnote{See \url{http://termination-portal.org/wiki/TPDB} for more
details.}
along with 5 examples in this paper. The results are summarized
to the right. A full evaluation page is available at
\begin{center}
\url{http://cl-informatik.uibk.ac.at/experiments/2016/cc}
\end{center}
About half of the systems in our example set could not be
handled. This is largely due to the presence of non-terminating
CCTRSs as well as systems with
exponential runtime complexity, which existing complexity tools do
not support. Many benchmarks of conditional rewriting have rules
similar to our Example~\ref{ex:evenodd}, which lead to exponential
complexity due to failed evaluations, and consequently cannot be
handled. We do, however, obtain
a constant upper bound for Example~\ref{ex:poschoice}, a quadratic upper
bound for Example~\ref{ex:whenstop}, as well as the
tight bound $\OO(n)$ for Example~\ref{exa:odd}.

\subsection{Related Work}
\label{subsec:related work}

We are not aware of any other attempt to study the complexity of
conditional rewriting, but
numerous transformations from CTRSs to TRSs have been proposed in the
literature. They can roughly be divided into so-called
\emph{unravelings} and \emph{structure-preserving} transformations. The
former were coined by Marchiori~\cite{M96} and have been extensively
investigated (e.g.\ \cite{M97,NSS12,O99,O02,SG10}), mainly to
establish (operational) termination and confluence of the input CTRS.
The latter originate from Viry~\cite{V99} and improved versions were
proposed in \cite{ABH03,SR06,GN14}.

The transformations that are known to transform
CTRSs into TRSs such that (simple) termination of the latter 
implies quasi-decreasingness of the former, are natural candidates 
for study from a complexity perspective. We observe that unravelings are
not suitable in this regard, since they do not take the cost
for failed computations into account.
For instance, the unraveling from \cite{M97}
transforms the CCTRS $\RR_\m{even}$ into
\begin{xalignat*}{3}
\m{even}(0) &\to \m{true} &
\m{even}(\m{s}(x)) &\to \m{U_1}(\m{odd}(x),x) &
\m{U_1}(\m{true},x) &\to \m{true}
\\
&&
\m{even}(\m{s}(x)) &\to \m{U_2}(\m{even}(x),x) &
\m{U_2}(\m{true},x) &\to \m{false}
\\
\m{odd}(0) &\to \m{false} &
\m{odd}(\m{s}(x)) &\to \m{U_3}(\m{odd}(x),x) &
\m{U_3}(\m{true},x) &\to \m{false}
\\
&&
\m{odd}(\m{s}(x)) &\to \m{U_4}(\m{even}(x),x) &
\m{U_4}(\m{true},x) &\to \m{true}
\end{xalignat*}
This TRS has a linear runtime complexity, which is readily confirmed by
\TCT. As the conditional runtime complexity is exponential, the
transformation is not suitable for measuring conditional complexity.
The same holds for the transformation in \cite{O99}.

Structure-preserving transformations are better suited for studying
conditional complexity since they keep track of the conditions in all
applicable rules.

However, existing transformations of this kind are also unsuitable
for measuring conditional runtime complexity. For instance, the CCTRS
$\RR_{\m{even}}$ is transformed into the TRS
\begin{xalignat*}{3}
\m{even}(0,x,y) &\to \m{m}(\m{true})
&
\m{odd}(0,x,y) &\to \m{m}(\m{false})
\\
\m{even}(\m{s}(x),\bot,z)
&\to \m{even}(\m{s}(x),\m{c}(\m{m}(\m{odd}(x,\bot,\bot))),z)
&
\m{even}(\m{s}(x),\m{c}(\m{m}(\m{true})),z) &\to \m{m}(\m{true})
\\
\m{even}(\m{s}(x),y,\bot)
&\to \m{even}(\m{s}(x),y,\m{c}(\m{m}(\m{even}(x,\bot,\bot))))
&
\m{even}(\m{s}(x),y,\m{c}(\m{m}(\m{true})))
&\to \m{m}(\m{false})
\\
\m{odd}(\m{s}(x),\bot,z)
&\to \m{odd}(\m{s}(x),\m{c}(\m{m}(\m{even}(x,\bot,\bot))),z)
&
\m{odd}(\m{s}(x),\m{c}(\m{m}(\m{true})),z)
&\to \m{m}(\m{true})
\\
\m{odd}(\m{s}(x),y,\bot)
&\to \m{odd}(\m{s}(x),y,\m{c}(\m{m}(\m{odd}(x,\bot,\bot))))
&
\m{odd}(\m{s}(x),y,\m{c}(\m{m}(\m{true})))
&\to \m{m}(\m{false})
\\
\m{even}(\m{m}(x),y,z) &\to \m{m}(\m{even}(x,\bot,\bot))
&
\m{s}(\m{m}(x)) &\to \m{m}(\m{s}(x))
\\
\m{odd}(\m{m}(x),y,z) &\to \m{m}(\m{odd}(x,\bot,\bot))
&
\m{m}(\m{m}(x)) &\to \m{m}(x)
\end{xalignat*}
\newcommand{\even}{\m{s}^{21}(\m{0}) + \m{s}^{21}(\m{0})}%
by the transformation of
\c{S}erb\u{a}nu\c{t}\u{a} and Ro\c{s}u~\cite{SR06}.
\TCT reports a constant
runtime complexity, which is explained by the fact that the symbol
$\m{s}$ is turned into a defined symbol. Hence a term like
$\m{even}(\m{s}(\m{0}),\top,\top)$ is 
\emph{not} basic and thus
disregarded for runtime complexity.
The derivational complexity of the transformed TRS is harder to
confirm automatically, as it is exponential, but likely not to differ
much from the conditional derivational complexity of $\RR_\m{even}$.
However, in general, we may well obtain much greater bounds due
to the forced reevaluation of conditions when a subterm is reduced.
Consider for instance a term $\m{even}(\m{s}(t))$ with $t = \even$
in an extension of $\RR_{\m{even}}$ with rules for $+$. This term is
encoded as $\m{even}(\m{s}(t),\bot,\bot)$, the $\bot$s indicating that no
condition has been evaluated yet, and might be reduced as follows:
\begin{align*}
\m{even}(\m{s}(t),\bot,\bot)
&\to\phantom{^*}
\m{even}(\m{s}(t),\m{c}(\m{m}(\m{odd}(t,\bot,\bot))),\bot) \\
&\to^*
\m{even}(\m{s}(t),\m{c}(\m{m}(\m{false})),\bot) \\
&\to\phantom{^*}
\m{even}(\m{s}(t),\m{c}(\m{m}(\m{false})),
\m{c}(\m{m}(\m{even}(t,\bot,\bot)))) \\
&\to^*
\m{even}(\m{s}(t),\m{c}(\m{m}(\m{false})),\m{c}(\m{m}(\m{true}))) \\
&\to\phantom{^*}
\m{even}(\m{s}(\m{m}(\m{s}^{42}(\m{0}))),\m{c}(\m{m}(\m{false})),
\m{c}(\m{m}(\m{true}))) \\
&\to\phantom{^*}
\m{even}(\m{m}(\m{s}^{43}(\m{0})),\m{c}(\m{m}(\m{false})),
\m{c}(\m{m}(\m{true}))) \\
&\to\phantom{^*} \m{m}(\m{even}(\m{s}^{43}(\m{0}),\bot,\bot))
\end{align*}
We observe that
an evaluation in the instance $\m{s}(t)$ of the pattern $\m{s}(x)$
forces a reevaluation of $t$
when checking the second condition.
The fundamental difference with our approach is that we have used
Lemma~\ref{lem:local} to avoid such reevaluations.

Less recent, the transformation of
Antoy \textit{et al.}~\cite{ABH03} operates in
a more restrictive setting: weakly orthogonal constructor-based CTRSs
without extra variables in the conditions. Like the
transformation in~\cite{SR06}, it
blocks conditions when their evaluation fails; however,
conditions are not reevaluated when arguments are modified.
A crucial difference with our transformation $\Xi$ is that
different conditions in the same
conditional rule are not evaluated from left to right but combined
into a single condition, which has a negative impact on complexity.
As an extreme example, consider the CCTRS $\RR$ consisting of the four
rules
\begin{xalignat*}{2}
\m{f}(x) &\to \m{a} \,\Leftarrow\,
\m{c} \approx \m{d},~ \m{g}(x) \approx \m{a},~ \m{g}(x) \approx \m{b} &
\m{g}(\m{s}(x)) &\to \m{f}(x)
\\
\m{f}(x) &\to \m{b} \,\Leftarrow\, \m{c} \approx \m{e} &
\m{c} &\to \m{e}
\end{xalignat*}
The conditional runtime complexity of $\RR$ is linear, which is
confirmed by running \TCT\ on $\Xi(\RR)$. The transformation of
\cite{ABH03} produces the TRS
\begin{xalignat*}{3}
\m{f}(x,\bot,\bot) &\to
\m{f}(x,\langle \m{c}, \m{g}(x), \m{g}(x) \rangle,\m{c}) &
\m{f}(x,\langle \m{d}, \m{a}, \m{b} \rangle,z) &\to \m{a} &
\m{g}(\m{s}(x)) &\to \m{f}(x,\bot,\bot)
\\
&& \m{f}(x,y,\m{e}) &\to \m{b} &
\m{c} &\to \m{e}
\end{xalignat*}
whose runtime complexity is at least exponential due of the
rules $\m{f}(x,\bot,\bot) \to
\m{f}(x,\langle \m{c}, \m{g}(x),\linebreak \m{g}(x) \rangle,\m{c})$
and $\m{g}(\m{s}(x)) \to \m{f}(x,\bot,\bot)$. If the (undecidable)
weak orthogonality restriction in \cite{ABH03} is not imposed, the same
phenomenon may occur if rules have at most one condition.

However, it is worth noting also the similarities to our
method, especially when there is at most one condition.
Consider for example the result of transforming our CCTRS
$\RR_{\m{even}}$:
\begin{xalignat*}{3}
  \m{even}(\m{0},y,z) &\to \m{true} &
  \m{even}(\m{s}(x),\m{true},y) &\to \m{true} &
  \m{even}(\m{s}(x),y,\m{true}) &\to \m{false} \\
  \m{odd}(\m{0},y,z) &\to \m{false} &
  \m{odd}(\m{s}(x),\m{true},y) &\to \m{true} &
  \m{odd}(\m{s}(x),y,\m{true}) &\to \m{false} \\
  \m{even}(\m{s}(x),\bot,\bot) &\to
  \rlap{$\m{even}(\m{s}(x),\m{odd}(x,\bot,\bot),\m{even}(x,\bot,\bot))$} \\
  \m{odd}(\m{s}(x),\bot,\bot) &\to
  \rlap{$\m{odd}(\m{s}(x),\m{even}(x,\bot,\bot),\m{odd}(x,\bot,\bot))$}
\end{xalignat*}
This does not look too different
from the result of our transformation $\Xi$ if the set $\m{AP}$ is not
used. In addition, the method used could be generalised with some of the
ideas from~\cite{SR06}, for instance by evaluating multiple conditions
sequentially rather than in parallel.

Even ignoring the issue of multiple
conditions---or, for~\cite{SR06}, the issue of
reevaluation---there are some fundamental differences between our
transformation $\Xi$ and the structure-preserving transformations
of~\cite{ABH03,SR06}. In both of these,
the conditions for different rules may be evaluated in parallel, which we
do not permit. Moreover, neither transformation separates defined symbols
(e.g.\ $\m{even}$) from ``active'' symbols used to evaluate conditions
(e.g.\ $\m{even}_2^1$). This separation is necessary to impose a
context-sensitive replacement map as we have done here, and
makes it much easier to use traditional techniques such
as polynomial interpretations.
Most importantly, neither transformation \emph{defines}---or is
based on a formal definition of---conditional complexity; rather, they
define upper bounds for a reasonable evaluation strategy.

\subsection{Avenues for Future Work}
\label{subsec:future work}\enlargethispage{\baselineskip}

There are several possibilities to continue our research.

\subparagraph*{\textbf{Weakening restrictions.}}
An obvious direction for future research is to broaden the class of
CTRSs we consider. While it would make little sense
to consider CTRSs that are not deterministic or of type 3---as the
rewrite relation in these systems is undecidable in general---it may be
possible to drop the variable and constructor requirements.

The linearity requirements in \emph{strong} CCTRSs are an
obvious target for improvement. These requirements were not needed in the
definition or justification of our primary complexity notion, but
essential
for the correctness of the way we use the anti-pattern set $\m{AP}$.
However, if we are willing to lose completeness, we may drop the
anti-pattern set, replacing the use of $v$ in $\m{AP}(\ell_i)$ or
$\m{AP}(b_j)$ in Definition~\ref{def:transformrules} by a fresh
variable; doing so, the transformation would not preserve derivation
heights, but we would retain the possibility to obtain \emph{upper}
bounds. Alternatively, we might consider an infinite set of
transformed rules $\Xi'(\RR)$ instead.

As for the restrictions in general CCTRSs, the proof of
the important locality Lemma~\ref{lem:local} requires only that the
left-hand side $\ell$ of every rule $\ell \to r \Leftarrow c$ is a basic
term such that
$\Var(\ell) \cap \Var(c) = \varnothing$.
This can always be satisfied by altering the system
without changing the rewrite relation in an essential way,
replacing for instance $\m{f}(\m{g}(x),y) \to r$ by
$\m{f}(z,y) \to r \,\Leftarrow\, z \approx \m{g}(x)$.
However, in such cases, the definition of conditional complexity
needs to be revisited, as the restrictions on the conditions are
needed for Lemma~\ref{lem:whatisfailure}, which is important to
justify our complexity notion.
For example, if the right-hand sides of conditions were allowed to
be arbitrary terms, it would be possible to define a system with rules
\begin{xalignat*}{4}
\m{g}(x) &\to x & \m{h}(x) &\to \m{g}(x) & \m{h}(x) &\to x & 
\m{f}(z,y) &\to \m{a} \,\Leftarrow\, z \approx \m{g}(x)
\end{xalignat*}
In this CTRS, a term $\m{f}(\m{h}(\m{0}),\m{0})$ \emph{can} be
reduced by the last rule, but we would only find this out if we reduced
$\m{h}(\m{0})$ with the second rule, rather than with the third.
Thus, to accurately analyze such a system, we would likely need a
backtracking mechanism.
To drop the restriction that the right-hand sides of
conditions may not repeat variables, we would need the same, or
alternatively a strategy which enforces that left-hand sides of
conditions must always be reduced to normal form.
Similar revisions could be used to extend the definition to take
non-confluence into account, as discussed at the end of
Section~\ref{sec:analysis}.

Alternatively, we could weaken the restrictions only partially,
allowing for instance irreducible patterns---terms $b$ such that for
no instance $b\gamma$, a reduction step is possible at a position in
$\Pos(b)$--- as right-hand sides of conditions rather than only
constructor terms.

\subparagraph*{\textbf{Rules with branching conditions.}}
Consider the following variant of $\RR_{\m{even}}$:

\noindent
\hspace{-2mm}
\begin{minipage}{.5\textwidth}
\setcounter{equation}{0}
\begin{align}
\m{even}(\m{0}) &\to \m{true}
\\
\m{even}(\m{s}(x)) &\to \makebox[7mm]{$\m{true}$} \,\Leftarrow\,
 \m{odd}(x) \approx \m{true}
 \label{even:true}
\\
\m{even}(\m{s}(x)) &\to \makebox[7mm]{$\m{false}$} \,\Leftarrow\,
 \m{odd}(x) \approx \m{false}
 \label{even:false}
\end{align}
\end{minipage}
\begin{minipage}{.5\textwidth}
\begin{align}
\m{odd}(\m{0}) &\to \m{false}
\\
\m{odd}(\m{s}(x)) &\to \makebox[7mm]{$\m{true}$} \,\Leftarrow\,
 \m{even}(x) \approx \m{true}
 \label{odd:true}
\\
\m{odd}(\m{s}(x)) &\to \makebox[7mm]{$\m{false}$} \,\Leftarrow\,
 \m{even}(x) \approx \m{false}
 \label{odd:false}
\end{align}
\end{minipage} \\

\noindent
Unlike Example~\ref{ex:evenodd}, rules \eqref{even:true} and
\eqref{even:false}, and rules \eqref{odd:true} and \eqref{odd:false}
have very similar conditions. Currently, we do not exploit this.
Evaluating $\m{even}(\m{s}^9(\m{0}))$ with 
rule \eqref{even:true} causes the calculation of the normal form 
$\m{false}$ of $\m{odd}(\m{s}^8(\m{0}))$, before concluding that the
rule does not apply. In our definitions (of $\rightharpoonup$ and
$\Xi$), and in line with
the behavior of Maude, we would dismiss the
result and continue
trying the next rule. In this case, that means
recalculating the normal form of $\m{odd}(\m{s}^8(\m{0}))$, but now
to verify whether rule \eqref{even:false} applies.

This is wasteful, as there is clearly no benefit in recalculating
this normal form.
The rules are defined in a \emph{branching}
manner: If the condition evaluation gives one result, we should
apply rule
\eqref{even:true};
if it gives another, we should use rule
\eqref{even:false}.
A clever rewriting engine could use this branching,
and avoid recalculating obviously unnecessary results. Thus, future
extensions of the complexity notion might take such \emph{groupings}
of rules into account.

\subparagraph*{\textbf{Improving the transformation.}}
With regard to the transformation $\Xi$, it
is would be easy to obtain smaller resulting systems using
various optimizations, such as reducing the set $\m{AP}$ of anti-patterns
using typing considerations, or leaving defined symbols untouched when
they are only defined by unconditional rules.

\subparagraph*{\textbf{Implementation and further complexity methods.}}
The strength of our implementa\-tion---which relies simply on a
transformation to unconditional complexity---is necessarily limited
by the possibilities of existing complexity tools. Thus, we hope
that, in the future, developers of complexity tools will branch out
towards context-sensitive rewriting. Moreover, we encourage
developers to add support for exponential upper bounds.

To take full advantage of the initial conditional setting, it would be
ideal for complexity tools to directly support
conditional rewriting. This would enable tools to use methods like
Recipe~\ref{def:recipeC}, which uses a $\max$-interpretation to
immediately eliminate a large number of rules---an interpretation
which an automatic tool is unlikely to find by itself.
It is likely that other, non-interpretation-based methods, can be
optimized for the conditional setting as well.

\section*{Acknowledgement}

We thank the reviewers for their detailed comments, which led to many
improvements.

\bibliographystyle{plain}
\bibliography{references}

\appendix

\section{Proof of Theorem~\ref{thm:transformcomplete}}
\label{app:transformcomplete}

Recall the statement of Theorem~\ref{thm:transformcomplete}:

\medskip

\begin{quote}
Let $\RR$ be a strong CCTRS and $s \in \TT(\GG)$.
If $\zeta(s)$ is terminating and there exists a
context-sensitive reduction $\zeta(s) \to_{\Xi(\RR),\mu}^* t$ for
some $t$ with cost $N$, then there exists a complexity-conscious
reduction $s \xrightharpoonup{}^* t'$ with cost at least $N$.
If there exists an infinite $(\Xi(\RR),\mu)$ reduction
starting from $\zeta(s)$ then $s \xrightharpoonup{\infty\,}$.
\end{quote}

\medskip

\noindent
In this appendix we present the proof. We fix a strong CCTRS
$\RR$, and corresponding signatures $\FF$, $\GG$ and $\HH$.
In order to relate
certain reduction sequences in $(\Xi(\RR),\mu)$ to complexity-conscious
reductions with
$\xrightharpoonup{}$, we start by defining an inverse of $\zeta$.

\begin{defi}
A term $s \in \TT(\HH,\VV)$ is \emph{proper} if
\begin{itemize}
\item
$s$ is a variable, or
\item
$s = f(\seq{s})$ with $f$ a constructor symbol and proper subterms
$\seq{s}$, or
\item
$s = f(\seq{s},\seq[m_f]{c})$ with $f$ a defined symbol,
proper subterms $\seq{s}$, and $\seq[m_f]{c} \in \{ \bot, \top \}$.
\end{itemize}
We denote the set of all proper (ground) terms by $\TTp(\HH,\VV)$
($\TTp(\HH)$).
For proper terms $s$ we define
$\zeta^{-}(s) \in \TT(\GG,\VV)$
as follows.
If $s$ is a variable then $\zeta^-(s) = s$,
if $s = f(\seq{s})$ with $f$ a constructor then
$\zeta^-(s) = f(\zeta^-(s_1),\dots,\zeta^-(s_n))$, and
if $s = f(\seq{s},\seq[m_f]{c})$ with $f$ a defined symbol then
$\zeta^-(s) = f_R(\zeta^-(s_1),\dots,\zeta^-(s_n))$
for $R = \{ \smash{\rho^f_i} \mid c_i = \top \}$.
\end{defi}

Note that $\bot$-patterns (Definition~\ref{def:xi}) are proper.
The following lemma collects some easy properties of $\zeta^-$.

\begin{lem}
\label{lem:inverseprops}
\begin{enumerate}
\item\label{lem:inverseprops:inverseouter}
If $s \in \TT(\GG,\VV)$ then
$\zeta(s) \in \TTp(\HH,\VV)$ and $\zeta^-(\zeta(s)) = s$.
\item\label{lem:inverseprops:inverseinner}
If $t \in \TTp(\HH,\VV)$ then $\zeta(\zeta^-(t)) = t$.
\item\label{lem:inverseprops:substitution}
If $t \in \TTp(\HH,\VV)$ and $\tau\colon \VV \to \TTp(\HH,\VV)$
then $t\tau \in \TTp(\HH,\VV)$ and
$\zeta^-(t\tau) = \zeta^-(t)\tau_{\zeta^-}$
(where $\tau_{\zeta^-} = \zeta^- \circ \tau$).
\item\label{lem:inverseprops:top}
If $u \in \TT(\FF,\VV)$ and $\tau\colon \VV \to \TTp(\HH,\VV)$ then
$\xi_\top(u)\tau \in \TTp(\HH,\VV)$ and
$\zeta^-(\xi_\top(u)\tau) = \m{label}(u)\tau_{\zeta^-}$.
\item\label{lem:inverseprops:antipattern}
If $v \in \m{AP}(u)$ for some linear constructor term $u$
then $v \in \TTp(\HH,\VV)$ and $\zeta^-(v)$ is a
linear labeled normal form which does not unify with $u$.
\end{enumerate}
\end{lem}

\begin{proof}
The first three statements are proved by an obvious induction argument.
\begin{enumerate}
\setcounter{enumi}{3}
\item
We have $\xi_\top(u) = \zeta(\m{label}(u))$ by Lemma~\ref{zeta xi}. From
statements (\ref{lem:inverseprops:substitution}) and
(\ref{lem:inverseprops:inverseouter}) we infer
$\zeta(\m{label}(u))\tau \in \TTp(\HH,\VV)$ and
$\zeta^-(\zeta(\m{label}(u))\tau) =
\zeta^-(\zeta(\m{label}(u)))\tau_{\zeta^-} =
\m{label}(u)\tau_{\zeta^-}$.
\smallskip
\item
From the definition of $\m{AP}$ it follows that $v$ is a $\bot$-pattern
and thus proper. By structural induction on $v$ we easily obtain
that $\zeta^-(v)$ is a linear labeled normal form
which does not unify with $u$.
\qedhere
\end{enumerate}
\end{proof}

An important preliminary result is that terminating proper ground terms
have a $\bot$-pattern as normal form. This allows us to eliminate
$f_i^j$ symbols in selected (sub)terms, which is crucial for
transforming a $(\Xi(\RR),\mu)$ reduction into a complexity-conscious
reduction.

\begin{lem}
\label{lem:bottomnormalforms}
If $s \in \TTp(\HH)$ then any $(\Xi(\RR),\mu)$ normal form of $s$ is a
$\bot$-pattern.
\end{lem}

\begin{proof}
For the purpose of this proof, a ground term $u$ in $\TT(\HH)$ is said
to be an \emph{intermediate} term if
\begin{itemize}
\item
$u = f(\seq{u})$ with $f$ a constructor symbol and intermediate
arguments
$\seq{u}$, or
\item
$u = f(\seq{u},\seq[m_f]{c})$ with $f$ a defined symbol,
$\seq[m_f]{c} \in \{ \bot, \top \}$, and intermediate arguments $\seq{u}$,
or
\item
$u = f_i^j(\seq{\ell},\langle \seq[m_f]{c} \rangle[\seq[j-1]{b},v]_i)
\sigma$ with $\seq[m_f]{c} \in \{ \bot, \top \}$ and
intermediate terms $v$ and $\sigma(y)$ for all
$y \in \Var(\seq{\ell},\seq[j-1]{b})$,
whenever $\smash{\rho^f_i}\colon f(\seq{\ell}) \to r \Leftarrow c$
and $1 \leqslant j \leqslant k$. (Note that $v\sigma = v$ since
intermediate terms are ground.)
\end{itemize}
We use $\TTi(\HH)$ to denote the set of intermediate terms.
The following properties are easily established:
\begin{enumerate}[label=\({\alph*}]
\item
proper ground terms are intermediate terms,
\item
if $u$ is proper and the domain of $\sigma\colon \VV \to \TTi(\HH)$
includes $\Var(u)$ then $u\sigma$ is an intermediate term,
\item
if $u$ is proper and $u\sigma$ an intermediate term then
$\sigma(x)$ is an intermediate term for every $x \in \Var(u)$.
\end{enumerate}
Next we prove that intermediate terms are closed under
$(\Xi(\RR),\mu)$ reduction. So let $u \in \TTi(\HH)$ and
$u \to_{\Xi(\RR),\mu} u'$. We use induction on the size of $u$. 
\begin{itemize}
\item
Suppose $u = f(\seq{u})$ with $f$ a constructor symbol and
intermediate
arguments $\seq{u}$. The reduction step from $u$ to $u'$ must take place
in one of the arguments, so $u' = f(\seq[i-1]{u},u_i',u_{i+1},\dots,u_n)$
for some $1 \leqslant i \leqslant n$ with $u_i \to_{\Xi(\RR),\mu} u_i'$.
The term $u_i'$ is intermediate according to the induction hypothesis.
Hence $u'$ is intermediate by definition.
\item
Suppose $u = f(\seq{u},\seq[m_f]{c})$ with $f$ a defined symbol,
$\seq[m_f]{c} \in \{ \bot, \top \}$, and intermediate arguments $\seq{u}$.
If the reduction step takes place in one of the arguments
$\seq{u}$, we reason as in the case above.
Suppose the step takes place at the root. We distinguish three
subcases, depending on which kind of rule of $\Xi(\RR)$ is used.
\smallskip
\begin{enumerate}
\item
If a rule of type \eqref{xi1} is used then $u' = \xi_\top(r)\sigma$ for
some right-hand side of an unconditional rule $\ell \to r$ in
$\RR{\restriction}f$ such that $\ell\sigma = f(\seq{u})$. From
property \(c\ we infer that $\sigma(y)$ is intermediate for all
$y \in \Var(\ell)$. Since $\Var(r) \subseteq \Var(\ell)$ and
$\xi_\top(r)$ is proper by
Lemma~\ref{lem:inverseprops}(\ref{lem:inverseprops:top}),
$u'$ is intermediate by property \(b.
\smallskip
\item
If a rule of type \eqref{xi2} is used then
$u = f(\ell_1\sigma,\dots,\ell_n\sigma,\seq[m_f]{c})$ such that
$c_i = \top$ for some $1 \leqslant i \leqslant n$ with
$\rho_i\colon f(\seq{\ell}) \to r \Leftarrow c$ in $\RR{\restriction}f$.
We have $u' = f_i^j(\seq{\ell},\langle \seq[m_f]{c} \rangle
[\xi_\top(a_1)]_i)\sigma$ and from property \(c\ we infer that
$\sigma(y)$ is intermediate for all $y \in \Var(\ell)$. Since
$\Var(a_i) \subseteq \Var(\ell)$, the term $\xi_\top(a_1)\sigma$ is
intermediate by property \(b\ and thus also ground. Hence 
$u' = f_i^j(\seq{\ell},\langle \seq[m_f]{c} \rangle
[\xi_\top(a_1)\sigma]_i)\sigma$, which is of the required shape to be
intermediate.
\smallskip
\item
The final possibility is that a rule of type \eqref{xi6} is used.
In this case we have
$u' = f(\seq{u},\langle \seq[m_f]{c} \rangle[\bot]_i)$ for some
$1 \leqslant i \leqslant m_f$.  Since the arguments $\seq{u}$ are
intermediate, $u'$ is intermediate by definition.
\end{enumerate}
\smallskip
\item
Suppose $u = f_i^j(\seq{\ell},\langle\seq[m_f]{c}\rangle
[\seq[j-1]{b},v]_i)\sigma$. If the reduction step from $u$ to $u'$
takes place below the root, it must take place in 
$v\sigma = v$, due to restrictions on the replacement map $\mu$.
Hence the result follows from the induction hypothesis. Suppose the step
takes place at the root.
Note that the rule $\rho_i^f\colon f(\seq{\ell}) \to r \Leftarrow c$
must exist in $\RR{\restriction}f$.
We again distinguish three
subcases, depending on which kind of rule of $\Xi(\RR)$ is used.
\smallskip
\begin{enumerate}
\item
If a rule of type \eqref{xi3} is used then
$u = f_i^j(\seq{\ell},\langle \seq[m_f]{c}\rangle[\seq[k]{b}]_i)\tau$,
and $u' = \xi_\top(r)\tau$
for some substitution $\tau$
with $\dom(\tau) \subseteq \Var(\seq{\ell},\seq[k]{b})$.
Hence
$\ell_l\tau = \ell_l\sigma$ for all $1 \leqslant l \leqslant n$,
$b_l\tau = b_l\sigma$ for all $1 \leqslant l < j = k$, and
$v = b_k\tau$.
From property \(c\ we infer that $\sigma(y) = \tau(y)$ is intermediate
for all
$y \in \Var(\seq{\ell},\seq[k]{b}) \supseteq \Var(r)$. Hence
$u'$ is intermediate by \(b\ since
$\xi_\top(r)$ is proper by
Lemma~\ref{lem:inverseprops}(\ref{lem:inverseprops:top}).
\smallskip
\item
If a rule of type \eqref{xi4} is used then
$u = f_i^j(\seq{\ell},\langle \seq[m_f]{c}\rangle[\seq[j]{b}]_i)\tau$,
$j < k$, and $u' = f_i^{j+1}(\seq{\ell},\langle \seq[m_f]{c} \rangle
[\seq[j]{b},\xi_\top(a_{j+1})]_i)\tau$ for some substitution $\tau$
with $\dom(\tau) \subseteq \Var(\seq{\ell},\seq[j]{b})$.
Hence $\ell_l\tau = \ell_l\sigma$ for all $1 \leqslant l \leqslant n$,
$b_l\tau = b_l\sigma$ for all $1 \leqslant l < j$, and
$v = b_j\tau$. Therefore,
\[u' = f_i^{j+1}(\seq{\ell},\langle \seq[m_f]{c} \rangle
[\seq[j]{b},\xi_\top(a_{j+1}))\tau]_i\sigma\] and this suffices, if
$\xi_\top(a_{j+1})\tau$ is an intermediate term. This follows from
$\Var(a_{j+1}) \subseteq \Var(\seq{\ell},\seq[j]{b})$ together
with Lemma~\ref{lem:inverseprops}(\ref{lem:inverseprops:top}) and
properties \(b\ and \(c.
\smallskip
\item
The final possibility is that a rule of type \eqref{xi5} is used.
In this case we have
$u' = f_i^j(\seq{\ell},\langle\seq[m_f]{c}\rangle[\bot]_i)\sigma$ for
some $1 \leqslant i \leqslant m_f$.  As 
$\ell_1\sigma, \dots, \ell_n\sigma$ are intermediate, $u'$ is
intermediate by definition.
\end{enumerate}
\end{itemize}
Now suppose that $s$ has a normal form $t$ in
$(\Xi(\RR),\mu)$. We already know that $t$ is an intermediate term.
So it suffices to show that intermediate terms in normal form
are $\bot$-patterns. We show instead that any intermediate term
$t$ which is not a $\bot$-pattern is reducible, by induction on its
size.
\begin{itemize}
\item
Suppose $t = f(\seq{t})$ with $f$ a constructor symbol. One of the
arguments, say $t_i$, is not a $\bot$-pattern. The induction hypothesis
yields the reducibility of $t_i$. Since $i \in \mu(f)$, $t$ is reducible
as well.
\item
Suppose $t = f(\seq{t},\seq[m_f]{c})$ with $f$ a defined symbol and
$\seq[m_f]{c} \in \{ \bot, \top \}$. If one of the terms $\seq{t}$
is not a $\bot$-pattern, we reason as in the previous case.
Otherwise, $c_i = \top$ for some $1 \leqslant i \leqslant m_f$.
Consider $\rho_i^f\colon f(\seq{\ell}) \to r \Leftarrow c$. 
If $f(\seq{t})$ is an instance of $f(\seq{\ell})$ then $t$ is reducible
by rule \eqref{xi1} or \eqref{xi2}. If $f(\seq{t})$ is not an instance
of $f(\seq{\ell})$ then, using the linearity of $f(\seq{l})$,
there exists an argument position $1 \leqslant j \leqslant n$ such that
$t_j$ is not an instance of $\ell_j$. According to Lemma~\ref{AP lemma}
$t_j$ is an instance of an anti-pattern in $\m{AP}(\ell_j)$. Consequently,
$t$ is reducible by rule \eqref{xi6}.
\item
The final case is
$t = f_i^j(\vec{\ell},\langle \seq[m_f]{c} \rangle[\seq[j-1]{b},v]_i)
\sigma$ with $\seq[m_f]{c} \in \{ \bot, \top \}$.
Consider the intermediate subterm $v$.
If $v$ is not a $\bot$-pattern we reason as in the first case.
If $v$ is an instance of $b_j$ then rule \eqref{xi4} is applicable.
Otherwise, again using Lemma~\ref{AP lemma}, $v$ must be an instance of
an anti-pattern in $\m{AP}(\ell_j)$ and thus $t$
is reducible by rule \eqref{xi5}.
\qedhere
\end{itemize}
\end{proof}

\noindent The restriction to proper terms in Lemma~\ref{lem:bottomnormalforms}
is essential. For instance, $\m{even}(\m{0},\bot,\bot,\m{0})$ and
$\m{even}_2^1(\m{0},\bot,\m{true},\bot)$ are ground normal forms
(w.r.t.~Example~\ref{ex:Xi(even/odd)}) but not
$\bot$-patterns.

We have reached the point
where we can prove the main result,
for terminating proper terms. Since a term whose subterms contain
symbols $f_i^j$ has no parallel in the labeled setting, the proof
will require a fair bit of reshuffling; some steps must be postponed, while
other subterms must be eagerly evaluated. This is all done in Lemma
\ref{lem:completeness:terminating}.

In the following, $s \to^* t\:[N]$ or $s \xrightharpoonup{}^* t\:[N]$
indicates a reduction of cost $N$.

\begin{lem}
\label{lem:completeness:terminating}
Let $s \in \TTp(\HH)$ be a terminating term, $t$ a $\bot$-pattern, and
$\sigma\colon \Var(t) \to \TT(\HH)$. If
$s \to_{\Xi(\RR),\mu}^* t\sigma\:[N]$ then there exist a substitution
$\tau\colon \Var(t) \to \TTp(\HH)$ and numbers $K$ and $M$ with
$K + M \geqslant N$ such that
$\zeta^-(s) \xrightharpoonup{}^* \zeta^-(t\tau)\:[K]$ and
$t\tau \to_{\Xi(\RR),\mu}^* t\sigma\:[M]$.
\end{lem}

\begin{proof}
We use induction on $s$ with respect to
${>} := ({\to_{\Xi(\RR),\mu}} \cup {\rhd_\mu})^+$, which is a
well-founded order on terminating terms. (Here $s \rhd_\mu t$ if
$t$ is a subterm of $s$ occuring at an active position.)
We distinguish a number of cases. First of all, if $t$ is a variable
then we can simply take $\tau = \{ t \mapsto s \}$,
$K = 0$, and $M = N$.
Next suppose $s = f(\seq{s})$ with $f$ a constructor symbol.
We have $\zeta^-(s) = f(\zeta^-(s_1),\dots,\zeta^-(s_n))$
and $t = f(\seq{t})$ with $s_i \to_{\Xi(\RR),\mu}^* t_i\sigma\:[N_i]$
for all $1 \leqslant i \leqslant n$, such that $N = N_1 + \cdots N_n$.
Fix $i$. Since $s \rhd_\mu s_i$, we can apply the induction hypothesis,
resulting in a substitution
$\tau_i\colon \Var(t_i) \to \TTp(\HH)$ and numbers $K_i$ and $M_i$ with
$K_i + M_i \geqslant N_i$ such that
$\zeta^-(s_i) \xrightharpoonup{}^* \zeta^-(t_i\tau_i)\:[K_i]$ and
$t_i\tau_i \to_{\Xi(\RR),\mu}^* t_i\sigma\:[M_i]$.
Since $\bot$-patterns are linear by definition, the
substitution $\tau := \tau_1 \cup \cdots \cup \tau_n$ is well-defined.
Let $ K = K_1 + \dots + K_n$ and $M = M_1 + \dots + M_n$. We clearly have
$K + M \geqslant N$. Furthermore, $\zeta^-(s) \xrightharpoonup{}^*
f(\zeta^-(t_1\tau),\dots,\zeta^-(t_n\tau)) = \zeta^-(t\tau)$ with cost
$K$ and $t\tau \to_{\Xi(\RR),\mu}^* t\sigma\:[M]$.

The remaining case for $s$ is $s = f(\seq[n]{s},\seq[m_f]{c})$
with $f$ a defined symbol. Let
$R = \{ \smash{\rho^f_i} \mid c_i = \top \}$. We have
$\zeta^-(s) = f_R(\zeta^-(s_1),\dots,\zeta^-(s_n))$.
If there is no root step in the reduction $s \to_{\Xi(\RR),\mu}^* t\sigma$
then the result is obtained exactly as in the preceding case. So suppose
the reduction contains a root step. We prove the following claim
($\ast$):
\smallskip
\begin{quote}
There exist a term $u \in \TTp(\HH)$ different from $s$ and numbers
$A$ and $B$ with $A + B \geqslant N$ such that
$\zeta^-(s) \xrightharpoonup{}^+ \zeta^-(u)\:[A]$ and
$u \to_{\Xi(\RR),\mu}^* t\sigma\:[B]$.
\end{quote}
\smallskip
The statement of the lemma follows from ($\ast$), as can be seen as
follows. We have
$s = \zeta(\zeta^-(s)) \to_{\Xi(\RR),\mu}^* \zeta(\zeta^-(u)) = u$ by
Lemma~\ref{lem:inverseprops}(\ref{lem:inverseprops:inverseinner}) and
Theorem~\ref{thm:transformsound}. Since $s \neq u$ we must have $s > u$
and thus we can apply the induction hypothesis to
$u \to_{\Xi(\RR),\mu}^* t\sigma$. This yields a substitution
$\tau\colon \Var(t) \to \TTp(\HH)$ and numbers $K$ and $M$ with
$K + M \geqslant B$ such that
$\zeta^-(u) \xrightharpoonup{}^* \zeta^-(t\tau)\:[K]$ and
$t\tau \to_{\Xi(\RR),\mu}^* t\sigma\:[M]$. Hence
$\zeta^-(s) \xrightharpoonup{}^* \zeta^-(t\tau)\:[A+K]$
and $(A + K) + M = A + (K + M) \geqslant A + B \geqslant N$.

To prove the claim, we distinguish a few subcases depending on
which rule of $\Xi(\RR)$ is applied in the first root step.
\begin{enumerate}[label=\({\alph*}]
\item
Suppose the first root step uses a rule of type \eqref{xi1} and
let $\rho_i\colon \ell = f(\seq{\ell}) \to r$ be the originating rule in
$\RR{\restriction}f$. (So $c_i = \top$ and $i \in R$.)
The reduction from $s$ to $t\sigma$ has the shape
\[
s \xrightarrow{\smash{> \epsilon}}^*
f(\ell_1\gamma,\dots,\ell_n\gamma,\seq[m_f]{c}) \xrightarrow{\epsilon}
\xi_\top(r)\gamma \to^* t\sigma
\]
for some substitution $\gamma$ with $\dom(\gamma) \subseteq \Var(\ell)$.
Fix $1 \leqslant j \leqslant n$ and let
$C_j$ be the cost of $s_j \to^* \ell_j\gamma$. Let
$C = C_1 + \cdots + C_n$. From the induction
hypothesis we obtain a substitution
$\delta_j\colon \Var(\ell_j) \to \TTp(\HH)$ and numbers $K_j$ and
$M_j$ with $K_j + M_j \geqslant C_j$ such that
$\zeta^-(s_j) \xrightharpoonup{}^* \zeta^-(\ell_j\delta_j)\:[K_j]$ and
$\ell_j\delta_j \to^* \ell_j\gamma\:[M_j]$.
Because $f(\seq{\ell})$ is
linear, the substitution $\delta := \delta_1 \cup \dots \cup \delta_n$ is
well-defined.
With help of
Lemma~\ref{lem:inverseprops}(\ref{lem:inverseprops:substitution}) we
obtain $\zeta^-(s) \xrightharpoonup{}^*
f_R(\ell_1\delta_{\zeta^-},\dots,\ell_n\delta_{\zeta^-})\:[K]$.
As $\seq{\ell}$ are constructor terms, the reductions
$\ell_j\delta_j \to^* \ell_j\gamma\:[M_j]$ take place in the
substitution part. Hence
for every $x \in \Var(\ell)$ we have $x\delta \to^* x\gamma\:[M_x]$
such that
$M := M_1 + \cdots + M_n = \sum \{ M_x \mid x \in \Var(\ell) \}$
and $K + M \geqslant C$, where $K = K_1 + \dots + K_n$.

After these preliminaries, we proceed as follows. Let
$V = \Var(\ell) \setminus \Var(r)$. For every $x \in V$ we fix a
$\bot$-pattern $u_x$ such that $\gamma(x) \to^* u_x$. The existence of
$u_x$ is guaranteed by Lemma~\ref{lem:bottomnormalforms} and the
termination of $\gamma(x)$, which follows because
$s \to^* \cdot \rhd_\mu \gamma(x)$. Define the substitution
$\eta\colon \Var(\ell) \to \TTp(\HH)$ as follows:
\[
\eta(x) = \begin{cases}
u_x & \text{if $x \in V$} \\
\delta(x) & \text{if $x \notin V$}
\end{cases}
\]
We divide $M$ into $M_V = \sum \{ M_x \mid x \in V \}$ and
$M_{\bar{V}} = \sum \{ M_x \mid x \notin V \} =  M - M_V$.
We have $\ell\delta \to^* \ell\eta$. Applying the induction hypothesis 
to this reduction (with $t = \ell\delta$ and empty substitution $\sigma$)
yields $\zeta^-(\ell\delta) \xrightharpoonup{}^* \zeta^-(\ell\eta)\:[L]$
for some $L \geqslant M_V$.
Let $u = \xi_\top(r)\eta$. 
Lemma~\ref{lem:inverseprops} yields
$\zeta^-(u) = \m{label}(r)\eta_{\zeta^-}$.
Hence $\zeta^-(s) \xrightharpoonup{}^* \zeta^-(u)\:[A]$
with $A = K + L + 1$. We clearly have $s \neq u$.
In order to conclude ($\ast$), it remains to show that
$u \to^* t\sigma\:[B]$ for some $B \geqslant N - A$.
We have $u = \xi_\top(r)\delta$ due to the definitions of $V$ and $\eta$.
Hence $u \to^* \xi_\top(r)\gamma\:[D]$ for some $D \geqslant M_{\bar{V}}$
and thus $u \to^* t\sigma\:[B]$ with
$B := D + N - (C + 1) \geqslant M_{\bar{V}} + N - (C + 1) \geqslant
M_{\bar{V}} + N - (K + M + 1) = N - (K + M_V + 1)
\geqslant N - (K + L + 1) \geqslant N - A$.
\smallskip
\item
Suppose the first root step uses a rule of type \eqref{xi6} and let
$f(\seq{\ell})$ be the left-hand side of the rule in $\RR$
that gave rise to this rule. The reduction
from $s$ to $t\sigma$ has the following shape:\enlargethispage{\baselineskip}
\[
s \xrightarrow{\smash{> \epsilon}}^*
f(\seq{u},\seq[m_f]{c}) \xrightarrow{\epsilon}
f(\seq{u},\langle \seq[m_f]{c} \rangle[\bot]_i) \to^* t\sigma
\]
with $u_j$ an instance of an anti-pattern $v \in \m{AP}(\ell_j)$,
so $u_j = v\gamma$ for some substitution $\gamma$ and fixed $j$.
We have $s_i \to^* u_i$ for all $1 \leqslant i \leqslant n$.
By postponing the steps in arguments different from $j$, we obtain
\begin{alignat*}{2}
s &\xrightarrow{\smash{\mathrel{\geqslant} j}}^*
f(s_1,\dots,u_j,\dots,s_n,\seq[m_f]{c}) &\quad& [\,A\,] \\
&\xrightarrow{\phantom{\smash{>\epsilon}}}\phantom{^*}
f(s_1,\dots,u_j,\dots,s_n,\langle \seq[m_f]{c} \rangle[\bot]_i)
&& [\,0\,] \\
&\xrightarrow{\smash{> \epsilon}}\phantom{^*}
f(u_1,\dots,u_j,\dots,u_n,\langle \seq[m_f]{c} \rangle[\bot]_i)
\to^* t\sigma && [\,N-A\,]
\end{alignat*}
Since $s \rhd_\mu s_j \to^* v\gamma$, we can apply
the induction hypothesis to obtain
a substitution $\delta\colon \Var(v) \to \TTp(\HH)$ and numbers $K$ and
$M$ with $K + M \geqslant A$ such that
$\zeta^-(s_j) \xrightharpoonup{}^* \zeta^-(v\delta)\:[K]$ and
$v\delta \to^* v\gamma\:[M]$.
Lemma~\ref{lem:inverseprops}(\ref{lem:inverseprops:antipattern})
yields $\zeta^-(v\delta) = \zeta^-(v)\delta_{\zeta^-}$ and from
Lemma~\ref{AP lemma} we know that $\zeta^-(v)$ is a linear
labeled normal form which does not unify with $\ell_j$. Therefore
\begin{alignat*}{2}
\zeta^-(s)
\xrightharpoonup{}^* {} &
f_R(\zeta^-(s_1),\dots,\zeta^-(v)\delta_{\zeta^-},\dots,\zeta^-(s_n))
&& [\,K\,] \\[-.75ex]
\xrightharpoonup{\bot}\phantom{^*} {} & f_{R \setminus \{ \rho_i \}}
(\zeta^-(s_1),\dots,\zeta^-(v)\delta_{\zeta^-},\dots,\zeta^-(s_n))
&\quad& [\,0\,]
\end{alignat*}
The latter term equals $\zeta^-(u)$ where
$u = f(s_1,\dots,v\delta,\dots,s_n,\langle \seq[m_f]{c} \rangle[\bot]_i)$.
Furthermore,
\begin{alignat*}{2}
u &\to^* f(s_1,\dots,v\gamma,\dots,s_n,\langle\seq[m_f]{c}\rangle[\bot]_i) 
&\quad& [\,M\,] \\
&\to^* t\sigma && [\,N-A\,]
\end{alignat*}
Hence $\zeta^-(s) \xrightharpoonup{}^+ \zeta^-(u)\:[K]$ and
$u \to^* t\sigma\:[M+N-A]$ with
$M + N - A \geqslant M + N - (K + M) = N - K$. Since $s \neq u$,
this proves ($\ast$).
\smallskip
\item
In the remaining case, the first root step in reduction from
$s$ to $t\sigma$ uses a rule of type \eqref{xi2}. Let
$\rho = \rho_i\colon \ell = f(\seq{\ell}) \to r \Leftarrow c$
Since $t$ is a non-variable $\bot$-pattern, $t\sigma$ cannot
have some $f_j^i$ as root symbol. Hence
the application of \eqref{xi2} will be followed by
(possibly zero) root steps of type
\eqref{xi4}, for $j = 1, \dots, m-1$,
until either a step of type \eqref{xi3} with cost $Q = 1$
(when $m = k$)
or a step of type \eqref{xi5} with cost $Q = 0$ is used at the
root position. We have
\begin{alignat}{2}
[C] &&
\qquad s \xrightarrow{\smash{> \epsilon}}^* {} &
f(\seq{\ell},\langle \seq[m_f]{c} \rangle[\top]_i)\gamma
\notag \\
[0] &&
\xrightarrow{\smash{\phantom{>} \!\!\epsilon\,\,}}\phantom{^*} {} &
f_i^1(\seq{\ell},\langle \seq[m_f]{c} \rangle[\xi_\top(a_1)]_i)\gamma
\tag{\ref{xi2}} \\
[D_1] &&
\xrightarrow{\smash{> \epsilon}}^* {} &
f_i^1(\seq{\ell},\langle \seq[m_f]{c} \rangle[b_1]_i)\gamma
\notag \\
[0] &&
\xrightarrow{\smash{\phantom{>} \!\!\epsilon\,\,}}\phantom{^*} {} &
f_i^2(\seq{\ell},\langle \seq[m_f]{c} \rangle[b_1,\xi_\top(a_2)]_i)\gamma
\tag{\ref{xi4}} \\
&&
\xrightarrow{\smash{> \epsilon}}^* {} &
\cdots \notag{} \\
[0] &&
\xrightarrow{\smash{\phantom{>} \!\!\epsilon\,\,}}\phantom{^*} {} &
f_i^m(\seq{\ell},\langle \seq[m_f]{c} \rangle
[\seq[m-1]{b},\xi_\top(a_m)]_i)\gamma
\tag{\ref{xi4}} \\
[D_m] &&
\xrightarrow{\smash{> \epsilon}}^* {} &
f_i^m(\seq{\ell},\langle \seq[m_f]{c} \rangle[\seq[m-1]{b},v]_i)\gamma
\notag \\
[Q] &&
\xrightarrow{\smash{\phantom{>} \!\!\epsilon\,\,}}\phantom{^*} {} &
w \tag*{\eqref{xi3} or \eqref{xi5}} \\
[E] &&
\xrightarrow{\smash{\phantom{> \epsilon}}}^* {} &
t\sigma \notag
\end{alignat}\enlargethispage{\baselineskip}
for some substitution $\gamma$, $\bot$-pattern $v$, ground term $w$,
and numbers $C, \seq[m]{D}, E$ such that
$N = C + D_1 + \cdots + D_m + E + Q$.
(Here we use the fact that $b_j$ does not
share variables with $\seq{\ell},\seq[j-1]{b}$, for
$1 \leqslant j < m$. Moreover, $b_m$ as well as members of $\m{AP}(b_m)$
are $\bot$-patterns.)
Like in case \(a, we obtain a substitution
$\delta\colon \Var(\ell) \to \TTp(\HH)$ 
and numbers $K_j$ and $M_j$ such that
$\zeta^-(s_j) \xrightharpoonup{}^* \zeta^-(\ell_j\delta)\:[K_j]$ and
$\ell_j\delta \to^* \ell_j\gamma\:[M_j]$. Moreover,
$K + M \geqslant C$ where $K = K_1 + \cdots + K_n$ and
$M = M_1 + \cdots + M_n$.
We now distinguish two cases, depending on whether
\eqref{xi3} or \eqref{xi5} is used in the step to $w$.
\smallskip
\begin{itemize}
\item
Suppose the step to $w$ uses \eqref{xi3}.
In this case we have $Q = 1$, $v = b_m$ and $w = \xi_\top(r)\gamma$.
Let $V = \Var(\nulseq[m]{b}) \setminus \Var(\seq[m+1]{a})$.
(Recall that $b_0 = \ell$ and $a_{k+1} = r$).
For every $x \in V$ we fix a
$\bot$-pattern $u_x$ such that $\gamma(x) \to^* u_x$. The existence of
$u_x$ is guaranteed by Lemma~\ref{lem:bottomnormalforms} and the
termination of $\gamma(x)$, which follows from
$s \to^* \cdot \rhd_\mu \xi_\top(a_j)\gamma$
for all $1 \leqslant j \leqslant m$.
We inductively define substitutions $\nulseq[m]{\eta}$
with $\eta_j\colon \Var(\nulseq[j]{b}) \to \TTp(\HH)$
as well as numbers $\nulseq[m]{L}$ and $G_x$ for all
$x \in \Var(\nulseq[m]{b}) \setminus V$
such that
\smallskip
\begin{enumerate}[label=\({\alph*}]
\item
$\eta_j(x) \to^* \gamma(x)\:[G_x]$ for all
$0 \leqslant j \leqslant m$ and $x \in \Var(b_j) \setminus V$,
\smallskip
\item
$\zeta^-(\ell\delta) \xrightharpoonup{}^* \zeta^-(\ell\eta_0)\:[L_0]$
with $L_0 \geqslant M - \sum \{ G_x \mid x \in \Var(b_0) \setminus V \}$,
and
\smallskip
\item
$\zeta^-(\xi_\top(a_j)\eta_{j-1}) \xrightharpoonup{}^* \zeta^-(b_j\eta_j)
\:[L_j]$ with
$L_j \geqslant D_j + \sum \{ G_x \mid x \in \Var(a_j) \} -
\sum \{ G_x \mid x \in \Var(b_j) \setminus V \}$
for all $0 < j \leqslant m$.
\end{enumerate}\bigskip

\begin{itemize}
\item
Let $j = 0$. We define
\[
\eta_0(x) = \begin{cases}
u_x &\text{if $x \in \Var(b_0) \cap V$} \\
\delta(x) &\text{if $x \in \Var(b_0) \setminus V$}
\end{cases}
\]
We obtain $\eta_0(x) \to^* \gamma(x)$
for all $x \in \Var(b_0) \setminus V$ from
$\ell\delta \to^* \ell\gamma$, and define $G_x$ as the cost of this
reduction. This establishes property \(a.
Applying the induction hypothesis to the reduction
$\ell\delta \to^* \ell\eta_0$ (with $t = \ell\eta_0$ and $\sigma$ the
empty substitution) yields
$\zeta^-(\ell\delta) \xrightharpoonup{}^* \zeta^-(\ell\eta_0)\:[L_0]$
for some $L_0 \geqslant \sum \{ \m{cost}(\delta(x) \to^* \eta_0(x)) \mid
x \in \Var(b_0) \cap V \}$.
Note that
$L_0 + \sum \{ G_x \mid x \in \Var(b_0) \setminus V \} \geqslant M$.
Hence property \(b\ holds.
Property \(c\ holds vacuously.
\smallskip
\item
Consider $0 < j \leqslant m$.
Since $\Var(a_j) \subset \Var(\nulseq[j-1]{b}) \setminus V$
we obtain $\xi_\top(a_j)\eta_{j-1} \to^* \xi_\top(a_j)\gamma\:[G_j]$
for some $G_j \geqslant \sum \{ G_x \mid x \in \Var(a_j) \}$.
(Equality need not hold if $a_j$ is a non-linear term.)
We apply the induction hypothesis to
$\xi_\top(a_j)\eta_{j-1} \to^* \xi_\top(a_j)\gamma \to^*
b_j\gamma\:[G_j + D_j]$, yielding a substitution
$\delta_j\colon \Var(b_j) \to \TTp(\HH)$ 
and numbers $L'$ and $N'$ with $L' + N' \geqslant G_j + D_j$
such that $\zeta^-(\xi_\top(a_j)\eta_{j-1}) \xrightharpoonup{}^*
\zeta^-(b_j\delta_j)\:[L']$ and
$b_j\delta_j \to^* b_j\gamma\:[N']$.
We divide $N'$ into $X + Y$ where
\begin{align*}
\qquad
X &= \sum \{ \m{cost}(\delta_j(x) \to^* \gamma(x)) \mid
x \in \Var(b_j) \cap V \} \\
Y &= \sum \{ \m{cost}(\delta_j(x) \to^* \gamma(x)) \mid
x \in \Var(b_j) \setminus V \}
\end{align*}
and define the substitution $\eta_j$ as follows:
\[
\eta_j(x) = \begin{cases}
\eta_{j-1}(x) &\text{if $x \in \Var(\nulseq[j-1]{b})$} \\
u_x &\text{if $x \in \Var(b_j) \cap V$} \\
\delta_j(x) &\text{if $x \in \Var(b_j) \setminus V$}
\end{cases}
\]
Since $b_j$ is a constructor term, from
$b_j\delta_j \to^* b_j\gamma$ we infer $\eta_j(x) \to^* \gamma(x)$
for all $x \in \Var(b_j) \setminus V$, at a cost we can safely define
as $G_x$. Hence property \(a\ holds. Property \(b\ holds vacuously.
Note that $Y = \sum \{ G_x \mid x \in \Var(b_j) \setminus V \}$.
Applying the induction hypothesis to
$b_j\delta_j \to^* b_j\eta_j$
(with $t = b_j\eta_j$ and $\sigma$ the empty substitution)
yields
$\zeta^-(b_j\delta_j) \xrightharpoonup{}^* \zeta^-(b_j\eta_j)\:[Z]$
for some number
\[
\qquad\qquad
Z \geqslant \sum \{ \m{cost}(\delta_j(x) \to^* \gamma(x) \to^* u_x) \mid
x \in \Var(b_j) \cap V \} \geqslant X
\]
Let $L_j = L' + Z$.
So $\zeta^-(\xi_\top(a_j)\eta_{j-1}) \xrightharpoonup{}^*
\zeta^-(b_j\eta_j)\:[L_j]$.
We have
\begin{align*}
\qquad\qquad\qquad
L_j &\geqslant L' + X = L' + N' - Y \geqslant G_j + D_j - Y \\
&\geqslant D_j + \sum \{ G_x \mid x \in \Var(a_j) \} -
\sum \{ G_x \mid x \in \Var(b_j) \setminus V \}
\end{align*}
establishing property \(c.
\smallskip
\end{itemize}
Let $\eta = \eta_m$.
Since $\eta$ coincides with $\eta_j$ on
$\Var(\nulseq[j]{b})$ for all $0 \leqslant j \leqslant m$, we obtain
\[
\m{label}(a_j)\eta_{\zeta^-} = \zeta^-(\xi_\top(a_j)\eta)
\xrightharpoonup{}^* \zeta^-(b_j\eta) = b_j\eta_{\zeta^-}\:[L_j]
\]
for $1 \leqslant j \leqslant m$. Hence
\[
\zeta^-(s) \xrightharpoonup{}^*
f_R(\ell_1\eta_{\zeta^-},\dots,\ell_n\eta_{\zeta^-})
\xrightharpoonup{}
\m{label}(r)\eta_{\zeta^-}\:[A]
\]
with $A = (K + L_0) + L_1 + \cdots + L_m + 1$. Let $u = \xi_\top(r)\eta$.
Lemma~\ref{lem:inverseprops} yields
$\zeta^-(u) = \m{label}(r)\eta_{\zeta^-}$. To establish the claim
($\ast$), it remains to show $u \to^* t\sigma\:[B]$ for some
$B$ such that $A+B \geqslant N$.
Because $\Var(r) \subseteq \Var(\nulseq[m]{b}) \setminus V$,
we obtain
\[
u = \xi_\top(r)\eta \to^* \xi_\top(r)\gamma = w \to^* t\sigma\:[B]
\]
with $B \geqslant \sum \{ G_x \mid x \in \Var(r) \} + E$. We have
\begin{align*}
\qquad\qquad
A + B \geqslant {}&
K + L_0 + L_1 + \cdots + L_m + 
\sum \{ G_x \mid x \in \Var(r) \} + E + 1 \\
\geqslant {}& (C - M) +
(M - \sum \{ G_x \mid x \in \Var(b_0) \setminus V \})
+ D_1 + \cdots + D_m
\\
&{} + \sum \{ G_x \mid x \in \Var(\seq[m+1]{a}) \} \\
&{} - \sum \{ G_x \mid x \in \Var(\seq[m]{b}) \setminus V \} + E + 1 \\
\geqslant {}& C + D_1 + \cdots + D_m +
\sum \{ G_x \mid x \in \Var(\seq[m+1]{a}) \} \\
&{} - \sum \{ G_x \mid x \in \Var(\nulseq[m]{b}) \setminus V \}
+ E + 1 \\
\geqslant {}& C + D_1 + \cdots + D_m + E + 1 = N
\end{align*}\enlargethispage{\baselineskip}%
where the last inequality follows from 
$(\Var(\nulseq[m]{b}) \setminus V) \subseteq \Var(\seq[m+1]{a})$.
\item
Suppose the step to $w$ uses \eqref{xi5}.
In this case we have $Q = 0$, $v \in \m{AP}(b_m)$ and
$w = f(\seq{\ell},\langle \seq[m_f]{c} \rangle[\bot]_i)\gamma$.
Let $V = \Var(\seq[m-1]{b},v) \setminus \Var(\seq[m]{a})$.
For every $x \in V$ we fix a
$\bot$-pattern $u_x$ such that $\gamma(x) \to^* u_x$. The existence of
$u_x$ is guaranteed by Lemma~\ref{lem:bottomnormalforms} and the
termination of $\gamma(x)$, which follows from
$s \to^* \cdot \rhd_\mu \xi_\top(a_j)\gamma$
for all $1 \leqslant j \leqslant m$.
We inductively define substitutions $\nulseq[m]{\eta}$
with $\eta_j\colon \Var(\nulseq[j]{b}) \to \TTp(\HH)$
for $1 \leqslant j < m$ and
$\eta_m\colon \Var(v) \to \TTp(\HH)$
as well as numbers $\seq[m]{L}$ and $G_x$ for all
$x \in \Var(\nulseq[m-1]{b}) \setminus V$
such that
\smallskip
\begin{enumerate}[label=\({\alph*}]
\item
$\eta_j(x) \to^* \gamma(x)\:[G_x]$ for all
$0 \leqslant j \leqslant m$ and $x \in \Var(b_j) \setminus V$,
\smallskip
\item
$\zeta^-(\xi_\top(a_j)\eta_{j-1}) \xrightharpoonup{}^* \zeta^-(b_j\eta_j)
\:[L_j]$ with
$L_j \geqslant D_j + \sum \{ G_x \mid x \in \Var(a_j) \} -
\sum \{ G_x \mid x \in \Var(b_j) \setminus V \}$
for all $0 < j < m$.
\smallskip
\item
$\zeta^-(\xi_\top(a_m)\eta_{m-1}) \xrightharpoonup{}^*
\zeta^-(v\eta_m) \:[L_m]$ with
$L_m \geqslant \sum \{ G_x \mid x \in \Var(a_m) \} + D_m$.
\smallskip
\end{enumerate}\bigskip

\begin{itemize}
\item
We define $\eta_0 = \delta$. We obtain $\eta_0(x) \to^* \gamma(x)$
for all $x \in \Var(\ell) = \Var(b_0) \setminus V$ from
$\ell\delta \to^* \ell\gamma$, and define $G_x$ as the cost of this
reduction. This establishes property \(a.
Note that $\sum \{ G_x \mid x \in \Var(b_0) \} = M$.
\smallskip
\item
The case $0 < j < m$ is exactly the same as for \eqref{xi3},
establishing properties \(a\ and \(b.
\item
\smallskip
For $j = m$ we have
$\xi_\top(a_m)\eta_{m-1} \to^* \xi_\top(a_m)\gamma \to^* v\gamma$.
Let $G_m$ be the cost of
$\xi_\top(a_m)\eta_{m-1} \to^* \xi_\top(a_m)\gamma$, so
$G_m \geqslant \sum \{ G_x \mid x \in \Var(a_m) \}$.
The induction hypothesis yields a substitution
$\delta_m\colon \Var(v) \to \TTp(\HH)$ and 
numbers $L'$ and $N'$ with $L' + N' \geqslant G_m + D_m$
such that
$\zeta^-(\xi_\top(a_m)\eta_{m-1}) \xrightharpoonup{}^*
\zeta^-(v\delta_m)\:[L']$ and $v\delta_m \to^* v\gamma\:[N']$.
We define the substitution $\eta_m$ as follows:
\[
\eta_m(x) = \begin{cases}
\eta_{m-1}(x) &\text{if $x \in \Var(\nulseq[m-1]{b})$} \\
u_x &\text{if $x \in \Var(v)$}
\end{cases}
\]
Applying the induction hypothesis to $v\delta_m \to^* v\eta_m$
(with $t = v\eta_m$ and $\sigma$ the
empty substitution) yields
$\zeta^-(v\delta_m) \xrightharpoonup{}^* \zeta^-(v\eta_m)\:[Z]$
for some number $Z \geqslant N'$. Let $L_m = L' + Z$.
Thus, $\zeta^-(\xi_\top(a_{m})\eta_m)
\xrightharpoonup{}^* \zeta^-(v\eta_m)\:[L_m]$.
We have $L_m \geqslant L' + N' \geqslant
G_m + D_m \geqslant \sum \{ G_x \mid x \in \Var(a_m) \} + D_m$.
Hence property \(c\ holds.
\end{itemize}
\smallskip
Let $\eta = \eta_m$. Since $\eta$ coincides with $\eta_j$ on
$\Var(\nulseq[j]{b})$ for all $0 \leqslant j < m$, we obtain

\begin{itemize}
\item
$\zeta^-(s) = f_R(\zeta^-(s_1),\dots,\zeta^-(s_n))
\xrightharpoonup{}^* f_R(\seq{\ell})\eta_{\zeta^-}\:[K]$
\smallskip
\item
$\m{label}(a_j)\eta_{\zeta^-} = \zeta^-(\xi_\top(a_j)\eta)
\xrightharpoonup{}^* \zeta^-(b_j\eta) = \zeta^-(b_j)\eta_{\zeta^-}
\:[L_j]$
for $1 \leqslant j < m$
\smallskip
\item
$\m{label}(a_m)\eta_{\zeta^-} = \zeta^-(\xi_\top(a_m)\eta)
\xrightharpoonup{}^* \zeta^-(v\eta) = \zeta^-(v)\eta_{\zeta^-}\:[L_m]$,
with $\zeta^-(v)$ a $\bot$-pattern that does not unify with $v$
according to Lemma~\ref{AP lemma}.
\end{itemize}
\smallskip
Let $u = f(\seq{\ell},\langle \seq[m_f]{c} \rangle[\bot]_i)\eta$.
We have
\[
\zeta^-(s) \xrightharpoonup{}
\zeta^-(f_{R \setminus \{ \rho_i \}}(\zeta^-(s_1),\dots,\zeta^-(s_n)))
= \zeta^-(u)\:[K + L]
\]
for $L = L_1 + \dots + L_m$.
Furthermore, $u \to^* w \to^* t\sigma\:[M + E]$.
It remains to show that $K + L + M + E \geqslant N$.
Since $K + M \geqslant C$, this amounts to showing
$L \geqslant D_1 + \cdots + D_m$.
We have
\begin{align*}
\qquad\qquad\qquad
L &\geqslant 
\sum_{j = 1}^{m-1} 
\biggl(
D_j + \sum \{ G_x \mid x \in \Var(a_j) \} -
\sum \{ G_x \mid x \in \Var(b_j) \setminus V \} \biggr) + L_m \\
&\geqslant \sum_{j = 1}^{m} D_m +
\sum \{ G_x \mid x \in \Var(\seq[m]{a}) \} \\[-2ex]
&\phantom{{} \geqslant \sum_{j = 1}^{m} D_m}
- \sum \{ G_x \mid x \in \Var(\seq[m-1]{b}) \setminus V \} \\[-2ex]
&\geqslant 
\sum_{j = 1}^{m} D_m
\end{align*}
where the last inequality follows from 
$\Var(\seq[m-1]{b}) \setminus V \subseteq \Var(\seq[m]{a})$.
Since $s \neq u$, we established ($\ast$).
\qedhere
\end{itemize}
\end{enumerate}
\end{proof}

\noindent Thus, we proved the main part of Theorem~\ref{thm:transformcomplete}
for \emph{terminating} terms. For non-terminating terms, we can use this
result, as we will see in the proof of
Lemma~\ref{lem:nontermtranslate:main}.
The following lemma handles the main step.

\begin{lem}
\label{lem:nontermtranslate:main}
For every minimal non-terminating term $s \in \TTp(\HH)$
there exists a non-terminating term $t \in \TTp(\HH)$ such that
$\zeta^-(s) \rightharpoonup^+ \zeta^-(t)$ or
$\zeta^-(s) \rightharpoonup^* \cdot \rhrhd \zeta^-(t)$.
\end{lem}

Here a minimal non-terminating term is a non-terminating term
with the property that every proper subterm at an active position is
terminating.

\begin{proof}
We must have $s = f(\seq{s},\seq[m_f]{c})$ for some defined
function symbol $f$. Let $R = \{ \smash{\rho^f_i} \mid c_i = \top \}$. We
have $\zeta^-(s) = f_R(\zeta^-(s_1),\dots,\zeta^-(s_n))$.
Since the terms $\seq{s}$ are terminating by minimality, any
infinite reduction starting at $s$ must contain a root step. So
\[
s \xrightarrow{\smash{> \epsilon}}^* u\gamma \xrightarrow{\epsilon}
v\gamma
\]
for some rule $u \to v$ of $\Xi(\RR)$ and substitution $\gamma$
such that $v\gamma$ is non-terminating. Inspecting the applicable
rules in $\Xi(\RR)$, it follows that $u$ is a linear basic term of the
form
$u = f(\seq{u},\langle \seq[m_f]{y} \rangle[\top]_i)$.
Let $\delta$ be the restriction of $\gamma$ to $\{ \seq[m_f]{y} \}$
We have $\delta(y_j) = c_j$ for all $1 \leqslant j \leqslant m_f$.
Let $u' = u\delta$ and $v' = v\delta$.
Clearly $u'\gamma = u\gamma$ and $v'\gamma = v\gamma$, while $u'$ is a
proper linear term. Because the terms $\seq{s}$ are terminating by
minimality, Lemma~\ref{lem:completeness:terminating} provides
substitutions $\seq{\tau}$ with $\tau_j\colon \Var(u_i) \to \TTp(\HH)$
such that $\zeta^-(s_j) \xrightarrow{}^* \zeta^-(u_j\tau_j)$ and
$u_j\tau_j \to^* u_j\gamma$. Since $u$ is linear, the substitution
$\tau = \tau_1 \cup \dots \cup \tau_n$ is well-defined.
We obtain
\[
\zeta^-(s) \xrightharpoonup{}^* \zeta^-(u'\tau) =
\zeta^-(u')\tau_{\zeta^-} =
f_R(\zeta^-(u_1),\dots,\zeta^-(u_n))\tau_{\zeta^-}
\]
with $\tau(x) \to^* \gamma(x)$ for all $x \in \Var(u')$.
We now distinguish three cases, depending on the nature of the rule
$u \to v$. Let $\rho_i\colon f(\seq{\ell}) \to r$ be the rule in $\RR$
that give rise to $u \to v$.
\begin{enumerate}
\item
Suppose $u \to v$ is a rule of type \eqref{xi6}.
There exists $1 \leqslant j \leqslant n$ such that
$u_j \in \m{AP}(\ell_j)$. We have
$v = f(\seq{u},\langle \seq[m_f]{x} \rangle[\bot]_i)$.
According to
Lemma~\ref{lem:inverseprops}(\ref{lem:inverseprops:antipattern})
$\zeta^-(u_j)$ is a linear labeled normal form which does not unify with
$\ell_j$. Hence
\[
\zeta^-(u'\tau) \xrightharpoonup{\bot\,}
f_{R \setminus \{ \rho_i \}}(\zeta^-(u_1\tau),\dots,\zeta^-(u_n\tau)) =
\zeta^-(v'\tau)
\]
Since all variables in $v'$ are at active positions,
we have $v'\tau \to^* v'\gamma = v\gamma$. It follows that $v'\tau$ is
non-terminating and thus we can take $v'\tau$ for $t$ to satisfy the first
possibility of the statement of the lemma.
\smallskip
\item
Suppose $u \to v$ is a rule of type \eqref{xi1}.
So $u_j = \ell_j$ for all $1 \leqslant j \leqslant n$ and
$v' = \xi_\top(r)$.
Using Lemma~\ref{lem:inverseprops} we obtain
$\zeta^-(u_j\tau) = u_j\tau_{\zeta^-}$ for
$1 \leqslant j \leqslant n$ as well as
$\zeta^-(v'\tau) = \m{label}(r)\tau_{\zeta^-}$.
Hence $\zeta^-(u'\tau) = f_R(u_1\tau_{\zeta^-},\dots,u_n\tau_{\zeta^-})
\xrightharpoonup{} \zeta^-(v'\tau)$
and we conclude as in the preceding case.
\smallskip
\item
Suppose $u \to v$ is a rule of type \eqref{xi2}.
So $u_j = \ell_j$ for all $1 \leqslant j \leqslant n$ and
$v' = f_i^1(\seq{\ell},\langle \seq[m_f]{c} \rangle[\xi_\top(a_1)]_i)$.
We have $\zeta^-(s) \xrightharpoonup{}^* f_R(\seq{\ell})\tau_{\zeta^-}$.
We will define a number $1 \leqslant m \leqslant k$,
substitutions $\tau_1, \gamma_1, \dots, \tau_m, \gamma_m$, and terms
$\seq[m]{r}$ such that
\begin{enumerate}[label=\({\alph*}]
\item
$\tau_j\colon \Var(\nulseq[j-1]{b}) \to \TTp(\HH)$,
\item
$r_j = f_i^j(\seq{\ell},\langle \seq[m_f]{c} \rangle
[\seq[j-1]{b},\xi_\top(a_j)]_i)$,
\item
$\m{label}(a_l)(\tau_j)_{\zeta^-} \rightharpoonup^* b_l(\tau_j)_{\zeta^-}$
for all $1 \leqslant l < j$,
\item
$\tau_j(x) \to^* \gamma_j(x)$ for all $x \in \Var(\nulseq[j-1]{b})$,
\item
$r_j\gamma_j$ is non-terminating, and
\item
$\zeta^-(s) \rightharpoonup^* f_R(\seq{\ell})(\tau_j)_{\zeta^-}$
\smallskip
\end{enumerate}
for all $1 \leqslant j \leqslant m$.
By defining $\tau_1 = \tau$, $\gamma_1 = \gamma$, and $r_1 = v'$, the
above properties are clearly satisfied for $j = 1$.
Consider $\xi_\top(a_j)\tau_j$, which is a ground proper term by
Lemma~\ref{lem:inverseprops}(\ref{lem:inverseprops:top}).
If $\xi_\top(a_j)\tau_j$ is non-terminating then we let $m = j$
and define
$t = \xi_\top(a_j)\tau_j$. In this case we have
$\zeta^-(t) = \m{label}(a_j)(\tau_j)_{\zeta^-}$ by the same lemma 
and thus $f_R(\seq{\ell})(\tau_j)_{\zeta^-} \rhrhd \zeta^-(t)$
by property \(c, establishing the second possibility of the statement of
the lemma.

So assume that $\xi_\top(a_j)\tau_j$ is terminating. We have
$\xi_\top(a_j)\tau_j \to^* \xi_\top(a_j)\gamma_j$, so the latter term
is terminating as well. Since $\xi_\top(a_j)\gamma_j$ is the only
active argument in $r_j\gamma_j$, the infinite reduction starting
from the latter term must contain a root step. So
$\smash{r_j\gamma_j \xrightarrow{>\epsilon\,} \ell'\gamma_{j+1}
\xrightarrow{\epsilon} r'\gamma_{j+1}}$ for some rule
$\ell' \to r' \in \Xi(\RR)$ and substitution
$\gamma_{j+1}$ with $\dom(\gamma_{j+1}) = \Var(\ell')$ such that
$r'\gamma_{j+1}$ is non-terminating.
Since $\m{root}(r_j\gamma_j) = f_i^j$, 
$\ell' = f_i^j(\seq{\ell},\langle \seq[m_f]{x} \rangle [\seq[j-1]{b},w]_i)$
for some $\bot$-pattern $w$ ($w = b_j$ when $\ell' \to r'$ is a rule
of type \eqref{xi3} or \eqref{xi4} and $w \in \m{AP}(b_j)$ when
$\ell' \to r'$ is a rule of type \eqref{xi5}) which has no
variables in common with $\seq{\ell},\seq[j-1]{b}$. We have
$\xi_\top(a_j)\tau_j \to^* \xi_\top(a_j)\gamma_j \to^* w\gamma_{j+1}$.
From Lemma~\ref{lem:completeness:terminating} we obtain a
substitution $\tau\colon \Var(w) \to \TTp(\HH)$ such that
$\zeta^-(\xi_\top(a_j)\tau_j) \rightharpoonup^* \zeta^-(w\tau)$
and $w\tau \to^* w\gamma_{j+1}$.
Let $\tau_{j+1} = \tau_j \cup \tau$.
We have
$\tau_{j+1}\colon \Var(\nulseq[j-1]{b},w) \to \TTp(\HH)$
as well as
$\zeta^-(\xi_\top(a_j)\tau_j) = \zeta^-(\xi_\top(a_j)\tau_{j+1})
= \m{label}(a_j)(\tau_{j+1})_{\zeta^-}$ by
Lemma~\ref{lem:inverseprops}(\ref{lem:inverseprops:top}).
Furthermore,
$\tau_{j+1}(x) \to^* \gamma_{j+1}(x)$ for all
$x \in \Var(\nulseq[j-1]{b},w)$.
We distinguish three subcases, depending of the
type of the rule $\ell' \to r'$. In the first and third case,
we obtain the statement of the lemma. In the second case, we
establish the properties \(a--\(f\ for $j+1$. Since rules of
type \eqref{xi4} can be used only finitely many times, this
concludes the proof.
\smallskip
\begin{itemize}[label=($1_\rho$)]
\item[\eqref{xi3}]
In this case we have $j = k$, $w = b_j$, and $r' = \xi_\top(r)$.
So $\m{label}(a_l)(\tau_{j+1})_{\zeta^-} \rightharpoonup^*
b_l(\tau_{j+1})_{\zeta^-}$ for all $1 \leqslant l \leqslant k$.
Since all variables in $\xi_\top(r)$ occur at active positions,
$r'\tau_{j+1} \to^* r'\gamma_{j+1}$ and thus $r'\tau_{j+1}$
is non-terminating. According to
Lemma~\ref{lem:inverseprops}(\ref{lem:inverseprops:top})
$r'\tau_{j+1}$ is proper and
$\zeta^-(r'\tau_{j+1}) = \m{label}(r)(\tau_{j+1})_{\zeta^-}$.
So we choose $t = r'\tau_{j+1}$ to obtain a successful
reduction step
$f_R(\seq{\ell})(\tau_{j+1})_{\zeta^-} \rightharpoonup \zeta^-(t)$. Hence
$\zeta^-(s) \rightharpoonup^+ \zeta^-(t)$ and thus
the first possibility of the statement of the lemma holds.
\item[\eqref{xi4}]
In this case, $w = b_j$ and $r' = f_i^{j+1}(\seq{\ell},
\langle \seq[m_f]{y} \rangle [\seq[j]{b},\xi_\top(a_{j+1})]_i)$ with
$j < k$. We have $\m{label}(a_l)(\tau_{j+1})_{\zeta^-} \rightharpoonup^*
b_l(\tau_{j+1})_{\zeta^-}$ for all $1 \leqslant l < j+1$.
Let $r_{j+1} = r'\delta$. One easily checks that the properties
\(a--\(f\ are satisfied for $j+1$.
\item[\eqref{xi5}]
In this case, $w \in \m{AP}(b_j)$
and $r' = f(\seq{\ell},\langle \seq[m_f]{y} \rangle[\bot]_i)$.
According to
Lemma~\ref{lem:inverseprops},
$\zeta^-(w)$ is a $\bot$-pattern which does not unify with
$b_j$ and $\zeta^-(w)(\tau_{j+1})_{\zeta^-} = \zeta^-(w\tau_{j+1})$.
Since $\m{label}(a_j)(\tau_{j+1})_{\zeta^-} \rightharpoonup^*
\zeta^-(w)(\tau_{j+1})_{\zeta^-}$ and
$\m{label}(a_l)(\tau_{j+1})_{\zeta^-} \rightharpoonup^*
b_l(\tau_{j+1})_{\zeta^-}$ for all $1 \leqslant l < j$,
the conditions for a failing step are satisfied and thus
$f_R(\seq{\ell})(\tau_{j+1})_{\zeta^-} \rightharpoonup
f_{R \setminus \{ \rho_i \}}(\seq{\ell})(\tau_{j+1})_{\zeta^-}$.
The term $t = r'\delta\tau_{j+1}$ is proper and since all variables in $r'\delta$ occur
at active positions, $t \to^* r'\delta\gamma_{j+1}$ and thus
$t$ is non-terminating. Since
$\zeta^-(t) = \zeta^-(r'\delta)(\tau_{j+1})_{\zeta^-} =
f_{R \setminus \{ \rho_i \}}(\seq{\ell})(\tau_{j+1})_{\zeta^-}$,
we obtain $\zeta^-(s) \rightharpoonup^+ \zeta^-(t)$ to satisfy
the first possibility of the statement of the lemma.
\qedhere
\end{itemize}
\end{enumerate}
\end{proof}

\begin{lem}
\label{lem:nontermtranslate}
If $s \in \TTp(\HH)$ is non-terminating then
$\zeta^-(s) \xrightharpoonup{\infty\,}$.
\end{lem}

\begin{proof}
We construct an infinite sequence of non-terminating proper ground
terms $s_0$, $s_1$, $s_2$, $\dots$
with $s_0 = s$ such that
$\zeta^-(s_i) \mathrel{(\rightharpoonup \cup \rhrhd)}^+ \zeta^-(s_{i+1})$
for all $i \geqslant 0$.
Suppose $s_j$ has been defined. Since $s_j$ is non-terminating, it
contains a minimal non-terminating subterm $u$, say at position $p$
$\in \Pos_\mu(s_j)$.
According to Lemma~\ref{lem:nontermtranslate:main} there
exists a non-terminating term $v \in \TTp(\HH)$ such that
$\zeta^-(u) \rightharpoonup^+ \zeta^-(v)$ or
$\zeta^-(u) \rightharpoonup^* \cdot \rhrhd \zeta^-(v)$.
We distinguish three cases.
\begin{itemize}
\item
If $\zeta^-(u) \rightharpoonup^+ \zeta^-(v)$ then
$\zeta^-(s_i) = \zeta^-(s_i[u]_p) = \zeta^-(s_i)[\zeta^-(u)]_p$ by
Lemma~\ref{lem:inverseprops}(\ref{lem:inverseprops:substitution})
and thus $\zeta^-(s_i)
\rightharpoonup^+ \zeta^-(s_i)[\zeta^-(v)]_p = \zeta^-(s_i[v]_p)$.
Note that $s_i[v]_p$ is non-terminating. Hence we can take
$s_{i+1} = s_i[v]_p$.
\item
Suppose $\zeta^-(u) \rhrhd \zeta^-(v)$. We have
$\zeta^-(s_i) = \zeta^-(s_i[u]_p) = \zeta^-(s_i)[\zeta^-(u)]_p$ and thus
$\zeta^-(s_i) \rhrhd \zeta^-(v)$ by the definition of $\rhrhd$.
Hence we define $s_{i+1} = v$.
\item
Suppose $\zeta^-(u) \rightharpoonup^+ w \rhrhd \zeta^-(v)$.
We have $\zeta^-(s_i) \rightharpoonup^+ \zeta^-(s_i)[w]_p
\rhrhd \zeta^-(v)$ and hence also in this case we take
$s_{i+1} = v$.
\qedhere
\end{itemize}
\end{proof}

\begin{proof}[Proof of Theorem~\ref{thm:transformcomplete}]
Let $\RR$ be a strong CCTRS and $s \in \TT(\GG)$. We have
$\zeta(s) \in \TTp(\HH)$ by
Lemma~\ref{lem:inverseprops}(\ref{lem:inverseprops:inverseouter}).
First suppose that $\zeta(s)$ is terminating and there
exists a context-sensitive reduction
$\zeta(s) \to_{\Xi(\RR),\mu}^* t\:[N]$.
Let $u$ be a normal form of $t$.
Obviously, $\zeta(s) \to_{\Xi(\RR),\mu}^* u\:[M]$ for some
$M \geqslant N$. According to
Lemma~\ref{lem:bottomnormalforms} the term $u$ is a $\bot$-pattern.
Lemma~\ref{lem:completeness:terminating} yields
$s = \zeta^-(\zeta(s)) \rightharpoonup^* \zeta^-(u)\:[K]$ with
$K \geqslant M$.
Next suppose the existence of
an infinite $(\Xi(\RR),\mu)$ reduction starting from
$\zeta(s)$. In this case
$s = \zeta^-(\zeta(s)) \xrightharpoonup{\infty\,}$ by
Lemma~\ref{lem:nontermtranslate}.
\end{proof}

\end{document}